\documentclass[journal,onecolumn]{IEEEtran}
\usepackage{amsmath,amsfonts,amssymb,amsthm}
\usepackage{array}
\usepackage{booktabs}
\usepackage{graphicx}
\usepackage{url}
\usepackage{cite}
\usepackage{braket}
\usepackage{tabularx}
\usepackage{ragged2e}

\makeatletter
\@ifundefined{llbracket}{%
  \DeclareRobustCommand{\llbracket}{\lbrack\!\lbrack}%
  \DeclareRobustCommand{\rrbracket}{\rbrack\!\rbrack}%
}{}
\makeatother

\newtheorem{theorem}{Theorem}[section]
\newtheorem{lemma}[theorem]{Lemma}
\newtheorem{proposition}[theorem]{Proposition}
\newtheorem{corollary}[theorem]{Corollary}
\newtheorem{definition}[theorem]{Definition}
\newtheorem{remark}[theorem]{Remark}
\newtheorem{example}[theorem]{Example}

\DeclareMathOperator{\Tr}{Tr}
\DeclareMathOperator{\rank}{rank}
\DeclareMathOperator{\wt}{wt}
\DeclareMathOperator{\Imag}{Im}
\DeclareMathOperator{\Fix}{Fix}
\newcommand{\F}{\mathbb{F}}
\newcommand{\Fqm}{\F_{q^m}}
\newcommand{\ptr}{\perp_{\mathrm{Tr}}}

\newcommand{\Kker}{\mathcal{K}}

\newcommand{\pstd}{\perp_{\mathrm{std}}}
\newcommand{\cD}{\mathcal{D}}

\begin{document}

\title{Singleton-Optimal Rank-Metric CSS Codes:\\
Equality Structure and Exact Projected-Recovery Radii}

\author{Myeongjun~Kim, Suseong Lee, Jaeho Jeon, and Young-Sik~Kim%
\thanks{This work was partly supported by Institute of Information \&
Communications Technology Planning \& Evaluation (IITP) grant funded by the Korea government (MSIT) (RS-2024-00399401, Development of Quantum-Safe Infrastructure Migration and Quantum Security Verification Technologies, 50\%) and Institute of Information \& Communications Technology Planning \& Evaluation (IITP) grant funded by the Korea government (MSIT) (RS-2026-25530181, Development of Inter-Domain Quantum Security System Combining PQC and QKD, 50\%).}
\thanks{The authors are with DGIST, Daegu, Republic of Korea.
M. Kim is with the Department of the Interdisciplinary Studies of
Artificial Intelligence.
S. Lee and J. Jeon are with the Department of Electrical Engineering
and Computer Sciences.
Y.-S. Kim is with the Department of the Interdisciplinary Studies of
Artificial Intelligence and the Department of Electrical Engineering
and Computer Sciences
(e-mail: sanmaru98u@dgist.ac.kr; mercury@dgist.ac.kr;
dgwogh@dgist.ac.kr; ysk@dgist.ac.kr).}%
}

\maketitle

\begin{abstract}
Correlated faults from shared control can have dense physical support yet low rank over a base field, motivating rank-metric quantum codes for stacked architectures. We prove an asymmetric rank-metric Singleton bound for Calderbank--Shor--Steane codes, including degenerate codes, and characterize equality through commuting maximum-rank-distance matrix codes. Optimal pairs exist for every admissible parameter triple and are necessarily pure. Comparison with erasure bounds establishes an exact advantage in stacked rank distance for general stabilizer codes at fixed physical resources on certain tall layouts. We determine whether relaxing the recovery target enlarges the worst-case correctable radius. For Singleton-optimal pairs with positive logical dimension, the measured syndrome determines the projected error modulo stabilizers exactly at radii strictly below half the sector rank distance, on every layout and for every nonzero projector. The projected syndrome is recoverable at arbitrary radii exactly when the complementary check space is invariant under the adjoint projector; otherwise, the same radius limit applies. For nontrivial idempotents, this invariance requires more rows than columns or sector distance one. For trace-self-adjoint projectors, linear ambient projected-syndrome interfaces in both sectors exist exactly when the code splits as a tensor product across the projector. Designs placing the two check spaces in complementary projector images are necessarily one-sided whenever they encode logical information. Two explicit families realize the extreme cases of the equality structure, including a two-sided family with an efficient certifying decoder attaining the optimal unique-decoding radius.
\end{abstract}


\section{Introduction}
\label{sec:introduction}

\IEEEPARstart{T}{opological} codes are attractive for their local structure and high thresholds
\cite{fowler2012}, but the dominant failure mechanisms of evolving hardware increasingly depart
from independent, geometrically local Pauli noise: shared control, common-mode dephasing, crosstalk
and link-level failures can affect many carriers at once while remaining confined to a
low-dimensional error subspace. The multi-qubit bursts observed under ionizing radiation
\cite{mcewen2022resolving} are spatially extended and not in general of low rank; the geometry addressed here is the common-mode fault of a shared control
acting through one fixed transversal pattern on a subset of carriers, as in stacked and modular
cell architectures \cite{delfosse2024stacked}, dense in support but of rank one over the base
field. Rank-metric codes \cite{gabidulin1985,silva2009} are the natural response; Delfosse and Z\'emor employ them against circuit faults in stacked quantum memories \cite{delfosse2024stacked}, identifying a fault-tolerant extraction circuit and an efficient decoder as key prerequisites.

Two technical obstructions shape the theory. First, commutation for CSS constructions
\cite{calderbank1996,steane1996} is governed by the trace pairing
$\langle x,y\rangle_{\Tr}=\sum_j\Tr_{\Fqm/\F_q}(x_jy_j)$, and complementarity of two projectors as
operators, $e_1e_2=0$, does not by itself make their images trace-orthogonal: for the basis
$(1,\omega)$ of $\F_4$ the idempotent $\mathrm{diag}(1,0)$ satisfies $e_1e_2=0$ while
$\Tr(1\cdot\omega)=1$. Trace self-adjointness resolves this (Theorem~\ref{thm:css-admissibility}).
Second, the $\F_q$-linear operators through which a rank-metric decoder acts are not
$\Fqm$-linear, so in general $e\star(HE^{\top})\neq H(e\star E)^{\top}$ and an interface between
the measured syndrome and the projected one is needed.

The interface makes the following question precise, and the paper is organized around it:
\emph{does a weaker recovery target enlarge the worst-case correctable radius?} A decoder acting on
one projected branch does not need the error, nor even the full syndrome; it needs the projected
syndrome, or at most the projected error modulo stabilizer labels. Both targets are strictly weaker
than the error itself, and both are what a projector-based post-processing layer would supply. For codes attaining the bound proved below we answer this exactly, and the two targets behave
differently: the projected error modulo stabilizer labels never extends the radius beyond half the
sector distance, on any layout and for any nonzero projector, whereas the projected syndrome is
recoverable at every radius precisely in one invariance case, which for a nontrivial idempotent
requires more rows than columns or a sector distance equal to one.

\subsection{Main results}
\label{subsec:contributions}

All terms are defined in Sections~\ref{sec:preliminaries} and \ref{sec:interfaces}. Throughout,
labels are $a\times n$ matrices over $\F_q$, $M:=\max(a,n)$, $L:=\min(a,n)$, $K$ is the logical
dimension, $d^{R}_X$ and $d^{R}_Z$ are the two sector rank distances, and $D$ is the stacked rank
distance of a general stabilizer code on the same layout.

\emph{Theorem A (bounds; Theorem~\ref{thm:rank-singleton}, Proposition~\ref{prop:erasure-bounds},
Proposition~\ref{prop:css-specific}, Corollary~\ref{cor:css-penalty-fixed}).} Every CSS pair with
$K>0$ obeys $K\le an-M(d^{R}_X+d^{R}_Z-2)$, degenerate codes included. Every stabilizer code on the
layout obeys $K\le\min\{an-2n(D-1),\,(a-2\iota)(n-2\kappa)\}$ with $D-1=2\kappa+\iota$ and
$\iota\in\{0,1\}$. The CSS bound does not extend to that generality: at a fixed layout, a fixed
number of physical qudits and a fixed logical dimension, a code outside the class can have a
strictly larger stacked rank distance.

\emph{Theorem B (equality structure; Theorems~\ref{thm:singleton-equality-mrd},
\ref{thm:optimal-existence}, Corollary~\ref{cor:purity}).} Equality holds exactly for a
CSS-compatible pair of MRD matrix codes of dimensions $M(d^{R}_Z-1)$ and $M(d^{R}_X-1)$; the
triples $(K,d^{R}_X,d^{R}_Z)$ so attained are exactly
$\bigl(M(L-u-v),u+1,v+1\bigr)$ with $u+v\le L-1$, for every prime power $q$ and every layout; and
every optimal pair is pure, so that a degenerate code never attains the bound.

\emph{Theorem C (exact radii; Theorem~\ref{thm:radius-dichotomy},
Corollary~\ref{cor:optimal-rigidity}).} Let a pair be Singleton-optimal with $K>0$, on any layout,
and let $e$ be any nonzero $\F_q$-linear map acting componentwise, neither idempotency nor
self-adjointness being assumed. Then the projected error modulo stabilizers is recoverable from the
measured syndrome on the rank-$t$ ball if and only if $2t<d^{R}_X$, and never on the whole label
space; the projected syndrome is recoverable at every radius if $e^{\dagger}\star S_Z\subseteq S_Z$
and otherwise again exactly when $2t<d^{R}_X$. For $m\le n$, $d^{R}_X\ge2$ and a nontrivial
idempotent the invariant case cannot occur, so the two targets then have the same exact radius; it does
occur at $d^{R}_X=1$, and on tall layouts (Proposition~\ref{prop:stacking-sharpness}), and there
the projected syndrome is recoverable at every radius while the projected error modulo stabilizer
labels remains unrecoverable.

\emph{Theorem D (linear interfaces and their price;
Theorems~\ref{thm:invariant-splitting}, \ref{thm:quotient-characterization},
\ref{thm:no-go}, Corollary~\ref{cor:interface-cost}).} For an arbitrary $\F_q$-linear $e$ the
ambient quotient-valued interface exists exactly when $e\star S_Z^{\ptr}\subseteq S_X$. For a
trace-self-adjoint idempotent this is equivalent to both check spans being invariant and the
projected branch carrying no logical qudit; starting from a pair that is invariant but has a
nontrivial projected branch, adding $X$-checks to reach the interface costs exactly that branch
logical dimension; and a linear ambient projected-syndrome interface in both sectors exists if and
only if the code is a tensor product across the projector. If the two sectors occupy complementary projector images then
$\min(d^{R}_X,d^{R}_Z)=1$ for every code of positive logical dimension.

To clarify the scope and delineate the novelty of this work from existing literature, we summarize the foundational ingredients alongside our new contributions at the outset. Existing ingredients are the
Hamming-metric asymmetric Singleton bound for degenerate CSS codes and its equality correspondence
with nested maximum-distance-separable pairs \cite{sarvepalli2009asymmetric,ezerman2013aqmds}, the
rigidity of the idealizers of MRD codes \cite{lunardon2018kernels,csajbok2020idealizers}, relative
generalized matrix weights \cite{martinezpenas2018rgmw}, Gabidulin codes and their duality
\cite{gabidulin1985,delsarte1978}, and the two quantum rank-metric families
\cite{delfosse2024stacked,nizuka2026}. The primary novel contributions of this work are the rank-metric bound with its equality analysis, the realization of every admissible parameter triple, the purity of optimal pairs, the exact radii of the two projected targets, the characterization and logical cost of linear interfaces, and the exact fixed-resource optima.

Table~\ref{tab:interfaces} collects the six interface conditions, their exact criteria and their
values on Singleton-optimal pairs. Two explicit families realize the two extreme cases of the equality
structure: the quantum Gabidulin codes of \cite{delfosse2024stacked}, for which we determine both
sector distances exactly and supply a certifying decoder at the optimal unique-decoding radius; and
an idempotent-generated
Gabidulin family, which realizes the one-sided case $d^{R}_Z=1$ with exact parameters on every
admissible instance.

We emphasize that the aforementioned statements characterize the information-theoretic recoverability and the structural properties of post-processing maps, rather than computational complexity lower bounds.

\subsection{Organization and scope}
\label{subsec:scope}

Section~\ref{sec:preliminaries} fixes the matrix-label model, Section~\ref{sec:singleton} proves
the bound and compares it with the erasure bounds for an arbitrary stabilizer code, and
Section~\ref{sec:equality} determines the equality structure. Section~\ref{sec:interfaces}
develops the interface criteria and the exact radii; Section~\ref{sec:linear} treats linear
interfaces, their logical cost and one-sidedness. Section~\ref{sec:witnesses} gives the two witness
families and the decoders, and Section~\ref{sec:discussion} the position relative to prior work, the limitations and the open
problems. All results are unconditional; the finite computations that
accompany them are listed in Section~\ref{subsec:verification}.

\begin{table*}[t]
\centering
\caption{The six interface conditions for the $X$ sector, their exact criteria for an arbitrary
$\F_q$-linear componentwise map $e$ with trace adjoint $e^{\dagger}$, and their values on a
Singleton-optimal pair with $K>0$ and $e\neq0$; at $e=0$ both relative distances are $+\infty$ and
all six conditions hold. Here $H$ spans $S_Z$, $\cD_X:=\ker\sigma_H=S_Z^{\ptr}$, and $B_t$ is the
ball of labels of rank at most $t$. All one-way implications of \eqref{eq:interface-hierarchy} are strict, with witnesses
Example~\ref{ex:hierarchy-strict}, Propositions~\ref{prop:separation} and
\ref{prop:stacking-sharpness}; the $Z$-sector statements follow by sector interchange.}
\label{tab:interfaces}
\small
\begin{tabular}{@{}p{0.055\textwidth}p{0.30\textwidth}p{0.30\textwidth}p{0.27\textwidth}@{}}
\toprule
& recovery target, map class, domain & exact criterion & value on a Singleton-optimal pair\\
\midrule
$\mathrm{I}_1$ & projected syndrome, arbitrary map, whole label space &
$e^{\dagger}\star S_Z\subseteq S_Z$ (Theorem~\ref{thm:projected-recoverability}) &
fails for every nontrivial idempotent when $m\le n$ and $d^{R}_X\ge2$
(Corollary~\ref{cor:optimal-rigidity}); can hold at $d^{R}_X=1$ and for $a>n$
(Remark~\ref{rem:dichotomy}, Proposition~\ref{prop:stacking-sharpness})\\
$\mathrm{I}_2$ & projected syndrome, $\F_q$-linear map, rank-$t$ ball &
equivalent to $\mathrm{I}_1$ (Corollary~\ref{cor:linearity-dichotomy}) & as for $\mathrm{I}_1$\\
$\mathrm{I}_3$ & projected error modulo $S_X$, arbitrary map, whole label space &
$e\star\cD_X\subseteq S_X$ (Theorem~\ref{thm:quotient-interface}); for $e=e^{\dagger}=e^2$ this is
equivalent to both spans being invariant with $K_{\Imag(e)}=0$
(Theorem~\ref{thm:quotient-characterization}) & fails for every nonzero $e$
(Theorem~\ref{thm:radius-dichotomy})\\
$\mathrm{I}_4$ & projected error modulo $S_X$, $\F_q$-linear map, rank-$t$ ball &
equivalent to $\mathrm{I}_3$ (Corollary~\ref{cor:quotient-linearity}) & as for $\mathrm{I}_3$\\
$\mathrm{I}_5$ & projected syndrome, arbitrary map, rank-$t$ ball &
$2t<\delta_{\mathrm{syn}}(e)$ (Lemma~\ref{lem:relative-criterion}) &
$\delta_{\mathrm{syn}}(e)=+\infty$ if $e^{\dagger}\star S_Z\subseteq S_Z$, else $d^{R}_X$
(Theorem~\ref{thm:radius-dichotomy})\\
$\mathrm{I}_6$ & projected error modulo $S_X$, arbitrary map, rank-$t$ ball &
$2t<\delta_{\mathrm{quot}}(e)$ (Lemma~\ref{lem:relative-criterion}) &
$\delta_{\mathrm{quot}}(e)=d^{R}_X$ on every layout (Theorem~\ref{thm:radius-dichotomy})\\
\bottomrule
\end{tabular}
\end{table*}

\section{Preliminaries}
\label{sec:preliminaries}

Let $q$ be a prime power and $m\ge1$ an extension degree. Vectors are length-$n$ row vectors over
$\Fqm$, and all linear subspaces are $\F_q$-linear unless stated otherwise. The two-sided witness
family of Section~\ref{subsec:qgab} is stated for $q=2$, so that each $\F_{2^m}$ symbol
corresponds to $m$ physical qubits, while the idempotent-generated family of
Section~\ref{subsec:igg} and all algebraic definitions are stated for a general prime power
\cite{ashikhmin2001,ketkar2006}.

\subsection{Trace pairing and trace-dual spaces}
\label{subsec:trace-pairing}

The field trace $\Tr:=\Tr_{\Fqm/\F_q}$ induces a nondegenerate $\F_q$-bilinear form. For
$x,z\in\Fqm^n$ define
\begin{equation}
\label{eq:trace-inner-product}
\langle x,z\rangle_{\Tr}:=\sum_{j=1}^{n}\Tr(x_jz_j),
\end{equation}
and for an $\F_q$-linear subspace $C\subseteq\Fqm^n$ let
$C^{\ptr}:=\{y:\langle c,y\rangle_{\Tr}=0\text{ for all }c\in C\}$. Since
\eqref{eq:trace-inner-product} is nondegenerate on $\Fqm^n$ regarded as an $nm$-dimensional
$\F_q$-vector space,
\begin{equation}
\label{eq:trace-dual-dimension}
\dim_{\F_q}(C^{\ptr})=nm-\dim_{\F_q}(C),\qquad (C^{\ptr})^{\ptr}=C,
\end{equation}
for every $\F_q$-linear subspace $C$ \cite{ashikhmin2001,ketkar2006}; consequently $(\cdot)^{\ptr}$
is inclusion-reversing and
\begin{equation}
\label{eq:dual-intersection}
(C_1\cap C_2)^{\ptr}=C_1^{\ptr}+C_2^{\ptr}.
\end{equation}
For an $\F_q$-linear map $e:\Fqm\to\Fqm$ we write the componentwise extension as
$e\star(v_1,\dots,v_n):=(e(v_1),\dots,e(v_n))$, where $\star$ denotes componentwise action and not
multiplication. The trace adjoint $e^{\dagger}$ is the unique $\F_q$-linear map satisfying
\begin{equation}
\label{eq:trace-adjoint}
\Tr\bigl(e(x)\,y\bigr)=\Tr\bigl(x\,e^{\dagger}(y)\bigr)\quad\text{for all }x,y\in\Fqm,
\end{equation}
and $e$ is \emph{self-adjoint} if $e=e^{\dagger}$. We write
$\Fix(e^{\dagger}):=\{h\in\Fqm:e^{\dagger}(h)=h\}$ and extend the notation componentwise. In a basis
$B=(b_1,\dots,b_m)$ with trace Gram matrix $G_B:=(\Tr(b_ib_j))_{i,j}$, the matrix of $e^{\dagger}$
is $G_B^{-1}M_e^{\top}G_B$, so that self-adjointness reads $M_e=G_B^{-1}M_e^{\top}G_B$; if a
self-dual basis exists then $G_B=I$ and this reduces to matrix symmetry.

\subsection{Pauli labels and trace-CSS commutation}
\label{subsec:nonbinary-css}

Pauli operators on $n$ $\Fqm$-valued symbols are indexed by pairs $(x,z)\in\Fqm^n\times\Fqm^n$
\cite{ashikhmin2001,ketkar2006,gottesman1997}. For a CSS construction let $C_X,C_Z\subseteq\Fqm^n$
denote the $\F_q$-linear spans of the implemented $X$- and $Z$-type check labels. The commutation
condition is
\begin{equation}
\label{eq:trace-css-commutation}
\langle x,z\rangle_{\Tr}=0\ \ \text{for all }x\in C_X,\ z\in C_Z,\quad\text{i.e.}\quad
C_Z\subseteq C_X^{\ptr}.
\end{equation}
Given a check label $h$ and an error label $E$ in a fixed Pauli sector, the measurable trace
syndrome is
\begin{equation}
\label{eq:trace-syndrome-single}
s(h;E):=\langle h,E\rangle_{\Tr}\in\F_q,
\end{equation}
and for a family $H=(h^{(1)},\dots,h^{(s)})$ we write
$\sigma_H(E):=\bigl(s(h^{(i)};E)\bigr)_{i=1}^{s}\in\F_q^{s}$.

\subsection{Physical embedding and rank invariance}
\label{subsec:physical-embedding}

Realizing \eqref{eq:trace-css-commutation} on physical qudits requires a linear identification of
labels with Pauli supports, chosen separately for each sector. Let $B$ be an ordered $\F_q$-basis
with coordinate map $[\cdot]_B$ and Gram matrix $G_B$, symmetric and invertible by nondegeneracy.

\begin{proposition}[Sector-split embedding and logical count]
\label{prop:sector-split-embedding}
Define componentwise $\varphi_X(x):=\bigl(G_B[x_j]_B\bigr)_j$ and
$\varphi_Z(z):=\bigl([z_j]_B\bigr)_j$, and realize labels on $nm$ physical $\F_q$-qudits by
$X(x)\mapsto X^{\varphi_X(x)}$ and $Z(z)\mapsto Z^{\varphi_Z(z)}$. Then
\begin{equation}
\label{eq:exponent-identity}
\varphi_X(x)\cdot\varphi_Z(z)=\sum_{j=1}^{n}\Tr(x_jz_j)=\langle x,z\rangle_{\Tr}.
\end{equation}
Consequently: (i) for prime $q$ the realized operators commute if and only if
$\langle x,z\rangle_{\Tr}=0$, so every trace-orthogonal pair $(S_X,S_Z)$ defines an abelian physical
stabilizer group, and the measured parity of any realized check against any error equals the
corresponding trace syndrome; (ii) for a general prime power $q=p^{\nu}$ trace orthogonality implies
commutation, the converse being weaker because the physical phase is $\Tr_{\F_q/\F_p}$ of
\eqref{eq:exponent-identity}; (iii) the mirrored convention
$(\varphi_X,\varphi_Z)=\bigl([\cdot]_B,\,G_B[\cdot]_B\bigr)$ has the same properties; and (iv) for
every prime power $q$ the realized group of a trace-orthogonal pair with $\F_q$-linear spans
$S_X,S_Z$ is scalar-free and elementary abelian of order $q^{\dim_{\F_q}S_X+\dim_{\F_q}S_Z}$, so
that the code space has dimension $q^{K_q}$ with
\begin{equation}
\label{eq:logical-count}
K_q=N-\dim_{\F_q}S_X-\dim_{\F_q}S_Z,\qquad N=nm.
\end{equation}
At most one sector may carry the Gram factor, since applying it to both yields the exponent
$\sum_j[x_j]_B^{\top}G_B^2[z_j]_B$, which differs from the trace pairing whenever $G_B^2\neq G_B$,
as is the case for the polynomial basis of $\F_{2^5}=\F_2[x]/(x^5+x^2+1)$ used in
Example~\ref{ex:primary-instance}.
\end{proposition}

\begin{proof}
For a single symbol coordinate, bilinearity of the trace form gives
$[x]_B^{\top}G_B[z]_B=\sum_{i,\ell}x_iz_{\ell}\Tr(b_ib_{\ell})=\Tr(xz)$; summing over $j$ yields
\eqref{eq:exponent-identity}. Two generalized Pauli strings commute exactly when the $\F_p$-valued
symplectic phase vanishes; for prime $q=p$ that phase is \eqref{eq:exponent-identity} itself,
giving (i), and for $q=p^{\nu}$ it is given by $\Tr_{\F_q/\F_p}$ of this symplectic phase, so $\Tr_{\F_q/\F_p}(0)=0$ gives (ii). Claim (iii) follows from the symmetry of $G_B$, and (iv) is
Lemma~\ref{lem:fp-accounting}, which for prime $q$ reduces to the standard count for a scalar-free
abelian subgroup of the generalized Pauli group \cite{ashikhmin2001,ketkar2006,gottesman1997}.
\end{proof}

\begin{corollary}[Embedding invariance of rank-metric quantities]
\label{cor:rank-invariance}
Let $\Phi_B(E)\in\F_q^{m\times n}$ be the coordinate matrix of a label $E$, with columns
$[E_j]_B$. If the physical coordinates of every symbol are transformed by one common invertible
$M\in\mathrm{GL}_m(\F_q)$, then $\Phi_B(E)\mapsto M\Phi_B(E)$ and rank is preserved; hence $K_q$,
$d^{R}_X$, $d^{R}_Z$, $d_R(S_Z)$, $d_R(\Kker)$ and rank-$t$ burst correctability agree under both
conventions of Proposition~\ref{prop:sector-split-embedding}, which differ per symbol by left
multiplication with $G_B$, whereas Hamming functionals such as physical check weights and physical
distances depend on the convention. The requirement that $M$ be uniform across all carriers cannot be relaxed: for $q=2$, $m=n=2$ the
label with both coordinates equal to $1$ has rank one, and applying different invertible maps to
the two carriers can raise its rank to two, so rank-metric quantities are invariants of the
aligned layout and not of the code alone.
\end{corollary}

\begin{proof}
$\rank(M\Phi)=\rank(\Phi)$ for invertible $M$; the listed rank quantities are defined through ranks
of coordinate matrices, while the Hamming functionals count nonzero physical coordinates, which $M$
need not preserve.
\end{proof}

When $q$ is not prime the physical outcome is the absolute trace of the pairing;
Appendix~\ref{app:nonprime} shows that the $\F_q$-valued syndrome is nevertheless measurable for
every prime power, at a $\nu$-fold increase in measured generators and with no change to the
stabilizer group. The next lemma converts trace-dual membership into a statement about physical
Pauli operators.

\begin{lemma}[Physical normalizer of a base-field-linear check space]
\label{lem:physical-normalizer}
Let $q=p^{\nu}$ with $p$ prime, adopt the embedding of
Proposition~\ref{prop:sector-split-embedding}, and let $S\subseteq\Fqm^n$ be an $\F_q$-linear space
of $Z$-type check labels. Write $G(S)$ for the group generated by the realized operators
$\{Z^{\varphi_Z(h)}:h\in S\}$. Then a realized $X$-type Pauli $X^{\varphi_X(E)}$ commutes with every
element of $G(S)$ if and only if $E\in S^{\ptr}$, and symmetrically with the two sectors exchanged.
Consequently, for every prime power $q$, the normalizer of the realized CSS stabilizer group is
$S_Z^{\ptr}$ in the $X$ sector and $S_X^{\ptr}$ in the $Z$ sector.
\end{lemma}

\begin{proof}
Sufficiency is Proposition~\ref{prop:sector-split-embedding}(ii): if $\langle h,E\rangle_{\Tr}=0$
the physical phase vanishes for every generator and hence on all of $G(S)$. For necessity with
prime $q$, the commutation phase of $X^{\varphi_X(E)}$ against $Z^{\varphi_Z(h)}$ is
$\langle h,E\rangle_{\Tr}$ itself by Proposition~\ref{prop:sector-split-embedding}(i), so
commutation with every generator forces $E\in S^{\ptr}$. For $q=p^{\nu}$ with $\nu>1$ the phase is
$\Tr_{\F_q/\F_p}\bigl(\langle h,E\rangle_{\Tr}\bigr)$ by
Proposition~\ref{prop:sector-split-embedding}(ii), which a single generator does not invert; but
$\lambda h\in S$ for every $\lambda\in\F_q$, so commutation with every element of $G(S)$ gives
$\Tr_{\F_q/\F_p}\bigl(\lambda\langle h,E\rangle_{\Tr}\bigr)=0$ for all $\lambda$, whence
$\langle h,E\rangle_{\Tr}=0$ by nondegeneracy of \eqref{eq:absolute-trace-nondeg};
Appendix~\ref{app:nonprime} records the measurement protocol and the group accounting. The final assertion follows by
applying the equivalence to each sector separately, the $X$-part of a label being paired only with
$Z$-checks and the $Z$-part only with $X$-checks.
\end{proof}

We adopt throughout the convention of Proposition~\ref{prop:sector-split-embedding}: $Z$-labels
in raw coordinates and $X$-labels in Gram-weighted coordinates, so a weight-one physical $Z$-Pauli
at $(j,\ell)$ carries the label $b_{\ell}$ and each single-symbol $X$-check acts on
$\{(j,\ell):\Tr(b_{\ell}u)\neq0\}$ when its label at carrier $j$ is $u$, since
$(G_B[u]_B)_{\ell}=\Tr(b_{\ell}u)$.

\subsection{Rank metric and code notation}
\label{subsec:rank-metric}

For $E\in\Fqm^n$ the rank weight is $\wt_R(E):=\rank_{\F_q}\Phi_B(E)$, independent of $B$
\cite{gabidulin1985}. For an $\F_q$-linear code $C$, in either ambient space $\Fqm^n$ or
$\F_q^{a\times n}$, we set $d_R(C):=\min\{\wt_R(c):0\neq c\in C\}$, \emph{with the convention
$d_R(\{0\}):=+\infty$}. Writing $B_{\tau}$ for the ball of elements of rank at most $\tau$ in the
ambient space of $C$, every hypothesis of the form $d_R(C)>\tau$ in this paper is used only through
the equivalence
\[
d_R(C)>\tau\iff C\cap B_{\tau}=\{0\},
\]
which under that convention holds for every $C$, so such hypotheses are vacuously satisfied by
$C=\{0\}$. We write
$\llbracket N,K,d^{R}_X/d^{R}_Z\rrbracket_q$ for a CSS code on $N$ physical $\F_q$-qudits
encoding $K$ logical qudits, where $d^{R}_X$ and $d^{R}_Z$ are the minimum rank weights of
nontrivial $X$- and $Z$-logical representatives; the superscript $R$ is retained throughout, since
the bracket notation is usually read as a Hamming distance and the two are related only by
$\wt_H\ge\wt_R$. When $m$ divides $K$, we write $K_{q^m}:=K/m$, a normalization of the logical
dimension that does not assume an $\Fqm$-module structure. The rank weight is subadditive,
\begin{equation}
\label{eq:rank-subadditive}
\wt_R(E+E')\le\wt_R(E)+\wt_R(E')\qquad(E,E'\in\Fqm^n),
\end{equation}
because the column space of $\Phi_B(E+E')$ is contained in the sum of the column spaces of $\Phi_B(E)$ and
$\Phi_B(E')$. We record two facts used throughout.

\begin{lemma}[Rank-one generators of a branch space]
\label{lem:rankone-span}
For any nonzero $\F_q$-subspace $V\subseteq\Fqm$, the rank-one vectors $\{c\,b:b\in V,\
c\in\F_q^n\}$ span $V^n$ over $\F_q$. Consequently no proper $\F_q$-subspace of $V^n$ contains every
rank-one vector of $V^n$.
\end{lemma}

\begin{proof}
Fix an $\F_q$-basis $b_1,\dots,b_d$ of $V$. For each $i$ and each unit vector
$\varepsilon_j\in\F_q^n$, the vector $\varepsilon_jb_i$ has $b_i$ in coordinate $j$ and $0$
elsewhere; these $dn$ vectors are $\F_q$-independent and $\dim_{\F_q}V^n=dn$, so they form a basis.
\end{proof}

\begin{lemma}[Rank non-expansion under projection]
\label{lem:rank-non-expansion}
For $\F_q$-linear $e$ and every $E\in\Fqm^n$, $\wt_R(e\star E)\le\wt_R(E)$.
\end{lemma}

\begin{proof}
Let $W=\mathrm{span}_{\F_q}\{E_1,\dots,E_n\}$. Since $e$ is $\F_q$-linear,
$\dim_{\F_q}e(W)\le\dim_{\F_q}W$, and the coordinates of $e\star E$ lie in $e(W)$; hence
$\wt_R(e\star E)\le\dim_{\F_q}e(W)\le\dim_{\F_q}W=\wt_R(E)$.
\end{proof}

\subsection{Gabidulin codes and the two regimes}
\label{subsec:gabidulin-regimes}

Rank-metric codes over $\Fqm$ are described by $\F_q$-linearized polynomials
$f(x)=\sum_{i=0}^{k-1}f_ix^{q^i}$ \cite{gabidulin1985}; $\mathcal L_{<a}$ denotes the $\F_q$-space
of those of $q$-degree less than $a$, so $\dim_{\F_q}\mathcal L_{<a}=am$ and
$\mathcal L_{<m}=\mathrm{End}_{\F_q}(\Fqm)$. A Gabidulin code of length $n$ evaluates $f$ at
$\F_q$-independent points, $\mathrm{ev}(f):=(f(\alpha_1),\dots,f(\alpha_n))$; for $n\le m$ it
attains $d_R=n-k+1$ and decodes uniquely up to $t=\lfloor(d_R-1)/2\rfloor$ \cite{gabidulin1985},
while for $n>m$ no $n$ independent points exist and the worst-case guarantees do not apply. We use
throughout the rank-metric Singleton bound in the matrix form
\begin{equation}
\label{eq:matrix-singleton}
\dim_{\F_q}C\le\max(a,n)\bigl(\min(a,n)-d_R(C)+1\bigr)
\end{equation}
for every nonzero $\F_q$-linear $C\subseteq\F_q^{a\times n}$, the zero code being excluded by the
convention $d_R(\{0\})=+\infty$,
and we use the fact that the Delsarte dual of an MRD code is MRD
\cite{gabidulin1985,delsarte1978}.

\subsection{Matrix-label CSS pairs and the stacked rank distance}
\label{subsec:matrix-css}

All bounds below are stated for pairs of matrix codes, the ambient label space and a projected
branch being the two instances. Fix a layout, that is, integers $a,n\ge1$, and write
$M:=\max(a,n)$, $L:=\min(a,n)$. A \emph{CSS pair} on the $a\times n$ layout is a pair of
$\F_q$-subspaces $S_X,S_Z\subseteq\F_q^{a\times n}$ with $S_Z\subseteq S_X^{\pstd}$, where
$\langle X,Y\rangle_{\mathrm{std}}:=\mathrm{tr}(X^{\top}Y)$; its logical dimension is
$K:=an-\dim_{\F_q}S_X-\dim_{\F_q}S_Z$ and its \emph{sector rank distances} are
\begin{align*}
d^{R}_X&:=\min\{\rank\ell:\ell\in S_Z^{\pstd}\setminus S_X\},\\
d^{R}_Z&:=\min\{\rank\ell:\ell\in S_X^{\pstd}\setminus S_Z\},
\end{align*}
the minimum over an empty set being $+\infty$ as in Section~\ref{subsec:rank-metric}. For labels in
$\Fqm^n$ the coordinate map $\Phi_B$ realizes a trace-orthogonal pair as a pair in
$\F_q^{m\times n}$, so that $a=m$, rows being indexed by intra-carrier position and columns by
carrier as in Proposition~\ref{prop:sector-split-embedding}; by
Lemma~\ref{lem:physical-normalizer} the sector rank distances are then the minimum rank weights of
nontrivial $X$- and $Z$-logical representatives. A \emph{rank-one burst} is a label of rank weight
one, that is, of the form $c\,b$ with $c\in\F_q^n$ and $b\in\Fqm$: a common-mode pattern applied
through one fixed intra-carrier direction.

The trace pairing carries a Gram factor, which the following lemma eliminates, thereby allowing
subsequent statements to be established for the standard pairing.

\begin{lemma}[Gram reduction]
\label{lem:gram-reduction}
Let $G\in\mathrm{GL}_a(\F_q)$ be symmetric and let
$\langle A,B\rangle_G:=\mathrm{tr}(A^{\top}GB)$ on $\F_q^{a\times n}$, so that
$\langle A,B\rangle_G=\langle GA,B\rangle_{\mathrm{std}}$ with
$\langle X,Y\rangle_{\mathrm{std}}:=\mathrm{tr}(X^{\top}Y)$. Let
$S_X,S_Z\subseteq\F_q^{a\times n}$ be $\F_q$-subspaces with $S_Z\subseteq S_X^{\perp_G}$ and define
the sector rank distances of the pair by
$d^{R}_X:=\min\{\rank\ell:\ell\in S_Z^{\perp_G}\setminus S_X\}$ and
$d^{R}_Z:=\min\{\rank\ell:\ell\in S_X^{\perp_G}\setminus S_Z\}$. Then
$S_Z\subseteq(GS_X)^{\pstd}$, and the pairs $(S_X,S_Z)$ under $\langle\cdot,\cdot\rangle_G$ and
$(GS_X,S_Z)$ under $\langle\cdot,\cdot\rangle_{\mathrm{std}}$ have the same logical dimension
$K=an-\dim S_X-\dim S_Z$ and the same sector rank distances.
\end{lemma}

\begin{proof}
From $\langle A,B\rangle_G=\langle GA,B\rangle_{\mathrm{std}}$ we obtain
$S_X^{\perp_G}=(GS_X)^{\pstd}$ and $S_Z^{\perp_G}=G^{-1}\bigl(S_Z^{\pstd}\bigr)$. Left
multiplication by $G$ is an $\F_q$-linear bijection of $\F_q^{a\times n}$ preserving rank, so
$\dim GS_X=\dim S_X$ and the logical dimensions agree. The $Z$-logical sets
$S_X^{\perp_G}\setminus S_Z$ and $(GS_X)^{\pstd}\setminus S_Z$ coincide, so $d^{R}_Z$ is unchanged;
and $G$ maps $S_Z^{\perp_G}\setminus S_X$ bijectively onto $S_Z^{\pstd}\setminus GS_X$ preserving
rank, so $d^{R}_X$ is unchanged.
\end{proof}

Finally, we fix the parameter against which the CSS bound of Section~\ref{sec:singleton} is
compared. We realize $an$ $\F_q$-qudits as an $a\times n$ array as described above. A Pauli label is a pair
$(A\mid B)\in\F_q^{a\times n}\times\F_q^{a\times n}$, and $\rank(A\mid B)$ denotes the rank of the
$a\times2n$ matrix formed by juxtaposition; the \emph{stacked rank distance} $D$ of a stabilizer
code is the minimum of that quantity over $N(S)\setminus S$, as in
\cite{delfosse2024stacked,nizuka2026}. For a CSS pair one has $D=\min(d^{R}_X,d^{R}_Z)$: every
label of $N(S)\setminus S$ has $X$-part outside $S_X$ or $Z$-part outside $S_Z$, and
$\rank(A\mid B)\ge\max(\rank A,\rank B)$, while a minimum-rank $X$-logical with vanishing $Z$-part
attains the minimum.

\section{An Asymmetric Rank-Metric Singleton Bound}
\label{sec:singleton}

This section proves the bound, records the erasure bounds that hold for an arbitrary stabilizer
code on the same layout, separates the two effects by which the CSS bound improves on their row
term, and shows that the CSS bound itself does not extend beyond the class.

\begin{theorem}[Asymmetric rank-metric Singleton bound]
\label{thm:rank-singleton}
Let $S_X,S_Z\subseteq\F_q^{a\times n}$ be $\F_q$-subspaces with $S_Z\subseteq S_X^{\pstd}$, let
$K:=an-\dim_{\F_q}S_X-\dim_{\F_q}S_Z>0$, and let $d^{R}_X,d^{R}_Z$ be the sector rank distances.
Then
\begin{equation}
\label{eq:rank-singleton}
K\ \le\ an-\max(a,n)\bigl(d^{R}_X+d^{R}_Z-2\bigr),
\end{equation}
equivalently, $d^{R}_X$ and $d^{R}_Z$ being integers,
\begin{equation}
\label{eq:rank-singleton-int}
d^{R}_X+d^{R}_Z\ \le\ \min(a,n)+2-\Bigl\lceil\frac{K}{\max(a,n)}\Bigr\rceil .
\end{equation}
Assume $a\le n$, the case $a>n$ following by transposition. For a row set
$A\subseteq\{1,\dots,a\}$ write $U_A$ for the subspace of matrices vanishing outside the rows in
$A$, and $\pi_A$ for the restriction to the rows in $A$, so that $\ker\pi_A=U_{A^{c}}$. Equality
holds in \eqref{eq:rank-singleton} if and only if, for some---equivalently, for every---pair of
disjoint row sets $A,B$ with $|A|=d^{R}_X-1$ and $|B|=d^{R}_Z-1$,
\begin{equation}
\label{eq:singleton-equality}
S_X\cap\ker\pi_B=S_X\cap U_A\qquad\text{and}\qquad S_Z\cap\ker\pi_A=S_Z\cap U_B .
\end{equation}
By Lemma~\ref{lem:gram-reduction} the bound applies verbatim to trace-orthogonal pairs in $\Fqm^n$,
with $a=m$, and to their restrictions to a branch $\Imag(e)^n$, with $a=\dim_{\F_q}\Imag(e)$.
\end{theorem}

\begin{proof}
Assume $a\le n$. For $A\subseteq\{1,\dots,a\}$ the restriction $\pi_A:U_A\to\F_q^{|A|\times n}$ is
an isomorphism and $\langle M,N\rangle_{\mathrm{std}}=\langle\pi_AM,\pi_AN\rangle_{\mathrm{std}}$
for $M\in U_A$, the rows outside $A$ contributing nothing.

\emph{Step 1: low-rank supports carry no logical labels.} Let $|A|=u\le d^{R}_X-1$. Every
$M\in U_A$ has $\rank M\le u<d^{R}_X$, while every element of $S_Z^{\pstd}\setminus S_X$ has rank at
least $d^{R}_X$; so $S_Z^{\pstd}\cap U_A\subseteq S_X$, and the reverse inclusion is
$S_X\subseteq S_Z^{\pstd}$. Hence $S_Z^{\pstd}\cap U_A=S_X\cap U_A$.

\emph{Step 2: the dimension of that intersection.} For $M\in U_A$, membership in $S_Z^{\pstd}$ is
equivalent to $\pi_AM\in(\pi_AS_Z)^{\pstd}$ inside $\F_q^{|A|\times n}$, so
$\dim(S_Z^{\pstd}\cap U_A)=un-\dim\pi_A(S_Z)$ and Step~1 gives
\begin{equation}
\label{eq:step2}
\begin{split}
\dim(S_X\cap U_A)&=un-\dim\pi_A(S_Z),\\
\dim(S_Z\cap U_B)&=vn-\dim\pi_B(S_X),
\end{split}
\end{equation}
the second for $B$ with $|B|=v\le d^{R}_Z-1$.

\emph{Step 3: rank--nullity on disjoint supports.} Let $A\cap B=\emptyset$, so
$U_A\subseteq\ker\pi_B$ and $U_B\subseteq\ker\pi_A$. Then
\begin{align}
\dim S_X&=\dim\pi_B(S_X)+\dim(S_X\cap\ker\pi_B)\nonumber\\
&\ge\dim\pi_B(S_X)+\dim(S_X\cap U_A),\label{eq:step3a}\\
\dim S_Z&=\dim\pi_A(S_Z)+\dim(S_Z\cap\ker\pi_A)\nonumber\\
&\ge\dim\pi_A(S_Z)+\dim(S_Z\cap U_B).\label{eq:step3b}
\end{align}
Adding \eqref{eq:step3a} and \eqref{eq:step3b} and inserting \eqref{eq:step2}, the two projection dimensions cancel and
\begin{equation}
\label{eq:step3}
\dim S_X+\dim S_Z\ \ge\ n\,(u+v),
\end{equation}
that is, $K\le n(a-u-v)$ for all disjoint $A,B$ with $|A|=u\le d^{R}_X-1$ and
$|B|=v\le d^{R}_Z-1$.

\emph{Step 4: choice of the two supports.} If $(d^{R}_X-1)+(d^{R}_Z-1)\ge a$ one may choose
disjoint $A,B$ with $|A|\le d^{R}_X-1$, $|B|\le d^{R}_Z-1$ and $|A|+|B|=a$, and Step~3 gives
$K\le0$, contradicting $K>0$. Hence $u:=d^{R}_X-1$ and $v:=d^{R}_Z-1$ are admissible, and Step~3
yields $K\le n(a-u-v)=an-\max(a,n)(d^{R}_X+d^{R}_Z-2)$ since $\max(a,n)=n$; dividing by $n$ and
using integrality gives \eqref{eq:rank-singleton-int}.

\emph{Transposition.} The map $M\mapsto M^{\top}$ is an $\F_q$-linear rank-preserving bijection
preserving $\langle\cdot,\cdot\rangle_{\mathrm{std}}$, so it carries $(S_X,S_Z)$ to a pair in
$\F_q^{n\times a}$ with the same $K$ and sector distances. For $a>n$, applying the bound established for the case $a\le n$ yields $K\le a\bigl(n-d^{R}_X-d^{R}_Z+2\bigr)$, which matches \eqref{eq:rank-singleton} with $\max(a,n)=a$.

\emph{The equality condition.} The right-hand side of \eqref{eq:rank-singleton} is $n(a-u-v)$ and
does not depend on the chosen supports, so equality forces $\dim S_X+\dim S_Z=n(u+v)$ and hence
equality in \eqref{eq:step3}, therefore in \eqref{eq:step3a} and \eqref{eq:step3b} separately, for
every admissible disjoint pair $(A,B)$. Since $U_A\subseteq\ker\pi_B$ and $U_B\subseteq\ker\pi_A$,
that is exactly \eqref{eq:singleton-equality}; conversely \eqref{eq:singleton-equality} for one
admissible pair turns \eqref{eq:step3a} and \eqref{eq:step3b} into equalities, hence
\eqref{eq:step3} into an equality, so that $\dim S_X+\dim S_Z=n(u+v)$ and $K=n(a-u-v)$. Necessity
therefore yields the condition for every admissible pair and sufficiency requires it for one.
\end{proof}

\subsection{Comparison with the erasure bounds}
\label{subsec:erasure}

Theorem~\ref{thm:rank-singleton} should be read against what erasure alone yields for an arbitrary
stabilizer code on the layout. Row erasure yields one such bound; carrier erasure yields a
second, which is the stronger of the two on layouts with many more rows than columns. Both are
special cases of a single statement.

\begin{proposition}[Erasure Singleton bounds for general stabilizer codes]
\label{prop:erasure-bounds}
Let $Q$ be a stabilizer code on the $a\times n$ layout with $K>0$ logical qudits and stacked rank
distance $D$, and write $D-1=2\kappa+\iota$ with $\iota\in\{0,1\}$. Then $a\ge2(D-1)$,
$n\ge2\kappa+1$, and
\begin{align}
K&\le an-2n(D-1),\label{eq:row-erasure}\\
K&\le (a-2\iota)(n-2\kappa).\label{eq:col-erasure}
\end{align}
More generally, for all nonnegative integers $r_1,r_2,c_1,c_2$ with $r_i+2c_i\le D-1$,
$r_1+r_2\le a$ and $c_1+c_2\le n$,
\begin{equation}
\label{eq:erasure-general}
K\le(a-r_1-r_2)(n-c_1-c_2),
\end{equation}
and the minimum of the right-hand side of \eqref{eq:erasure-general} over all admissible choices
is the smaller of \eqref{eq:row-erasure} and \eqref{eq:col-erasure}. The column bound
\eqref{eq:col-erasure} is strictly smaller than the row bound \eqref{eq:row-erasure} if and only
if $\kappa\ge1$ and $a>2n+2\iota$. In particular, on the $a\times2$ layout every stabilizer code
with $K>0$ has $D\le2$, for every $a$.
\end{proposition}

\begin{proof}
\emph{Step 1: correctable erasures.} Let $R\subseteq\{1,\dots,a\}$, $C\subseteq\{1,\dots,n\}$ and
let $P$ be a Pauli supported on the qudit set $(R\times\{1,\dots,n\})\cup(\{1,\dots,a\}\times C)$.
In its label $(A\mid B)$ every row outside $R$ vanishes outside the $2|C|$ columns of $(A\mid B)$
indexed by the carriers in $C$, so those rows span a space of dimension at most $2|C|$, while the
rows in $R$ contribute at most $|R|$; hence $\rank(A\mid B)\le|R|+2|C|$. If $|R|+2|C|\le D-1$, no
element of $N(S)\setminus S$ is supported on that qudit set, and its erasure is correctable
\cite{gottesman1997}.

\emph{Step 2: two disjoint erasures.} Let $r_1,r_2,c_1,c_2$ be as in the statement, choose disjoint
row sets $R_1,R_2$ with $|R_i|=r_i$ and disjoint carrier sets $C_1,C_2$ with $|C_i|=c_i$, and put
\[
E_i:=\bigl(R_i\times\{1,\dots,n\}\bigr)\cup
\bigl((\{1,\dots,a\}\setminus(R_1\cup R_2))\times C_i\bigr).
\]
The sets $E_1,E_2$ are disjoint, and each is contained in a set of the form of Step~1 with
$|R|+2|C|=r_i+2c_i\le D-1$, hence correctable. The complement consists of the
$(a-r_1-r_2)(n-c_1-c_2)$ qudits $(i,j)$ with $i\notin R_1\cup R_2$ and $j\notin C_1\cup C_2$, which we denote by $C^{\ast}$. Let $\mathsf R$ be a reference system maximally entangled with the code space, so
that $S(\mathsf R)=K\log q$. Correctability of the erasure $E_i$ is equivalent to
$I(\mathsf R{:}E_i)=0$ \cite{cerf1997,grassl2022entropic}, that is, to
$S(\mathsf RE_i)=S(\mathsf R)+S(E_i)$. Purity of the state on $\mathsf RE_1E_2C^{\ast}$ gives
$S(\mathsf RE_1)=S(E_2C^{\ast})\le S(E_2)+S(C^{\ast})$ and
$S(\mathsf RE_2)=S(E_1C^{\ast})\le S(E_1)+S(C^{\ast})$. Adding the two resulting inequalities and
canceling $S(E_1)+S(E_2)$ yields $2S(\mathsf R)\le2S(C^{\ast})$, so
$K\log q=S(\mathsf R)\le S(C^{\ast})\le|C^{\ast}|\log q$, which is \eqref{eq:erasure-general}.

\emph{Step 3: feasibility.} If $a<2(D-1)$, take $r_1=\lceil a/2\rceil$, $r_2=\lfloor a/2\rfloor$
and $c_1=c_2=0$; both $r_i\le D-1$, and \eqref{eq:erasure-general} gives $K\le0$, contradicting
$K>0$. If $n\le2\kappa$, take $c_1=\lceil n/2\rceil$, $c_2=\lfloor n/2\rfloor$ and $r_1=r_2=0$;
both $c_i\le\kappa$, so $r_i+2c_i\le2\kappa\le D-1$, and \eqref{eq:erasure-general} again gives
$K\le0$. Hence $a\ge2(D-1)$ and $n\ge2\kappa+1$.

\emph{Step 4: optimization.} The right-hand side of \eqref{eq:erasure-general} is nonincreasing in
$r_1+r_2$ and in $c_1+c_2$, so its minimum is attained at choices with $r_i+2c_i=D-1$, which force
$r_i=\iota+2j_i$ and $c_i=\kappa-j_i$ with $0\le j_i\le\kappa$. All such choices are admissible:
$r_1+r_2\le2\iota+4\kappa=2(D-1)\le a$ by Step~3 and $c_1+c_2\le2\kappa\le n-1$. With
$s:=j_1+j_2\in\{0,\dots,2\kappa\}$ the right-hand side equals
\[
f(s):=(a-2\iota-2s)(n-2\kappa+s),
\]
a concave quadratic in $s$, whose minimum over $\{0,\dots,2\kappa\}$ is therefore attained at an
endpoint: $f(0)=(a-2\iota)(n-2\kappa)$ is \eqref{eq:col-erasure} and
$f(2\kappa)=n\bigl(a-2(D-1)\bigr)$ is \eqref{eq:row-erasure}. Finally $f(0)<f(2\kappa)$ if and only
if $2\kappa(a-2\iota)>4\kappa n$, that is, if and only if $\kappa\ge1$ and $a>2n+2\iota$. For
$n=2$ the inequality $n\ge2\kappa+1$ forces $\kappa=0$, hence $D=1+\iota\le2$.
\end{proof}

\begin{remark}[Where Theorem~\ref{thm:rank-singleton} improves on the erasure bounds]
\label{rem:bound-regime}
Steps~1--3 of the proof of Theorem~\ref{thm:rank-singleton} use no relation between $a$ and $n$
and already give
\begin{equation}
\label{eq:rank-singleton-untransposed}
K\ \le\ an-n\bigl(d^{R}_X+d^{R}_Z-2\bigr),
\end{equation}
which for $d^{R}_X=d^{R}_Z$ is \eqref{eq:row-erasure}; The transposition step then replaces the factor $n$ by
$M=\max(a,n)$ by transposing the labels, a step available because both the rank metric and the
standard pairing are transpose invariant and with no counterpart in the Hamming metric. Two distinct mechanisms distinguish \eqref{eq:rank-singleton} from \eqref{eq:row-erasure} and warrant clear distinction. Writing $D:=\min(d^{R}_X,d^{R}_Z)$ and
$\Delta:=|d^{R}_X-d^{R}_Z|$, so that $d^{R}_X+d^{R}_Z-2=2(D-1)+\Delta$, the right-hand sides
satisfy
\begin{equation}
\label{eq:bound-gap}
\begin{split}
&\bigl[an-2n(D-1)\bigr]-\bigl[an-M(d^{R}_X+d^{R}_Z-2)\bigr]\\
&\qquad=2(M-n)(D-1)+M\Delta .
\end{split}
\end{equation}
The term $M\Delta$ is the gain from using the two sector distances separately, with no counterpart
for a general stabilizer code, whose single parameter $D$ does not see the two sectors; it is
positive already for $a\le n$, as at $a=n=5$, $d^{R}_X=3$, $d^{R}_Z=2$, where
\eqref{eq:row-erasure} gives $K\le15$ and \eqref{eq:rank-singleton} gives $K\le10$. The term
$2(M-n)(D-1)$ is the transpose-induced gain, positive exactly when $a>n$ and $D\ge2$, and part of
it is shared by every stabilizer code: for $a>2n+2\iota$ and $D\ge3$ the column bound \eqref{eq:col-erasure} is the sharper general
statement, and in the symmetric case $d^{R}_X=d^{R}_Z=D$ the part of the gain that is specific to
the CSS class is
\begin{equation}
\label{eq:bound-gap-css}
\begin{split}
&(a-2\iota)(n-2\kappa)-\bigl[an-2a(D-1)\bigr]\\
&\qquad=2a(\kappa+\iota)-2\iota(n-2\kappa),
\end{split}
\end{equation}
which equals $a(D-1)$ for odd $D$, against $2(a-n)(D-1)$ in \eqref{eq:bound-gap}. For $n<a\le2n+2\iota$, and for $D\le2$ on every layout, the row bound is the binding one and \eqref{eq:bound-gap} fully accounts for the gap; the two erasure bounds take the same value exactly when $\kappa=0$ or $a=2n+2\iota$, and otherwise the column bound is strictly weaker, as at $a=5$, $n=3$, $D=3$, where they evaluate to $3$ and $5$, respectively. In all cases \eqref{eq:rank-singleton} implies both
erasure bounds for a CSS pair, the difference at $d^{R}_X=d^{R}_Z=D$ being
\eqref{eq:bound-gap-css} for $a\ge n$ and $2\kappa(2n-a+2\iota)$ for $a<n$, both nonnegative.

The smallest instance of the transpose-induced gain is $a=3$, $n=2$, $d^{R}_X=d^{R}_Z=2$, where
both erasure bounds allow $K\le2$ whereas \eqref{eq:rank-singleton} gives $K\le0$; an exhaustive
enumeration over $\F_2^{3\times2}$ confirms it and finds every optimal pair to have the structure
of Theorem~\ref{thm:singleton-equality-mrd}.
\end{remark}

\begin{proposition}[The bound \eqref{eq:rank-singleton} fails for general stabilizer codes]
\label{prop:css-specific}
The following three statements hold.
\begin{enumerate}
\item[(i)] On the $3\times2$ layout over $\F_2$ there is a stabilizer code with $K=2$ and $D=2$,
generated by the four operators below, each written row by row with two carriers per row:
\[
\begin{array}{ll}
\mathrm{IY}\ \mathrm{ZX}\ \mathrm{XX}, &
\mathrm{XI}\ \mathrm{XX}\ \mathrm{YI},\\
\mathrm{ZY}\ \mathrm{XZ}\ \mathrm{IX}, &
\mathrm{YZ}\ \mathrm{XY}\ \mathrm{IY}.
\end{array}
\]
They commute pairwise and are independent, $|N(S)|=2^{8}$, and the code attains
\eqref{eq:row-erasure} and \eqref{eq:col-erasure} with equality, $6-2\cdot2\cdot1=2=K$. By
Theorem~\ref{thm:rank-singleton} every CSS pair on that layout with $K>0$ has
$\min(d^{R}_X,d^{R}_Z)=1$, hence $D=1$.
\item[(ii)] The codes of \cite{nizuka2026} occupy the $2m\times m$ layout with $K=2m(m-k)$ and
$D=k+1$ for $1\le k<m$, in the orientation in which the $2m$ coordinates of the Hermitian
self-orthogonal Gabidulin code are the layers, that is, the rows of the label matrix, and the $m$
symbols carried by each coordinate are the carriers. They therefore attain \eqref{eq:row-erasure}
with equality, whereas $an-\max(a,n)(2D-2)=2m(m-2k)<K$ for every $k\ge1$. For the instance $m=2$,
$k=1$ the four generators listed in \cite{nizuka2026} are pairwise commuting and independent, and
direct computation gives $|N(S)|=2^{12}$, $K=4$ and $D=2$ on the $4\times2$ layout, against a
value of $0$ for the right-hand side of \eqref{eq:rank-singleton}; for $m=3$, $k=2$ the
construction gives $|N(S)|=2^{24}$, $K=6$ and $D=3$ on the $6\times3$ layout. The stacked rank is
not transpose invariant, and the layout orientation is critical: when evaluated on the $m\times2m$ layout, the same labels of the instance $m=2$, $k=1$ have $D=1$.
\item[(iii)] Consequently a stabilizer code whose parameters violate
$K\le an-\max(a,n)(2D-2)$ is not CSS, and it is not carried to a CSS code on the same layout by
any relabeling that preserves $K$ and the \emph{stacked} rank $\rank(A\mid B)$ of every label:
the image would be a CSS pair with $d^{R}_X,d^{R}_Z\ge D$ on the same layout, which
contradicts Theorem~\ref{thm:rank-singleton}. The class of admissible relabelings must be defined rigorously: a change of the intra-carrier basis preserves the rank of each sector label separately (Corollary~\ref{cor:rank-invariance}) but
not in general the stacked rank, acting in the embedding of
Proposition~\ref{prop:sector-split-embedding} as $(A,B)\mapsto(T^{-\top}A,\,TB)$ with different
left factors on the two sectors. For $q=2$ and the $3\times3$ matrix $T=\mathrm{diag}(T_0,1)$ with
$T_0=\bigl(\begin{smallmatrix}1&1\\0&1\end{smallmatrix}\bigr)$, the label
$A=B=(1,0,0)^{\top}$ has $\rank(A\mid B)=1$ while $\rank(T^{-\top}A\mid TB)=2$, and applying this
carrier-uniform change to the code of (i) leaves both sector ranks unchanged while lowering its
stacked rank distance from $2$ to $1$. Exhaustively, of the $168$ elements of
$\mathrm{GL}_3(\F_2)$ exactly the six with $T^{-\top}=T$, the permutation matrices, preserve $D=2$
for that code; the remaining $162$ lower it to $1$. The stacked rank distance of a non-CSS code is therefore a property of the code \emph{together
with} the intra-carrier basis, whereas the sector rank distances are not; the row-wise conjugation
by a Clifford applied identically to every row does preserve it, acting on the joint label by one
invertible right multiplication \cite{gottesman1997,ketkar2006}.
\end{enumerate}
\end{proposition}

\begin{proof}
(i) Commutation and independence of the four generators are a direct computation, whence
$|N(S)|=2^{2\cdot6-4}=2^{8}$ and $K=6-4=2$. The condition $D\ge2$ is equivalent to asserting that no label of stacked rank one belongs to $N(S)\setminus S$; the labels of stacked rank one are the $7\cdot15=105$ matrices
$uc^{\top}$ with $0\neq u\in\F_2^{3}$ and $0\neq c\in\F_2^{4}$, and direct verification against the
four generators confirms that none lies in $N(S)\setminus S$. Since there exists a label of stacked rank two in
$N(S)\setminus S$, we obtain $D=2$. The CSS statement is Theorem~\ref{thm:rank-singleton} at $a=3$,
$n=2$: a CSS pair with $K>0$ and $d^{R}_X,d^{R}_Z\ge2$ would give $K\le6-3\cdot2=0$.

(ii) The identity $2m(m-k)=2m\cdot m-2m\bigl((k+1)-1\bigr)$ follows by elementary algebra, and
$2m\cdot m-\max(2m,m)\bigl(2(k+1)-2\bigr)=2m(m-2k)$. The parameters of the family for general
$(m,k)$ are those reported in \cite{nizuka2026}; the two instances quoted were recomputed from the
construction described there, by forming the Hermitian self-orthogonal Gabidulin code over
$\F_{2^{2m}}$ in a self-dual basis and applying the symplectic realization map, and by enumerating
all labels of stacked rank at most $D-1$ for the resulting binary stabilizer group. The
orientation claim for $m=2$, $k=1$ is the same enumeration performed on the transposed labels.

(iii) The first assertion is Theorem~\ref{thm:rank-singleton} applied to the image pair, a
relabeling preserving the stacked rank leaving $D$, and hence the sector distances of a CSS image,
at least $D$. The displayed $3\times3$ computation is straightforward, and the resulting effect on the code in (i), together with the exhaustive count over $\mathrm{GL}_3(\F_2)$, is verified computationally.
\end{proof}

The comparison becomes exact when the layout and the logical dimension are held fixed. The
following lemma supplies the general-stabilizer side on tall layouts.

\begin{lemma}[Row-block stacking]
\label{lem:row-stacking}
Let $Q_0$ be a stabilizer code on the $a_0\times n$ layout with $K_0>0$ logical qudits and stacked
rank distance $D_0$, and let $\ell\ge1$. Placing $\ell$ independent copies of $Q_0$ on disjoint
blocks of $a_0$ rows gives a stabilizer code on the $\ell a_0\times n$ layout with $K=\ell K_0$ and
stacked rank distance $D=D_0$.
\end{lemma}

\begin{proof}
The blocks occupy disjoint qudits, so the stabilizer group is the direct product of the blockwise
groups, $\dim S=\ell\dim S_0$ and $K=\ell K_0$; moreover a label lies in $N(S)$ if and only if each
of its block restrictions lies in the corresponding $N(S_i)$, and in $S$ if and only if each lies
in $S_i$. Hence a label of $N(S)\setminus S$ has, for some $i$, a block restriction in
$N(S_i)\setminus S_i$, whose stacked rank is at least $D_0$; deleting rows cannot increase rank, so
that restriction is a submatrix of the full label and $\rank(A\mid B)\ge D_0$. Conversely,
extending a minimum-rank element of $N(S_1)\setminus S_1$ by zero rows gives an element of
$N(S)\setminus S$ of stacked rank exactly $D_0$.
\end{proof}

\begin{corollary}[Exact optima at a fixed layout and logical dimension]
\label{cor:css-penalty-fixed}
Let $q=2$, let $n\ge2$ and $1\le k<n$, let $\ell\ge1$, and consider the $a\times n$ layout with
$a=2\ell n$, $N=an$ physical qubits and logical dimension $K=a(n-k)$. Then
\begin{enumerate}
\item[(i)] the largest stacked rank distance of a stabilizer code with these parameters is exactly
$k+1$;
\item[(ii)] the largest stacked rank distance of a CSS pair with these parameters is exactly
$\lfloor k/2\rfloor+1$;
\item[(iii)] for even $k=2\kappa$ and $\ell\ge2$ the optimal stabilizer codes attain the column
bound \eqref{eq:col-erasure} with equality, in the regime $a>2n$ where that bound is strictly
stronger than the row bound \eqref{eq:row-erasure}; the smallest instance is
$(a,n,K,D)=(12,3,12,3)$, where \eqref{eq:row-erasure} allows $K\le24$ and \eqref{eq:col-erasure}
allows $K\le12$.
\end{enumerate}
At a fixed layout, a fixed number of physical qubits and a fixed logical dimension, departing from the CSS
framework therefore permits a stacked rank distance larger by exactly $\lceil k/2\rceil$.
\end{corollary}

\begin{proof}
(i) \emph{Upper bound.} Let $D'$ be the stacked rank distance of such a code and write
$D'-1=2\kappa'+\iota'$ with $\iota'\in\{0,1\}$. Proposition~\ref{prop:erasure-bounds} gives
$n\ge2\kappa'+1$ and $K\le(a-2\iota')(n-2\kappa')$. For $\iota'=0$ this reads
$a(n-k)\le a(n-2\kappa')$, so $2\kappa'\le k$ and $D'=2\kappa'+1\le k+1$. For $\iota'=1$ it reads
$a(n-k)\le(a-2)(n-2\kappa')$, that is $a(2\kappa'-k)+2(n-2\kappa')\le0$; since $n-2\kappa'>0$ this
forces $2\kappa'<k$, hence $D'=2\kappa'+2\le k+1$. \emph{Attainability.} By
Proposition~\ref{prop:css-specific}(ii) the codes of \cite{nizuka2026} occupy the $2n\times n$
layout with $K_0=2n(n-k)$ and $D_0=k+1$; stacking $\ell$ of them on disjoint row blocks gives, by
Lemma~\ref{lem:row-stacking}, a stabilizer code on the $a\times n$ layout with
$K=\ell\cdot2n(n-k)=a(n-k)$ and $D=k+1$.

(ii) Since $a=2\ell n>n$, Theorem~\ref{thm:rank-singleton} reads
$a(n-k)=K\le an-a(d^{R}_X+d^{R}_Z-2)$, whence $d^{R}_X+d^{R}_Z\le k+2$ and
$\min(d^{R}_X,d^{R}_Z)\le\lfloor(k+2)/2\rfloor=\lfloor k/2\rfloor+1$. Conversely put
$u:=\lfloor k/2\rfloor$ and $v:=\lceil k/2\rceil$, so that $u+v=k\le n-1=L-1$;
Theorem~\ref{thm:optimal-existence}, whose proof is independent of the present corollary, then
provides a CSS pair on the $a\times n$ layout with $K=M(L-u-v)=a(n-k)$, $d^{R}_X=u+1$ and
$d^{R}_Z=v+1$, whose stacked rank distance is $\min(d^{R}_X,d^{R}_Z)=\lfloor k/2\rfloor+1$ by
Section~\ref{subsec:matrix-css}.

(iii) At $k=2\kappa$ one has $D=k+1$ odd, so $\iota=0$ and \eqref{eq:col-erasure} reads
$K\le a(n-2\kappa)=a(n-k)$, an equality for the codes of~(i). That the column bound is then
strictly stronger than the row bound exactly for $a>2n$ is Proposition~\ref{prop:erasure-bounds},
and $a=2\ell n>2n$ for $\ell\ge2$. The displayed instance is $n=3$, $k=2$, $\ell=2$, the base code
being the $6\times3$ instance with $K_0=6$ and $D_0=3$ of
Proposition~\ref{prop:css-specific}(ii).
\end{proof}

The statement is a separation in the distance--dimension trade-off and not a claim of hardware
advantage: fixing $N$, $K$ and the layout does not fix the check weights, the extraction depth or
the gate error rates.
The contrast with the Hamming metric concerns the \emph{form} of the bound: there the asymmetric
bound carries no factor depending on the orientation of the layout
\cite{sarvepalli2009asymmetric,ezerman2013aqmds}, whereas here the erasure bounds of
Proposition~\ref{prop:erasure-bounds} are not transpose symmetric, a carrier erasure costing a
factor two in stacked rank because every carrier contributes two columns to $(A\mid B)$, while a
CSS pair is transpose symmetric and incurs the penalty $\max(a,n)(d^{R}_X+d^{R}_Z-2)$ in both
orientations. In the symmetric case $d^{R}_X=d^{R}_Z=D$ the resulting cost of CSS structure is
$2(a-n)(D-1)$ for $n<a\le2n+2\iota$ and \eqref{eq:bound-gap-css} for $a>2n+2\iota$, in particular
$a(D-1)$ for odd $D$, and Proposition~\ref{prop:css-specific} shows that it is incurred already on
the $3\times2$ and $2m\times m$ layouts, on which the row bound is the binding erasure bound.

\section{The Structure of Singleton-Optimal Pairs}
\label{sec:equality}

Equality in \eqref{eq:rank-singleton} reflects a fundamental structural rigidity: it implies that both check spaces must be maximum-rank-distance, ensures achievability for every admissible parameter triple, and guarantees
purity. The three statements of this section are used throughout
Sections~\ref{sec:interfaces}--\ref{sec:witnesses}.

\begin{theorem}[Singleton equality forces MRD check spaces]
\label{thm:singleton-equality-mrd}
Let $S_X,S_Z\subseteq\F_q^{a\times n}$ be as in Theorem~\ref{thm:rank-singleton}, with $K>0$, and
write $M:=\max(a,n)$ and $L:=\min(a,n)$.
\begin{enumerate}
\item[(i)] \emph{(Necessity.)} If equality holds in \eqref{eq:rank-singleton}, then, with
$u:=d^{R}_X-1$ and $v:=d^{R}_Z-1$,
\begin{equation}
\label{eq:mrd-equality}
\dim_{\F_q}S_X=Mv,\qquad \dim_{\F_q}S_Z=Mu,
\end{equation}
and each nonzero check space is maximum-rank-distance, with
\begin{equation}
\label{eq:mrd-distances}
d_R(S_X)=L-v+1,\qquad d_R(S_Z)=L-u+1;
\end{equation}
the case $v=0$ reads $S_X=\{0\}$ and $u=0$ reads $S_Z=\{0\}$.
\item[(ii)] \emph{(Sufficiency.)} Conversely, let $u,v\ge0$ be integers with $u+v<L$ and suppose
that $\dim_{\F_q}S_X=Mv$, $\dim_{\F_q}S_Z=Mu$, and that each nonzero check space is
maximum-rank-distance, so that \eqref{eq:mrd-distances} holds. Then $d^{R}_X=u+1$, $d^{R}_Z=v+1$,
$K=M(L-u-v)>0$, and equality holds in \eqref{eq:rank-singleton}.
\end{enumerate}
\end{theorem}

\begin{proof}
By the transposition step of Theorem~\ref{thm:rank-singleton} assume $a\le n$, so $M=n$, $L=a$;
transposition preserves \eqref{eq:mrd-equality} and \eqref{eq:mrd-distances}.

\emph{Necessity of \eqref{eq:mrd-equality}.} Equality means $\dim S_X+\dim S_Z=n(u+v)$ and, by
Theorem~\ref{thm:rank-singleton}, \eqref{eq:singleton-equality} holds for every disjoint pair
$A,B$ with $|A|=u$, $|B|=v$. Fix $B$; since $K>0$ we have $u+v<a$, so $a-v>u$ and every index
outside $B$ is avoided by some admissible $A$, whence $\bigcap_AU_A=\{0\}$ and intersecting
\eqref{eq:singleton-equality} over all such $A$ gives $S_X\cap\ker\pi_B=\{0\}$, so $\pi_B$ is
injective on $S_X$ and $\dim S_X\le vn$. By symmetry, $\dim S_Z\le un$, and adding this to
$\dim S_X+\dim S_Z=n(u+v)$ forces both to be equalities.

\emph{Necessity of \eqref{eq:mrd-distances}.} For $T\in\mathrm{GL}_a(\F_q)$ the assignment $(S_X,S_Z)\mapsto(TS_X,T^{-\top}S_Z)$ preserves the
standard pairing, since $\mathrm{tr}((TX)^{\top}T^{-\top}Z)=\mathrm{tr}(X^{\top}Z)$, preserves
rank and hence $K$, and preserves both sector distances, because
$(T^{-\top}S_Z)^{\pstd}=T\,S_Z^{\pstd}$ and $(TS_X)^{\pstd}=T^{-\top}S_X^{\pstd}$ carry the two
logical sets to $T(S_Z^{\pstd}\setminus S_X)$ and $T^{-\top}(S_X^{\pstd}\setminus S_Z)$. The
transformed pair therefore also attains equality, so $TS_X\cap\ker\pi_B=\{0\}$ for every $T$ and
every $B$ with $|B|=v$.

If some nonzero $X\in S_X$ had $\rank X\le a-v$, its column space, of dimension at most $a-v$,
could be carried by some $T\in\mathrm{GL}_a(\F_q)$ into the span of the coordinate vectors indexed
by $\{1,\dots,a\}\setminus B$; then $0\neq TX\in TS_X\cap\ker\pi_B$, a contradiction. Hence
$d_R(S_X)\ge a-v+1$ whenever $S_X\neq\{0\}$. Combining with \eqref{eq:matrix-singleton}, which gives
$nv=\dim S_X\le n\bigl(a-d_R(S_X)+1\bigr)$ and hence $d_R(S_X)\le a-v+1$, yields
$d_R(S_X)=a-v+1=L-v+1$ and equality in \eqref{eq:matrix-singleton}, that is, $S_X$ is MRD. By
symmetry, $d_R(S_Z)=L-u+1$ and the MRD property of $S_Z$ follows. If $v=0$ then
\eqref{eq:mrd-equality} reads $\dim S_X=0$, and dually for $u=0$.

\emph{Sufficiency.} Assume the hypotheses of (ii). Then $K=an-M(u+v)=M(L-u-v)$, which is positive
because $u+v<L$. If $u>0$ then $S_Z$ is MRD of dimension $Mu$, so the Delsarte dual of an MRD code
being MRD \cite{gabidulin1985,delsarte1978}, $S_Z^{\pstd}$ is MRD of dimension $M(L-u)$ with
$d_R(S_Z^{\pstd})=u+1$; if $u=0$ then $S_Z=\{0\}$, $S_Z^{\pstd}$ is the whole space and
$d_R(S_Z^{\pstd})=1=u+1$. If $v>0$ then $S_X$ is MRD with $d_R(S_X)=L-v+1>u+1$, the inequality
because $u+v<L$, so no minimum-rank element of $S_Z^{\pstd}$ lies in $S_X$ and $d^{R}_X=u+1$; if $v=0$ then $S_X=\{0\}$ and again $d^{R}_X=u+1$. Symmetrically, $d^{R}_Z=v+1$, whence
$an-M(d^{R}_X+d^{R}_Z-2)=an-M(u+v)=K$.
\end{proof}

Theorem~\ref{thm:singleton-equality-mrd} characterizes optimality but leaves open which parameter
triples occur. They all do, over every field and on every layout, and the witnesses are nested
Gabidulin codes expanded in a pair of trace-dual bases.

\begin{theorem}[The Singleton-optimal parameter set]
\label{thm:optimal-existence}
Let $q$ be a prime power, let $a,n\ge1$, $M:=\max(a,n)$, $L:=\min(a,n)$, and let $u,v\ge0$ be
integers. The following are equivalent.
\begin{enumerate}
\item[(i)] $u+v\le L-1$.
\item[(ii)] There is a CSS pair $S_Z\subseteq S_X^{\pstd}$ in $\F_q^{a\times n}$ with $K>0$,
$d^{R}_X=u+1$ and $d^{R}_Z=v+1$ attaining equality in \eqref{eq:rank-singleton}.
\end{enumerate}
Every such pair has $K=M(L-u-v)$, so the set of triples $(K,d^{R}_X,d^{R}_Z)$ attained with
equality is exactly $\bigl\{\bigl(M(L-u-v),\,u+1,\,v+1\bigr):u,v\ge0,\ u+v\le L-1\bigr\}$. A pair
may be taken $\F_{q^{M}}$-linear along the longer side: for $a\le n$ let
$\boldsymbol\alpha\in\F_{q^n}^{\,a}$ have $\F_q$-independent entries, let
$0\neq\boldsymbol h\in\F_{q^n}^{\,a}$ satisfy $\sum_ih_i\alpha_i^{q^{j}}=0$ for $0\le j\le a-2$,
put $\boldsymbol\gamma:=\boldsymbol h^{\,q^{-(a-v-1)}}$ componentwise, and let $B$ be an
$\F_q$-basis of $\F_{q^n}$ with trace-dual basis $B^{\ast}$. Then the matrix codes
\begin{equation}
\label{eq:construction}
\begin{split}
S_X&:=\bigl\{[\mathrm{coord}_B f(\alpha_i)]_{i\le a}:f\in\mathcal L_{<v}\bigr\},\\
S_Z&:=\bigl\{[\mathrm{coord}_{B^{\ast}}g(\gamma_i)]_{i\le a}:g\in\mathcal L_{<u}\bigr\},
\end{split}
\end{equation}
whose $i$-th rows are the coordinate vectors of the $i$-th code symbols, form such a pair, with
$S_X=\{0\}$ when $v=0$ and $S_Z=\{0\}$ when $u=0$; for $a>n$, the corresponding pair is obtained by matrix transposition. Consequently, trace-orthogonal Singleton-optimal check spans in $\Fqm^n$ exist for every prime power $q$, every $m\le n$ and every admissible $(u,v)$.
\end{theorem}

\begin{proof}
(ii)$\Rightarrow$(i): equality in \eqref{eq:rank-singleton} with $K>0$ reads $M(L-u-v)=K>0$ by
Theorem~\ref{thm:singleton-equality-mrd}.

(i)$\Rightarrow$(ii): assume $a\le n$ and $u,v\ge1$, the cases $u=0$ or $v=0$ being obtained by
replacing the corresponding space by $\{0\}$ and the case $a>n$ by transposition, which preserves
rank, the standard pairing, $K$ and both sector distances
(Theorem~\ref{thm:rank-singleton}). Identify $\F_q^{a\times n}$ with $\F_{q^n}^{\,a}$ by reading
the $i$-th row of a matrix as the $B$-coordinate vector of a symbol $x_i\in\F_{q^n}$; then
$\rank X=\dim_{\F_q}\langle x_1,\dots,x_a\rangle_{\F_q}$, and the same holds for the basis
$B^{\ast}$, rank being independent of the basis (Corollary~\ref{cor:rank-invariance}).

Write $\mathrm{Gab}(\boldsymbol\alpha,k):=\{(f(\alpha_i))_i:f\in\mathcal L_{<k}\}$, of
$\F_q$-dimension $nk$ and minimum rank distance $a-k+1$ \cite{gabidulin1985}. The dot product
$x\cdot y:=\sum_ix_iy_i$ on $\F_{q^n}^{\,a}$ satisfies
$\langle X,Y\rangle_{\mathrm{std}}=\sum_i\Tr(x_iy_i)=\Tr(x\cdot y)$ when $X$ carries
$B$-coordinates and $Y$ carries $B^{\ast}$-coordinates, by the definition of the trace-dual basis;
so a dot-orthogonal pair of $\F_{q^n}$-linear codes yields a pair of matrix codes orthogonal for
$\langle\cdot,\cdot\rangle_{\mathrm{std}}$.

\emph{The vector $\boldsymbol h$ and its shifts.} The $a-1$ conditions
$\sum_ih_i\alpha_i^{q^{j}}=0$, $0\le j\le a-2$, signify that $\boldsymbol h$ is dot-orthogonal to
$\mathrm{Gab}(\boldsymbol\alpha,a-1)$, whose dot-orthogonal space has $\F_{q^n}$-dimension one. That
space is the Delsarte dual of an MRD code of $\F_q$-dimension $n(a-1)$ and minimum rank distance
$2$, hence is MRD of $\F_q$-dimension $n$ and minimum rank distance $a$
\cite{gabidulin1985,delsarte1978}; its nonzero elements therefore have rank weight $a$, so the
entries of $\boldsymbol h$ are $\F_q$-independent and the $a$ Frobenius shifts
$\boldsymbol h^{\,q^{t}}$, $t\in\mathbb Z/n\mathbb Z$ read modulo the order of the Frobenius, form
an $\F_{q^n}$-basis of $\F_{q^n}^{\,a}$ when $t$ ranges over any $a$ consecutive values, the
corresponding matrix being a Moore matrix of a vector with $\F_q$-independent entries
\cite{gabidulin1985}. For $v-a+1\le t\le0$ and $0\le j\le v-1$ one has $t\le j\le t+a-2$, so
\[
\sum_ih_i^{q^{t}}\alpha_i^{q^{j}}
=\Bigl(\sum_ih_i\alpha_i^{q^{j-t}}\Bigr)^{q^{t}}=0 ,
\]
and the $a-v$ vectors $\boldsymbol h^{\,q^{t}}$, $v-a+1\le t\le0$, are dot-orthogonal to
$\mathrm{Gab}(\boldsymbol\alpha,v)$ and $\F_{q^n}$-independent. Their span is
$\mathrm{Gab}(\boldsymbol\gamma,a-v)$ with $\boldsymbol\gamma$ as stated, and comparing dimensions
with $\dim_{\F_q}\mathrm{Gab}(\boldsymbol\alpha,v)^{\perp}=n(a-v)$ gives
$\mathrm{Gab}(\boldsymbol\alpha,v)^{\perp}=\mathrm{Gab}(\boldsymbol\gamma,a-v)$.

\emph{Conclusion.} Since $u+v\le a-1$ we have $u\le a-v$, so
$S_Z=\mathrm{Gab}(\boldsymbol\gamma,u)\subseteq\mathrm{Gab}(\boldsymbol\gamma,a-v)$ is
dot-orthogonal to $S_X=\mathrm{Gab}(\boldsymbol\alpha,v)$ and hence
$S_Z\subseteq S_X^{\pstd}$ in the chosen coordinates. Both are MRD with
$\dim_{\F_q}S_X=nv=Mv$, $\dim_{\F_q}S_Z=nu=Mu$, $d_R(S_X)=a-v+1=L-v+1$ and $d_R(S_Z)=L-u+1$, so
Theorem~\ref{thm:singleton-equality-mrd}(ii) gives $d^{R}_X=u+1$, $d^{R}_Z=v+1$,
$K=M(L-u-v)>0$ and equality in \eqref{eq:rank-singleton}. For the last assertion take $a=m$ and
apply Lemma~\ref{lem:gram-reduction} to the pair $(G^{-1}S_X,S_Z)$, where $G$ is the trace Gram
matrix of an $\F_q$-basis of $\Fqm$: it is trace-orthogonal in $\Fqm^n$ and has the same logical
dimension and the same sector rank distances.
\end{proof}

\begin{corollary}[Singleton-optimal pairs are pure]
\label{cor:purity}
If a CSS pair with $K>0$ attains equality in \eqref{eq:rank-singleton}, then
$d_R(S_Z^{\pstd})=d^{R}_X$ and $d_R(S_X^{\pstd})=d^{R}_Z$. In particular a degenerate code, one
with a stabilizer label of rank weight below its sector distance, never attains the bound, although
the bound is valid for degenerate codes.
\end{corollary}

\begin{proof}
Write $u=d^{R}_X-1$. By Theorem~\ref{thm:singleton-equality-mrd}(i), $S_Z$ is either $\{0\}$ or
MRD of $\F_q$-dimension $Mu$ with $d_R(S_Z)=L-u+1$. In the first case $S_Z^{\pstd}$ is the whole
space, whose minimum rank weight is $1=u+1$. In the second, the Delsarte dual of an MRD code is
MRD \cite{gabidulin1985,delsarte1978}, so $S_Z^{\pstd}$ is MRD of $\F_q$-dimension $M(L-u)$ and minimum rank distance $u+1$. In either case, $d_R(S_Z^{\pstd})=u+1=d^{R}_X$, and since
$d^{R}_X$ is a minimum of $\rank$ over the subset $S_Z^{\pstd}\setminus S_X$ of $S_Z^{\pstd}\setminus\{0\}$, the
two values coincide. Symmetrically, the second identity holds, and a degenerate pair has $d_R(S_X)<d^{R}_X\le d_R(S_Z^{\pstd})$ with $S_X\subseteq S_Z^{\pstd}$, a contradiction.
\end{proof}

Purity is therefore a consequence of optimality and not a hypothesis. In the Hamming metric, the analogous correspondence between optimal asymmetric CSS codes and nested maximum-distance-separable pairs is usually stated for pure codes \cite{sarvepalli2009asymmetric,ezerman2013aqmds}, and the classification of the attainable
parameters there is conditional on the classical MDS conjecture, whereas Theorem~\ref{thm:optimal-existence} is unconditional, MRD codes of every admissible dimension existing over every finite field.

The next three lemmas are the rigidity properties of MRD codes that
Section~\ref{sec:interfaces} uses. The first two are stated for $a\le n$, where a full-rank
codeword is available; the third holds on every layout.

\begin{lemma}[MRD codes contain a full-rank codeword]
\label{lem:mrd-fullrank}
Let $C\subseteq\F_q^{a\times n}$ with $a\le n$ be $\F_q$-linear, nonzero and MRD, so that
$\dim_{\F_q}C=n\bigl(a-d_R(C)+1\bigr)$. Then $C$ contains a matrix of rank $a$.
\end{lemma}

\begin{proof}
Write $d:=d_R(C)\le a$ and let $H\subseteq\F_q^{a}$ be a hyperplane. Every $T\in\mathrm{GL}_a(\F_q)$
carrying $H$ to the span of the first $a-1$ coordinate vectors maps
$C_H:=\{X\in C:\mathrm{colspace}(X)\subseteq H\}$ into the matrices with vanishing last row,
preserving rank; the image is an $\F_q$-subspace of $\F_q^{(a-1)\times n}$ of minimum rank distance
at least $d$, so $\dim_{\F_q}C_H\le n(a-d)$ by \eqref{eq:matrix-singleton}, the right-hand side
being $0$ when $d=a$. A matrix of rank less than $a$ has its column space inside some hyperplane,
and there are $(q^{a}-1)/(q-1)$ of them, so the number of codewords of rank less than $a$ is at
most $\frac{q^{a}-1}{q-1}\,q^{n(a-d)}<q^{a}q^{n(a-d)}\le q^{n(a-d+1)}=|C|$, using $a\le n$. Some
codeword therefore has rank $a$.
\end{proof}

\begin{lemma}[Non-invariance of MRD codes with $a\le n$ under nontrivial idempotents]
\label{lem:mrd-noninvariance}
Let $a\le n$ and let $C\subseteq\F_q^{a\times n}$ be a nonzero $\F_q$-linear MRD code with
$2\le d_R(C)\le a$. Then $PC\not\subseteq C$ for every idempotent $P\in\F_q^{a\times a}$ with
$P\notin\{0,I_a\}$. Equivalently, for labels in $\Fqm^{n}$ with $m\le n$ and an $\F_q$-linear
idempotent $e$ on $\Fqm$ with $1\le\rank(e)\le m-1$, self-adjoint or not, one has
$e\star C\not\subseteq C$ for every $\F_q$-subspace $C\subseteq\Fqm^{n}$ that is MRD as an
$m\times n$ matrix code with $2\le d_R(C)\le m$.
\end{lemma}

\begin{proof}
Suppose $PC\subseteq C$. Then $(I_a-P)C\subseteq C$ as well, and for $Y\in C^{\pstd}$ and $X\in C$
one has $\langle P^{\top}Y,X\rangle_{\mathrm{std}}=\mathrm{tr}(Y^{\top}PX)
=\langle Y,PX\rangle_{\mathrm{std}}=0$, so $P^{\top}C^{\pstd}\subseteq C^{\pstd}$ and likewise for
$(I_a-P)^{\top}$. Since $\F_q^{a}=\Imag P\oplus\Imag(I_a-P)$ we have
$\rank P+\rank(I_a-P)=a$; let $Q\in\{P,I_a-P\}$ have the smaller rank $s$, so that
$1\le s\le\lfloor a/2\rfloor$, the lower bound because $P\notin\{0,I_a\}$.

Write $d:=d_R(C)$, so $\dim_{\F_q}C=n(a-d+1)$ and, by Delsarte duality
\cite{gabidulin1985,delsarte1978}, $C^{\pstd}$ is MRD of dimension $n(d-1)>0$ with
$d_R(C)+d_R(C^{\pstd})=a+2$. By Lemma~\ref{lem:mrd-fullrank}, where
the hypothesis $a\le n$ is invoked, there are $X_0\in C$ and $Y_0\in C^{\pstd}$ of rank $a$.
The column space of $X_0$ is all of $\F_q^{a}$, so $\rank(QX_0)=\dim Q(\F_q^{a})=s$, and likewise
$\rank(Q^{\top}Y_0)=\rank Q^{\top}=s$. As $QX_0\in C$ and $Q^{\top}Y_0\in C^{\pstd}$ are nonzero,
$d_R(C)\le s$ and $d_R(C^{\pstd})\le s$, whence $a+2\le2s\le a$, a contradiction.

For the second formulation, fix an $\F_q$-basis of $\Fqm$ and let $P$ be the matrix of $e$ acting
on coordinate columns; the componentwise action $e\star\cdot$ is left multiplication by $P$, which
is completely equivalent to the matrix statement, rank and the MRD property being
independent of the basis by Corollary~\ref{cor:rank-invariance}.
\end{proof}

\begin{lemma}[MRD codes are spanned by their minimum-rank codewords]
\label{lem:mrd-generation}
Let $C\subseteq\F_q^{a\times n}$ be $\F_q$-linear, nonzero and MRD on an arbitrary layout, and
put $d:=d_R(C)$. Then $C=\mathrm{span}_{\F_q}\{X\in C:\rank X=d\}$.
\end{lemma}

\begin{proof}
Assume first $a\le n$ and fix a row set $A$ with $|A|=d-1$, possible since $d\le a$. For $j\notin A$, let $U_j$ be the matrices vanishing outside $A\cup\{j\}$ and $C_j:=C\cap U_j$. Every nonzero element of $C_j$ has rank at most $d$, since it is supported on the $d$ rows in $A\cup\{j\}$, and at least $d$, since it lies in $C$; hence it has rank exactly $d$. The projection of $C_j$ onto row $j$ is injective, a kernel element being
supported on $A$ and of rank at most $d-1$, so $\dim C_j\le n$, while
\[
\dim C_j\ge\dim C+\dim U_j-an=n(a-d+1)+dn-an=n ,
\]
whence $\dim C_j=n$. The sum over $j\notin A$ is direct: if $\sum_jX_j=0$ with $X_j\in C_j$, then
for each $j_0\notin A$ only $X_{j_0}$ has a possibly nonzero row $j_0$, which therefore vanishes,
giving $X_{j_0}=0$. Hence $\dim\bigoplus_{j\notin A}C_j=n(a-d+1)=\dim C$, so $C$ is the span of the
$C_j$, each generated by codewords of rank $d$. For $a>n$, we apply the result established for $a\le n$ to $C^{\top}\subseteq\F_q^{n\times a}$, which is $\F_q$-linear, nonzero, and MRD with the same minimum rank distance, since rank and \eqref{eq:matrix-singleton} are transpose-invariant; transposing back a spanning set of
minimum-rank codewords of $C^{\top}$ yields one for $C$.
\end{proof}

\begin{remark}
\label{rem:idealiser}
Lemma~\ref{lem:mrd-noninvariance} assumes neither a square layout nor self-adjointness. It is
consistent with the classical fact that for $a\le n$ and $d_R(C)>1$ the middle nucleus of an MRD
code $C$, the idealizer acting on its shorter side, is a field
\cite[Th.~5.4]{lunardon2018kernels}, hence contains no idempotent other than $0$ and the identity,
while the corresponding statement for the right nucleus carries further hypotheses
\cite{lunardon2018kernels,csajbok2020idealizers}. Although these algebraic foundations are known in the literature, the proof presented above is self-contained---relying solely on \eqref{eq:matrix-singleton}, Delsarte duality, and Lemma~\ref{lem:mrd-fullrank}---and our primary contribution lies in leveraging this obstruction to establish Theorem~\ref{thm:radius-dichotomy} and Corollary~\ref{cor:optimal-rigidity}. The hypothesis
$a\le n$ cannot be dropped: Proposition~\ref{prop:stacking-sharpness} exhibits a Singleton-optimal
pair on a $6\times3$ layout with an invariant nontrivial idempotent.
\end{remark}

\section{Projected-Recovery Interfaces and Their Exact Radii}
\label{sec:interfaces}

A rank-metric decoder acts through $\F_q$-linear operators on $\Fqm$, which are in general not
$\Fqm$-linear, so $e\star(HE^{\top})\neq H(e\star E)^{\top}$: a decoder designed for a projected
error component cannot in general be supplied by projecting a globally measured syndrome. This section establishes the precise conditions under which a measured check family recovers the projected syndrome or the projected error modulo stabilizer labels, and characterizes the exact recovery radius for each target on Singleton-optimal pairs.

\begin{definition}[Quotient-valued interface]
\label{def:quotient-interface}
Let $T\subseteq\Fqm^n$ be an $\F_q$-subspace of stabilizer labels (acting trivially on the code space) and $\mathcal E\subseteq\Fqm^n$ an error set. The family $H$ admits a \emph{$T$-valued projected interface} on $\mathcal E$ if there is $g:\F_q^s\to\Fqm^n/T$ with $e\star E+T=g\bigl(\sigma_H(E)\bigr)$ for every $E\in\mathcal E$.
\end{definition}

\begin{definition}[The six interface conditions]
\label{def:interfaces}
Let $(S_X,S_Z)$ be trace-orthogonal $\F_q$-linear check spans of a CSS code, let $H$ be the
measured family spanning $S_Z$ with syndrome map $\sigma_H$, and let $e$ be an $\F_q$-linear map on
$\Fqm$ acting componentwise. In the $X$ sector consider the conditions
\begin{itemize}
\item[$(\mathrm I_1)$] the measured syndrome determines the $e$-projected syndrome on the whole
label space;
\item[$(\mathrm I_2)$] it determines it through an $\F_q$-linear map on the rank-$t$ ball, for
some $t\ge1$;
\item[$(\mathrm I_3)$] the measured syndrome determines the projected error modulo $S_X$ on the
whole label space;
\item[$(\mathrm I_4)$] it determines it through an $\F_q$-linear map on the rank-$t$ ball, for
some $t\ge1$;
\item[$(\mathrm I_5)$] it determines the $e$-projected syndrome through an arbitrary function on
the rank-$t$ ball, for some $t\ge1$;
\item[$(\mathrm I_6)$] it determines the projected error modulo $S_X$ through an arbitrary
function on the rank-$t$ ball, for some $t\ge1$;
\end{itemize}
together with the same conditions with the two Pauli sectors exchanged. We write
$\mathrm I_{5,X}(t)$ for $\mathrm I_5$ at the fixed radius $t$ in the $X$ sector, and similarly
for the other conditions and for the $Z$ sector.
The results of this section give
\begin{equation}
\label{eq:interface-hierarchy}
\mathrm I_3\Leftrightarrow\mathrm I_4\Rightarrow\mathrm I_1\Leftrightarrow\mathrm I_2
\Rightarrow\mathrm I_5,\qquad \mathrm I_4\Rightarrow\mathrm I_6\Rightarrow\mathrm I_5,
\end{equation}
where the two equivalences are established in Theorem~\ref{thm:projected-recoverability} with
Corollary~\ref{cor:linearity-dichotomy} and Theorem~\ref{thm:quotient-interface} with
Corollary~\ref{cor:quotient-linearity}, while the implications $\mathrm I_3\Rightarrow\mathrm I_1$ and
$\mathrm I_6\Rightarrow\mathrm I_5$ follow from Lemma~\ref{lem:interface-hierarchy}. All one-way
implications are strict; the witnesses are listed in Table~\ref{tab:interfaces}.
\end{definition}

\subsection{Ambient and relative factorization}
\label{subsec:relative}

The branch decoders of Section~\ref{sec:witnesses} act not on $\Fqm^n$ but on a componentwise
subspace $W=U^n$; the ambient criterion is the case $U=\Fqm$, which we state as part of the same
theorem.

\begin{theorem}[Relative projected-syndrome factorization]
\label{thm:projected-recoverability}
Let $H=(h^{(1)},\dots,h^{(s)})$ be a family of check labels in $\Fqm^n$, let
$S_H:=\mathrm{span}_{\F_q}\{h^{(i)}\}$, and let $\sigma_H$ be the measured trace-syndrome map. Let
$e:\Fqm\to\Fqm$ be $\F_q$-linear with trace adjoint $e^{\dagger}$, let $U\subseteq\Fqm$ be an
$\F_q$-subspace and $W:=U^n$. Then $W^{\ptr}=(U^{\ptr})^n$, and the following are equivalent.
\begin{enumerate}
\item[(a)] There is a function $f:\F_q^s\to\F_q^s$ with
$\sigma_H(e\star E)=f\bigl(\sigma_H(E)\bigr)$ for every $E\in W$.
\item[(b)] There is a matrix $A\in\F_q^{s\times s}$ with $\sigma_H(e\star E)=A\,\sigma_H(E)$ for
every $E\in W$.
\item[(c)] $\ker\sigma_H\cap W\subseteq\ker\bigl(\sigma_H\circ(e\star\cdot)\bigr)$.
\item[(d)] $e^{\dagger}\star S_H\subseteq S_H+W^{\ptr}$.
\end{enumerate}
Moreover $A$ may be taken to be the identity, so that every measured parity coincides with the
corresponding projected parity, if and only if
\begin{equation}
\label{eq:relative-alignment}
e^{\dagger}\star h^{(i)}-h^{(i)}\in W^{\ptr}\qquad(i=1,\dots,s),
\end{equation}
in which case $H$ is called \emph{$W$-relatively aligned}. For $U=\Fqm$ one has $W^{\ptr}=\{0\}$,
so that (d) becomes the \emph{ambient criterion} $e^{\dagger}\star S_H\subseteq S_H$ and
\eqref{eq:relative-alignment} becomes $h^{(i)}\in\Fix(e^{\dagger})$ for every $i$.
\end{theorem}

\begin{proof}
We use the componentwise adjoint identity: for $\F_q$-linear $e$ and any $h,E\in\Fqm^n$,
\begin{equation}
\label{eq:adjoint-star}
\begin{split}
\langle h,e\star E\rangle_{\Tr}
&=\sum_j\Tr\bigl(h_j\,e(E_j)\bigr)\\
&=\sum_j\Tr\bigl(e^{\dagger}(h_j)\,E_j\bigr)
=\langle e^{\dagger}\star h,E\rangle_{\Tr},
\end{split}
\end{equation}
which is \eqref{eq:trace-adjoint} applied in each coordinate and summed.

\emph{The identity $W^{\ptr}=(U^{\ptr})^n$.} If $y_j\in U^{\ptr}$ for every $j$ then
$\langle w,y\rangle_{\Tr}=0$ on $W$; conversely $w=u\varepsilon_j$ gives $\Tr(uy_j)=0$ for all
$u\in U$.

(b)$\Rightarrow$(a) is immediate, and (a)$\Rightarrow$(c) follows from $f(0)=0$.

(c)$\Rightarrow$(b). Both $\sigma_H|_W$ and $\bigl(\sigma_H\circ(e\star\cdot)\bigr)|_W$ are
$\F_q$-linear, and by (c) the kernel of the first lies in that of the second, so the second factors
through the first as an $\F_q$-linear $A_0:\sigma_H(W)\to\F_q^s$; any $\F_q$-linear extension $A$
of $A_0$ to $\F_q^s$ gives (b).

(c)$\Leftrightarrow$(d). We have $\ker\sigma_H=S_H^{\ptr}$ and, by \eqref{eq:adjoint-star},
\[
\ker\bigl(\sigma_H\circ(e\star\cdot)\bigr)
=\{E:\langle e^{\dagger}\!\star h^{(i)},E\rangle_{\Tr}\!=0\,\forall i\}
=\bigl(e^{\dagger}\!\star S_H\bigr)^{\ptr}\!\!,
\]
where $e^{\dagger}\star S_H$ is the $\F_q$-span of the $e^{\dagger}\star h^{(i)}$, a subspace
because $e^{\dagger}\star\cdot$ is $\F_q$-linear. Thus (c) reads
$S_H^{\ptr}\cap W\subseteq(e^{\dagger}\star S_H)^{\ptr}$. Applying the inclusion-reversing
involution $(\cdot)^{\ptr}$ and using \eqref{eq:dual-intersection} together with
$(W^{\ptr})^{\ptr}=W$ gives $e^{\dagger}\star S_H\subseteq S_H+W^{\ptr}$.

\emph{The case $A=I$.} The identity $\sigma_H(e\star E)=\sigma_H(E)$ on $W$ reads
$\langle e^{\dagger}\star h^{(i)}-h^{(i)},E\rangle_{\Tr}=0$ for all $E\in W$ and every $i$, which is
\eqref{eq:relative-alignment}; the converse holds by bilinearity. For $U=\Fqm$ nondegeneracy of the
trace form turns \eqref{eq:relative-alignment} into $e^{\dagger}\star h^{(i)}=h^{(i)}$.
\end{proof}

The relative criterion exhibits no gap between arbitrary and linear recovery maps: on a subspace,
factorization of one linear map through another is automatically linear. The correction term
$W^{\ptr}$ in (d) cannot be omitted, as demonstrated by the counterexample $q=2$, $m=2$, $n=1$, $U=\langle1\rangle_{\F_2}$, $e=0$,
$h=1$: $s(h;E)=\Tr(E)=0$ on $W$, so the relative identity holds with $A=I$, yet
$e^{\dagger}\star h-h=1\in W^{\ptr}$ is nonzero.

\subsection{Bounded-rank compatibility and certified decoding}
\label{subsec:rank-restricted}

A bounded-rank decoder need act correctly only on errors of rank at most $t$. The exact criterion
for that weaker demand replaces the kernel of the syndrome map by its intersection with a rank
ball. The following lemma decomposes a low-rank kernel element into a difference of two vectors
within smaller rank balls, a property applied repeatedly below.

\begin{lemma}[Rank halving inside a componentwise subspace]
\label{lem:rank-halving}
Let $U\subseteq\Fqm$ be an $\F_q$-subspace, $W:=U^n$, $B^W_{\tau}:=\{E\in W:\wt_R(E)\le\tau\}$ and
$t\ge1$. Every $X\in B^W_{2t}$ admits $E,E'\in B^W_t$ with $E-E'=X$.
\end{lemma}

\begin{proof}
Let $w:=\wt_R(X)\le2t$ and choose a rank decomposition $\Phi_B(X)=\sum_{i=1}^{w}u_iv_i^{\top}$ in
which $u_1,\dots,u_w\in\F_q^m$ is a basis of the column space of $\Phi_B(X)$. Since $X\in W=U^n$,
every column of $\Phi_B(X)$ lies in $[U]_B$, hence so does its column space and hence every $u_i$.
Setting
\begin{equation}
\label{eq:rank-halving}
\begin{split}
E&:=\Phi_B^{-1}\Bigl(\sum_{i\le\lceil w/2\rceil}u_iv_i^{\top}\Bigr),\\
E'&:=\Phi_B^{-1}\Bigl(-\!\!\sum_{i>\lceil w/2\rceil}u_iv_i^{\top}\Bigr),
\end{split}
\end{equation}
both lie in $W$ with $\wt_R(E)\le\lceil w/2\rceil\le t$, $\wt_R(E')\le\lfloor w/2\rfloor\le t$ and
$E-E'=X$.
\end{proof}

The same lemma yields the exact condition under which a bounded-rank syndrome decoding problem has
at most one solution, which we record here because the criterion of
Theorem~\ref{thm:rank-restricted} refers to it.

\begin{definition}[Certifying radius-$t$ decoder]
\label{def:certifying-decoder}
Let $U\subseteq\Fqm$ be an $\F_q$-subspace, $W:=U^n$, let $\sigma:W\to\F_q^{s}$ be $\F_q$-linear
with kernel $\Kker:=\ker\sigma$, and let $t\ge1$. A map
$\mathrm{Dec}_t:\F_q^{s}\to W\cup\{\bot\}$ is a \emph{certifying radius-$t$ decoder} for $\sigma$
if it returns the unique $X\in W$ with $\wt_R(X)\le t$ and $\sigma(X)=y$ when exactly one such $X$
exists, and $\bot$ otherwise. Such a map exists and is unique, and $\mathrm{Dec}_t(0)=0$ if and
only if $d_R(\Kker)>t$. On any candidate produced by an algorithm three conditions are checked:
membership $X\in W$, the rank bound $\wt_R(X)\le t$, and syndrome consistency $\sigma(X)=y$
against the full record. These certify a valid radius-$t$ preimage but do not by themselves
exclude a second one; under the hypothesis of Theorem~\ref{thm:exact-unique}, uniqueness is guaranteed
a priori, so the three checks suffice. An algorithm may produce a candidate under a promise; it is
this verification that certifies the output.
\end{definition}

\begin{theorem}[Exact deterministic uniqueness]
\label{thm:exact-unique}
In the notation of Definition~\ref{def:certifying-decoder}, every syndrome has at most one preimage
in $W$ of rank at most $t$ if and only if $d_R(\Kker)>2t$, that is, if and only if
$\Kker\cap B^{W}_{2t}=\{0\}$. Under that hypothesis the certifying decoder returns
$\mathrm{Dec}_t(\sigma(E))=E$ for every $E\in W$ with $\wt_R(E)\le t$.
\end{theorem}

\begin{proof}
Two errors $E,E'\in W$ satisfy $\sigma(E)=\sigma(E')$ if and only if $E-E'\in\Kker$, by
$\F_q$-linearity of $\sigma$, and then $\wt_R(E-E')\le\wt_R(E)+\wt_R(E')\le2t$ by
\eqref{eq:rank-subadditive}.

If $\Kker=\{0\}$ then $\sigma$ is injective on $W$, so every syndrome has at most one preimage of
any rank, while $d_R(\Kker)=+\infty>2t$; both sides of the equivalence hold. Assume henceforth
$\Kker\neq\{0\}$. ($\Leftarrow$) If $E,E'$ have rank at most $t$ and $\sigma(E)=\sigma(E')$ then
$E-E'\in\Kker$ with $\wt_R(E-E')\le2t<d_R(\Kker)$, so $E=E'$. ($\Rightarrow$) If
$d_R(\Kker)\le2t$, choose a nonzero $X\in\Kker$ with $\wt_R(X)\le2t$ and apply
Lemma~\ref{lem:rank-halving} to write $X=E-E'$ with $E,E'\in B^{W}_{t}$; then
$\sigma(E)=\sigma(E')$ while $E\neq E'$. The final assertion is
Definition~\ref{def:certifying-decoder} applied to the syndrome $\sigma(E)$, whose preimage of rank
at most $t$ in $W$ is then unique.
\end{proof}

\begin{theorem}[Rank-restricted factorization]
\label{thm:rank-restricted}
Let $H$ be a family of check labels, $e:\Fqm\to\Fqm$ be $\F_q$-linear, $U\subseteq\Fqm$ an
$\F_q$-subspace, $W:=U^n$ and $t\ge1$. The following are equivalent.
\begin{enumerate}
\item[(a)] There is $f_t:\F_q^s\to\F_q^s$ with $\sigma_H(e\star E)=f_t\bigl(\sigma_H(E)\bigr)$ for
every $E\in B^W_t$.
\item[(b)] $\ker\sigma_H\cap B^W_{2t}\subseteq\ker\bigl(\sigma_H\circ(e\star\cdot)\bigr)$.
\item[(c)] Every $X\in S_H^{\ptr}\cap W$ with $\wt_R(X)\le2t$ lies in
$\bigl(e^{\dagger}\star S_H\bigr)^{\ptr}$.
\end{enumerate}
Moreover, if $d_R\bigl(S_H^{\ptr}\cap W\bigr)>2t$ then (b) holds for every $\F_q$-linear $e$ and
every family $H$; in that case $\sigma_H$ is injective on $B^W_t$ and one may take
$f_t=\sigma_H\circ(e\star\cdot)\circ\mathrm{Dec}_t$ on $\sigma_H(B^W_t)$, where $\mathrm{Dec}_t$ is
a certifying radius-$t$ decoder in the sense of Definition~\ref{def:certifying-decoder}, at the
cost of one bounded-rank decoding per syndrome.
\end{theorem}

\begin{proof}
(a)$\Rightarrow$(b). Let $X\in\ker\sigma_H\cap B^W_{2t}$ and split it by
Lemma~\ref{lem:rank-halving} as $X=E-E'$ with $E,E'\in B^W_t$. As $\sigma_H$ is $\F_q$-linear and
$\sigma_H(X)=0$, we get $\sigma_H(E)=\sigma_H(E')$, so (a) gives
$\sigma_H(e\star E)=\sigma_H(e\star E')$, that is $\sigma_H(e\star X)=0$.

(b)$\Rightarrow$(a). Define $f_t(y):=\sigma_H(e\star E)$ for any $E\in B^W_t$ with
$\sigma_H(E)=y$, and $f_t(y):=0$ otherwise. If $E,E'\in B^W_t$ share a syndrome then
$X:=E-E'\in W$ has $\sigma_H(X)=0$ and, by \eqref{eq:rank-subadditive}, $\wt_R(X)\le2t$; hence (b)
gives $\sigma_H(e\star E)=\sigma_H(e\star E')$ and $f_t$ is well defined.

(b)$\Leftrightarrow$(c) is the identification $\ker\sigma_H=S_H^{\ptr}$ together with
$\ker\bigl(\sigma_H\circ(e\star\cdot)\bigr)=\bigl(e^{\dagger}\star S_H\bigr)^{\ptr}$ from
\eqref{eq:adjoint-star}.

\emph{The minimum-rank certificate.} The hypothesis $d_R(S_H^{\ptr}\cap W)>2t$ gives
$\ker\sigma_H\cap B^W_{2t}=\{0\}$, which is contained in every subspace; injectivity on $B^W_t$ and
$\mathrm{Dec}_t(\sigma_H(E))=E$ for $E\in B^W_t$ follow from Theorem~\ref{thm:exact-unique} and
Definition~\ref{def:certifying-decoder}.
\end{proof}

\begin{corollary}[Linearity dichotomy]
\label{cor:linearity-dichotomy}
Let $t\ge1$, $U\subseteq\Fqm$ an $\F_q$-subspace and $W=U^n$. If some $A\in\F_q^{s\times s}$
satisfies $\sigma_H(e\star E)=A\,\sigma_H(E)$ for every $E\in W$ of rank at most $t$, then the same
identity holds for every $E\in W$, so $e^{\dagger}\star S_H\subseteq S_H+W^{\ptr}$; and if $A=I$
then $H$ is $W$-relatively aligned. For $U=\Fqm$ these read $e^{\dagger}\star S_H\subseteq S_H$ and
$e^{\dagger}\star h^{(i)}=h^{(i)}$, that is, every check is a fixed point of the \emph{adjoint}
$e^{\dagger}$; if in addition $e$ is a trace-self-adjoint idempotent, then
$\Fix(e^{\dagger})=\Fix(e)=\Imag(e)$ and this is $e$-alignment in the sense of
Definition~\ref{def:aligned-check}. Consequently any interface beyond criterion (d) of
Theorem~\ref{thm:projected-recoverability} must use a nonlinear recovery map, for which the
minimum-rank certificate of Theorem~\ref{thm:rank-restricted} is a sufficient but not a necessary
condition.
\end{corollary}

\begin{proof}
$L:=\sigma_H\circ(e\star\cdot)-A\circ\sigma_H$ is $\F_q$-linear and vanishes on every $E\in W$ of
rank at most $t$, in particular on every rank-one element of $W=U^n$. By
Lemma~\ref{lem:rankone-span} with $V=U$ those span $W$, so $L|_W\equiv0$, which is condition (b) of
Theorem~\ref{thm:projected-recoverability}; implication (b)$\Rightarrow$(d) gives the containment.
If $A=I$ then $L|_W\equiv0$ is \eqref{eq:relative-alignment}. Taking $U=\Fqm$ gives
$W^{\ptr}=\{0\}$. The last assertion is the contrapositive.
\end{proof}

\subsection{Recovering the projected error modulo stabilizer labels}
\label{subsec:quotient}

Theorems~\ref{thm:projected-recoverability} and \ref{thm:rank-restricted} concern the projected
syndrome, whereas a decoder must produce a correction, needed only modulo the stabilizer. That
demand is weaker than exact recovery of the projected error but, as
Lemma~\ref{lem:interface-hierarchy} records, strictly stronger than recovery of its syndrome; its
exact criterion has the same two-level structure.

\begin{theorem}[Quotient-valued factorization]
\label{thm:quotient-interface}
Let $H$, $S_H$, $e$, $e^{\dagger}$ be as in Theorem~\ref{thm:projected-recoverability}, let
$U\subseteq\Fqm$ be an $\F_q$-subspace, $W:=U^n$, and let $T\subseteq\Fqm^n$ be an $\F_q$-subspace.
\begin{enumerate}
\item[(a)] $H$ admits a $T$-valued projected interface on $\Fqm^n$ if and only if
$e\star\ker\sigma_H\subseteq T$, equivalently if and only if
\begin{equation}
\label{eq:quotient-criterion}
e^{\dagger}\star T^{\ptr}\subseteq S_H,
\end{equation}
and $g$ may then be taken $\F_q$-linear.
\item[(b)] For $t\ge1$, $H$ admits a $T$-valued projected interface on $B^W_t$ if and only if
$e\star\bigl(\ker\sigma_H\cap B^W_{2t}\bigr)\subseteq T$.
\item[(c)] Taking $T=\{0\}$ in (a) gives $\Imag(e^{\dagger})^n\subseteq S_H$, so exact recovery of
the projected error is strictly stronger than recovery of its syndrome; and the minimum-rank
certificate of Theorem~\ref{thm:rank-restricted} implies (b) for every $T$.
\end{enumerate}
\end{theorem}

\begin{proof}
(a) If $g$ exists and $\sigma_H(E)=\sigma_H(E')$ then $e\star(E-E')\in T$; taking
$E'=0$ and letting $E$ range over $\ker\sigma_H$ gives $e\star\ker\sigma_H\subseteq T$.

\emph{Sufficiency and linearity.} The map $E\mapsto e\star E+T$ is $\F_q$-linear and, under the
hypothesis, vanishes on $\ker\sigma_H$, so it factors through
$\Fqm^n/\ker\sigma_H\cong\sigma_H(\Fqm^n)$ as an $\F_q$-linear map $g_0$; any $\F_q$-linear
extension of $g_0$ to $\F_q^s$ is the required $g$.

\emph{Duality.} For componentwise $\F_q$-linear $\varphi$ and any $\F_q$-subspace
$U_1\subseteq\Fqm^n$, \eqref{eq:adjoint-star} gives
$\varphi(U_1)^{\ptr}=(\varphi^{\dagger}\star\cdot)^{-1}\bigl(U_1^{\ptr}\bigr)$; since
$(\cdot)^{\ptr}$ is inclusion-reversing and involutive by \eqref{eq:trace-dual-dimension},
$\varphi(U_1)\subseteq T$ is equivalent to
$\varphi^{\dagger}\star T^{\ptr}\subseteq U_1^{\ptr}$. With $\varphi=e\star\cdot$ and
$U_1=\ker\sigma_H=S_H^{\ptr}$, so that $U_1^{\ptr}=S_H$, this turns
$e\star\ker\sigma_H\subseteq T$ into \eqref{eq:quotient-criterion}.

(b) Let $X\in\ker\sigma_H\cap B^W_{2t}$ and split it by
Lemma~\ref{lem:rank-halving} into $E,E'\in B^W_t$ with $E-E'=X$. Then $\sigma_H(E)=\sigma_H(E')$, so
$e\star X=e\star E-e\star E'\in T$.

\emph{Sufficiency.} Define $g$ on $\sigma_H(B^W_t)$ by $g(y):=e\star E+T$ for any $E\in B^W_t$ with
$\sigma_H(E)=y$, and arbitrarily elsewhere. If $E,E'\in B^W_t$ share a syndrome then
$X:=E-E'\in W$ satisfies $\sigma_H(X)=0$ and $\wt_R(X)\le2t$ by \eqref{eq:rank-subadditive}, so
$e\star X\in T$ and $g$ is well defined.

(c) For $T=\{0\}$ one has $T^{\ptr}=\Fqm^n$ and \eqref{eq:quotient-criterion} reads
$\Imag(e^{\dagger})^n\subseteq S_H$. The certificate gives $\ker\sigma_H\cap B^W_{2t}=\{0\}$, whose
image is $\{0\}\subseteq T$.
\end{proof}

\begin{corollary}[Quotient-valued linearity dichotomy]
\label{cor:quotient-linearity}
Let $t\ge1$, let $U\subseteq\Fqm$ be an $\F_q$-subspace, $W:=U^n$, and $T\subseteq\Fqm^n$ an
$\F_q$-subspace. If some $\F_q$-linear $g:\F_q^s\to\Fqm^n/T$ satisfies
$e\star E+T=g\bigl(\sigma_H(E)\bigr)$ for every $E\in W$ of rank at most $t$, then the same identity
holds for every $E\in W$. In particular, on the ambient label space $W=\Fqm^n$ a linear
bounded-rank quotient-valued interface satisfies the ambient criterion
$e\star\ker\sigma_H\subseteq T$ of Theorem~\ref{thm:quotient-interface}(a), for an arbitrary
$\F_q$-linear $e$, and the interfaces left outside that criterion are exactly the \emph{nonlinear}
bounded-rank quotient-valued ones. At $T=S_X$ the criterion is in addition equivalent to the
structural description of Theorem~\ref{thm:quotient-characterization} when $e$ is a
trace-self-adjoint idempotent, but not for a general $e$: in the matrix form of
Section~\ref{subsec:matrix-css} with $q=2$, $a=3$, $n=1$, $S_X=\langle(1,0,0)^{\top}\rangle$ and
$S_Z=\langle(0,1,0)^{\top}\rangle$, the idempotent $P$ with first row $(1,1,0)$ and zero elsewhere
has $P^{\top}\neq P$ and $P\ker\sigma_H=S_X$, yet $PS_Z=S_X\not\subseteq S_Z$.
\end{corollary}

\begin{proof}
The map $L(E):=\bigl(e\star E+T\bigr)-g\bigl(\sigma_H(E)\bigr)$ is $\F_q$-linear on $W$, being a
difference of composites of $\F_q$-linear maps, and vanishes on every $E\in W$ of rank at most $t$,
in particular on every rank-one element of $W=U^n$; those span $W$ by
Lemma~\ref{lem:rankone-span}, so $L|_W\equiv0$. For $W=\Fqm^n$ and $E\in\ker\sigma_H$, linearity of
$g$ gives $g(0)=T$, whence $e\star E\in T$.
\end{proof}

\begin{lemma}[The interface hierarchy]
\label{lem:interface-hierarchy}
Let $(S_X,S_Z)$ be trace-orthogonal, let $H$ span $S_Z$ and let $T\subseteq\ker\sigma_H$, as holds
for $T=S_X$ by \eqref{eq:trace-css-commutation}. Then a $T$-valued projected interface on an error
set determines the projected syndrome there. Consequently
$\mathrm{I}_3\Rightarrow\mathrm{I}_1$ and $\mathrm{I}_6\Rightarrow\mathrm{I}_5$ in the notation
of Definition~\ref{def:interfaces}: recovering the projected error modulo the check space of
the corrected sector is strictly stronger than recovering the projected syndrome and strictly
weaker than recovering the projected error itself
(Theorem~\ref{thm:quotient-interface}(c)).
\end{lemma}

\begin{proof}
Since $T\subseteq\ker\sigma_H$, the map $\sigma_H$ is constant on cosets of $T$ and descends to
$\bar\sigma_H$ with $\bar\sigma_H(y+T)=\sigma_H(y)$; composing with $g$ expresses
$\sigma_H(e\star E)$ as a function of $\sigma_H(E)$. The strictness of these implications is
witnessed by Example~\ref{ex:hierarchy-strict} and Proposition~\ref{prop:separation}.
\end{proof}

\begin{example}[$\mathrm{I}_5$ does not imply $\mathrm{I}_6$]
\label{ex:hierarchy-strict}
Let $q=2$, $m=3$, $n=1$, let $(b_1,b_2,b_3)$ be a self-dual basis of $\F_8/\F_2$
\cite{seroussi1980}, and set $S_X:=\langle b_2\rangle$, $S_Z:=\langle b_3\rangle$,
$e(z):=\Tr(b_1z)\,b_1$. Then $(S_X,S_Z)$ is trace-orthogonal with $K_q=1$, $e$ is a rank-one
trace-self-adjoint idempotent, and self-duality gives
$\sigma_H(e\star E)=\Tr(b_1E)\Tr(b_3b_1)=0$ for every $E$, so the zero map is an $\F_2$-linear
ambient interface and $\mathrm{I}_1,\mathrm{I}_2,\mathrm{I}_5$ hold. But $E=0$ and $E'=b_1$ share a
syndrome and have rank weight at most one, while $e\star E+S_X=S_X$ and $e\star E'+S_X=b_1+S_X$
differ since $b_1\notin\langle b_2\rangle$; so $\mathrm{I}_6$, hence
$\mathrm{I}_3,\mathrm{I}_4$, fail. Both spans are $e$-invariant, $e\star b_2=e\star b_3=0$, yet the
branch carries $K_{\Imag(e)}=1>0$, so
Theorem~\ref{thm:quotient-characterization} does not apply.
\end{example}

The example also separates the two linear classes. Throughout this comparison $e$ is a
trace-self-adjoint idempotent, as Theorem~\ref{thm:quotient-characterization} requires. Write
$\mathcal L_{\mathrm{syn}}$ for the pairs
whose check spans are both $e$-invariant, which by
Theorem~\ref{thm:projected-recoverability} and Corollary~\ref{cor:linearity-dichotomy} is exactly
the class satisfying $\mathrm{I}_1$ in both sectors, and $\mathcal L_{\mathrm{quot}}$ for those
satisfying $\mathrm{I}_3$. By Theorem~\ref{thm:quotient-characterization},
\begin{equation}
\label{eq:linear-class-inclusion}
\mathcal L_{\mathrm{quot}}
=\{(S_X,S_Z)\in\mathcal L_{\mathrm{syn}}:K_{\Imag(e)}=0\}
\subsetneq\mathcal L_{\mathrm{syn}},
\end{equation}
the inclusion being strict by Example~\ref{ex:hierarchy-strict}. Treating these two classes interchangeably would obscure the necessity of the second condition in Theorem~\ref{thm:quotient-characterization} and overlook the dimensional cost characterized in Corollary~\ref{cor:interface-cost}.

\subsection{The two relative distances and the exact radii}
\label{subsec:radii}

The conditions $\mathrm I_5$ and $\mathrm I_6$ are governed by two invariants, which are exact on
Singleton-optimal pairs.

\begin{definition}[Relative projection distances]
\label{def:relative-distances}
Let $e$ be $\F_q$-linear acting componentwise, let $H$ span $S_Z$, so that
$\cD:=\ker\sigma_H=S_Z^{\ptr}$, and let $T:=S_X$. Put
\begin{equation}
\label{eq:relative-distances}
\begin{split}
\delta_{\mathrm{syn}}(e)&:=\min\{\wt_R(x):x\in \cD,\ e\star x\notin \cD\},\\
\delta_{\mathrm{quot}}(e)&:=\min\{\wt_R(x):x\in \cD,\ e\star x\notin T\},
\end{split}
\end{equation}
with $\min\emptyset:=+\infty$, and define the mirrored quantities in the $Z$ sector by sector
interchange.
\end{definition}

\begin{remark}[Relation to relative generalized matrix weights]
\label{rem:rgmw}
Writing $\Kker_{\mathrm{syn}}:=\cD\cap(e\star\cdot)^{-1}(\cD)$ and
$\Kker_{\mathrm{quot}}:=\cD\cap(e\star\cdot)^{-1}(T)$, both $\F_q$-subspaces of $\cD$,
\eqref{eq:relative-distances} reads
$\delta_{\mathrm{syn}}(e)=\min\{\wt_R(x):x\in\cD\setminus\Kker_{\mathrm{syn}}\}$ and
$\delta_{\mathrm{quot}}(e)=\min\{\wt_R(x):x\in\cD\setminus\Kker_{\mathrm{quot}}\}$. For a nested
pair of matrix codes $C_2\subsetneq C_1$ the quantity $\min\{\rank x:x\in C_1\setminus C_2\}$ is
the first relative generalized matrix weight $d_{M,1}(C_1,C_2)$ of
\cite{martinezpenas2018rgmw}, so that both invariants are relative generalized matrix weights
whenever the relevant subspace is proper, the convention $\min\emptyset=+\infty$ covering the
degenerate cases. Rather than introducing the abstract notion of a relative distance, our contribution lies in identifying the two specific nested pairs induced by the projector and quantum syndrome semantics, along with their exact evaluation in Theorem~\ref{thm:radius-dichotomy}.
\end{remark}

\begin{lemma}[The two invariants govern the bounded-rank interfaces]
\label{lem:relative-criterion}
For every $t\ge1$ one has $\mathrm{I}_{5,X}(t)$ if and only if $2t<\delta_{\mathrm{syn}}(e)$, and
$\mathrm{I}_{6,X}(t)$ if and only if $2t<\delta_{\mathrm{quot}}(e)$. Moreover
$d_R(\cD)\le\delta_{\mathrm{quot}}(e)\le\delta_{\mathrm{syn}}(e)$.
\end{lemma}

\begin{proof}
Criterion (b) of Theorem~\ref{thm:rank-restricted} at $W=\Fqm^n$ reads
$e\star(\cD\cap B_{2t})\subseteq \cD$, that is, no $x\in \cD$ with $\wt_R(x)\le2t$ has
$e\star x\notin \cD$, which is $2t<\delta_{\mathrm{syn}}(e)$. Criterion (b) of
Theorem~\ref{thm:quotient-interface} at $T=S_X$ and $W=\Fqm^n$ reads
$e\star(\cD\cap B_{2t})\subseteq S_X$, which is $2t<\delta_{\mathrm{quot}}(e)$. Since
$S_X\subseteq \cD$ by CSS commutation, $e\star x\notin \cD$ implies $e\star x\notin S_X$, whence
$\delta_{\mathrm{quot}}\le\delta_{\mathrm{syn}}$; and $e\star x\notin S_X$ forces $x\neq0$, whence
$\delta_{\mathrm{quot}}\ge d_R(\cD)$.
\end{proof}

\begin{theorem}[Radius dichotomy on every layout]
\label{thm:radius-dichotomy}
Let $(S_X,S_Z)$ be trace-orthogonal check spans attaining equality in
\eqref{eq:rank-singleton} with $K_q>0$, on an arbitrary layout, and let $e$ be any nonzero
$\F_q$-linear map on $\Fqm$ acting componentwise; neither idempotency nor trace-self-adjointness
is assumed. Write $\cD_X:=S_Z^{\ptr}=\ker\sigma_H$. Then
\begin{enumerate}
\item[(i)] $\delta_{\mathrm{quot}}(e)=d^{R}_X$;
\item[(ii)] $\delta_{\mathrm{syn}}(e)=+\infty$ if $e\star\cD_X\subseteq\cD_X$, equivalently if
$e^{\dagger}\star S_Z\subseteq S_Z$, and $\delta_{\mathrm{syn}}(e)=d^{R}_X$ otherwise.
\end{enumerate}
Consequently, for every $t\ge1$, $\mathrm I_{6,X}(t)$ holds if and only if $2t<d^{R}_X$, while
$\mathrm I_{5,X}(t)$ holds if and only if $e^{\dagger}\star S_Z\subseteq S_Z$ or $2t<d^{R}_X$; and
$\mathrm I_{3,X}$, hence $\mathrm I_{4,X}$, fails for every nonzero $e$. The mirrored statements
hold in the $Z$ sector.
\end{theorem}

\begin{proof}
Write $u:=d^{R}_X-1$ and $v:=d^{R}_Z-1$, so that $u+v\le L-1$ because $K_q>0$. By
Theorem~\ref{thm:singleton-equality-mrd}(i) the space $S_Z$ is either $\{0\}$ or MRD with
$d_R(S_Z)=L-u+1\ge2$, and $S_X$ is either $\{0\}$ or MRD with
$d_R(S_X)=L-v+1>u+1=d^{R}_X$; by Corollary~\ref{cor:purity} and Lemma~\ref{lem:gram-reduction} the
space $\cD_X$ is MRD with $d_R(\cD_X)=d^{R}_X$, left multiplication by the inverse trace Gram
matrix preserving rank. By Lemma~\ref{lem:mrd-generation}, valid on every layout, $\cD_X$ is
spanned over $\F_q$ by its elements of rank weight $d^{R}_X$.

(i) By Lemma~\ref{lem:relative-criterion}, $\delta_{\mathrm{quot}}(e)\ge d_R(\cD_X)=d^{R}_X$.
Suppose $\delta_{\mathrm{quot}}(e)>d^{R}_X$. Then $e\star x\in S_X$ for every $x\in\cD_X$ of rank
weight $d^{R}_X$. Since $\wt_R(e\star x)\le\wt_R(x)=d^{R}_X<d_R(S_X)$ by
Lemma~\ref{lem:rank-non-expansion}, with the convention $d_R(\{0\})=+\infty$ covering
$S_X=\{0\}$, this forces $e\star x=0$; since these elements span $\cD_X$ and $e\star\cdot$ is
$\F_q$-linear, $e\star\cD_X=\{0\}$, that is, $\cD_X\subseteq\ker(e)^n$. Taking trace duals and
using $\ker(e)^{\ptr}=\Imag(e^{\dagger})$, which holds because the trace form is nondegenerate and
$\langle e\star x,y\rangle_{\Tr}=\langle x,e^{\dagger}\star y\rangle_{\Tr}$ by
\eqref{eq:adjoint-star}, we get
$S_Z=\cD_X^{\ptr}\supseteq\bigl(\ker(e)^n\bigr)^{\ptr}=\Imag(e^{\dagger})^n$. The adjoint of a
nonzero map is nonzero, so $\Imag(e^{\dagger})\neq\{0\}$ and, by Lemma~\ref{lem:rankone-span},
$S_Z$ contains a nonzero label of rank weight one. That contradicts $d_R(S_Z)\ge2$, the case
$S_Z=\{0\}$ being excluded as well. Hence $\delta_{\mathrm{quot}}(e)=d^{R}_X$.

(ii) If $e\star\cD_X\subseteq\cD_X$, the candidate set in \eqref{eq:relative-distances} is empty,
yielding a minimum of $+\infty$ by convention. Otherwise, $\delta_{\mathrm{syn}}(e)\ge\delta_{\mathrm{quot}}(e)=d^{R}_X$ by Lemma~\ref{lem:relative-criterion}; and if $\delta_{\mathrm{syn}}(e)>d^{R}_X$ then every element of $\cD_X$ of rank weight $d^{R}_X$ would have its image in $\cD_X$, whence $e\star\cD_X\subseteq\cD_X$
by the spanning property, the excluded case. The stated equivalence is the adjoint identity: for
$x\in\cD_X$ and $z\in S_Z$ one has $\langle e\star x,z\rangle_{\Tr}=\langle x,e^{\dagger}\star
z\rangle_{\Tr}$, so $e\star\cD_X\subseteq\cD_X=S_Z^{\ptr}$ holds if and only if
$e^{\dagger}\star S_Z\subseteq\cD_X^{\ptr}=S_Z$.

The characterizations of $\mathrm I_{5,X}(t)$ and $\mathrm I_{6,X}(t)$ follow from
Lemma~\ref{lem:relative-criterion}. Finally $\mathrm I_{3,X}$ would give a $T$-valued interface on
the whole label space, hence on $B_t$ for every $t$, so $\mathrm I_{6,X}(t)$ for every $t$, which
contradicts $\delta_{\mathrm{quot}}(e)=d^{R}_X<+\infty$; and
$\mathrm I_{4,X}\Leftrightarrow\mathrm I_{3,X}$ by Corollary~\ref{cor:quotient-linearity}.
\end{proof}

\begin{corollary}[Rigidity and exact radii for $m\le n$]
\label{cor:optimal-rigidity}
Let $m\le n$, let $S_X,S_Z\subseteq\Fqm^n$ be trace-orthogonal check spans with $K_q>0$ attaining
equality in \eqref{eq:rank-singleton} at $a=m$, and let $e$ be any $\F_q$-linear idempotent on
$\Fqm$ with $1\le\rank(e)\le m-1$, not necessarily trace-self-adjoint. Then:
\begin{enumerate}
\item[(i)] if $d^{R}_X\ge2$, then $\cD_X$ is invariant under neither $e$ nor $e^{\dagger}$, and
$S_Z$ is invariant under neither; hence the measured $Z$-check family admits no ambient
projected-syndrome interface at $e$, linear or not
(Theorem~\ref{thm:projected-recoverability} at $U=\Fqm$),
$\delta_{\mathrm{syn}}(e)=\delta_{\mathrm{quot}}(e)=d^{R}_X$, and for every integer $t\ge1$
\begin{equation}
\label{eq:exact-radius}
\mathrm I_{5,X}(t)\iff\mathrm I_{6,X}(t)\iff 2t<d^{R}_X;
\end{equation}
\item[(ii)] the mirrored statement holds in the $Z$ sector when $d^{R}_Z\ge2$;
\item[(iii)] if $d^{R}_X\ge2$ and $d^{R}_Z\ge2$, then neither check space is invariant under $e$
or $e^{\dagger}$, so for a trace-self-adjoint $e$ the code splits along no nontrivial projector
(Theorem~\ref{thm:invariant-splitting}) and no ambient quotient-valued interface exists in either
sector; both sectors admit an unrestricted bounded-rank interface exactly at the radii
$1\le t\le t_{\max}$ with $t_{\max}=\lfloor(\min(d^{R}_X,d^{R}_Z)-1)/2\rfloor$, and
$t_{\max}\ge1$ if and only if $\min(d^{R}_X,d^{R}_Z)\ge3$.
\end{enumerate}
\end{corollary}

\begin{proof}
(i) By Theorem~\ref{thm:singleton-equality-mrd} and Corollary~\ref{cor:purity} the space $\cD_X$ is
MRD with $2\le d_R(\cD_X)=d^{R}_X=u+1\le m$, the upper bound holding because $u\le L-1=m-1$. The trace
adjoint $e^{\dagger}$ is again an idempotent of the same rank, its matrix being
$G_B^{-1}M_e^{\top}G_B$, so Lemma~\ref{lem:mrd-noninvariance}, applied to $\cD_X$ with $P$ the
matrix of $e$ and then of $e^{\dagger}$, excludes both invariances; the equivalence with the
invariance of $S_Z$ follows from Theorem~\ref{thm:radius-dichotomy}(ii). The absence of an ambient
projected-syndrome interface follows from criterion~(d) of Theorem~\ref{thm:projected-recoverability} at
$U=\Fqm$, which reads $e^{\dagger}\star S_Z\subseteq S_Z$, and the two values of the relative
distances are established in Theorem~\ref{thm:radius-dichotomy}; \eqref{eq:exact-radius} then follows from
Lemma~\ref{lem:relative-criterion}. (ii) Claim~(ii) follows by sector interchange. (iii) When $d^{R}_Z\ge2$ we have $v\ge1$, so $S_X$ is MRD with $2\le d_R(S_X)=m-v+1\le m$ and Lemma~\ref{lem:mrd-noninvariance} applies to it as well; the splitting statement follows from Theorem~\ref{thm:invariant-splitting}, the absence of an ambient quotient-valued interface is guaranteed by Theorem~\ref{thm:radius-dichotomy}, and the range of admissible radii follows from \eqref{eq:exact-radius} in the two sectors.
\end{proof}

\begin{remark}[Scope and Boundaries of the Dichotomy]
\label{rem:dichotomy}
Theorem~\ref{thm:radius-dichotomy} removes the hypotheses $m\le n$, $d^{R}_Z\ge2$ and idempotency
from Corollary~\ref{cor:optimal-rigidity}. For $m\le n$ and $d^{R}_X\ge2$ the invariant case never occurs
at a nontrivial idempotent, by Lemma~\ref{lem:mrd-noninvariance}; both hypotheses are needed. At
the one-sided endpoint $u=0$, where $S_Z=\{0\}$ and $d^{R}_X=1$, the space $\cD_X$ is the whole
label space, so every $e$ leaves it invariant and $\delta_{\mathrm{syn}}(e)=+\infty$ while
$\delta_{\mathrm{quot}}(e)=1$; and for $a>n$ the invariant case occurs at sector distance two. The Singleton-optimal $6\times3$ pair of
Proposition~\ref{prop:stacking-sharpness}, two stacked copies of
$\mathrm{QGab}(\boldsymbol\alpha,1,1)$ over $\F_8$, has exactly $72=(8+1)\cdot8$ nontrivial
idempotents of $\F_2^{6\times6}$ leaving $\cD_X$ invariant, all of rank three, matching the
rank-one idempotents of $M_2(\F_8)$; for those, $\delta_{\mathrm{syn}}=+\infty$ while
$\delta_{\mathrm{quot}}=2$, so $\mathrm I_1,\mathrm I_2,\mathrm I_5$ hold at every radius while
$\mathrm I_6$ requires $2t<d^{R}_X$ and $\mathrm I_3,\mathrm I_4$ fail. Singleton optimality cannot be dropped either, by
Example~\ref{ex:nonoptimal-pair}.

At $\min(d^{R}_X,d^{R}_Z)=2$ part~(iii) of Corollary~\ref{cor:optimal-rigidity} leaves no
admissible radius, so Singleton optimality with two-sided protection does not by itself place a
code in the unrestricted bounded-rank classes; the smallest instance is
$\mathrm{QGab}(\boldsymbol\alpha,1,1)$ over $\F_8$, namely
$\llbracket9,3,2^{R}/2^{R}\rrbracket_2$, whereas the family of
Proposition~\ref{prop:separation} has $d^{R}_X=d^{R}_Z=r+1\ge3$ and $t_{\max}=\lfloor
r/2\rfloor\ge1$.
\end{remark}

\begin{remark}[Interpretation and Scope of the Quotient-Valued Target]
\label{rem:target-meaning}
The target of $\mathrm I_3$ and $\mathrm I_6$ is the label-valued function
$F_e(E):=e\star E+S_X$, which descends to the stabilizer class $E+S_X$ if and only if
$e\star S_X\subseteq S_X$. On the whole label space that invariance is implied by $\mathrm I_3$
itself, since $S_X\subseteq\cD_X$ and the criterion of Theorem~\ref{thm:quotient-interface}(a)
reads $e\star\cD_X\subseteq S_X$; on a rank ball it is not, and for a Singleton-optimal two-sided pair
with $m\le n$ it fails at every nontrivial idempotent by
Corollary~\ref{cor:optimal-rigidity}; on tall layouts it can hold, as it does for the block
projectors of Proposition~\ref{prop:stacking-sharpness}. For
$\mathrm{QGab}(\boldsymbol\alpha,1,1)$ over $\F_8$ with $\boldsymbol\alpha=(3,5,7)$ and
$e(z)=\Tr(z)\cdot1$ in the coordinates of Example~\ref{ex:separation}, the stabilizer label
$s=\boldsymbol\alpha\in S_X$ has $e\star s=(1,1,1)\notin S_X$, so the stabilizer-equivalent errors
$E=0$ and $E=s$ have different targets. Theorem~\ref{thm:radius-dichotomy} is accordingly a
statement about inferring the function $F_e$ of a rank-bounded label from the measured syndrome,
which is what a projector-based post-processing layer computes, and not by itself one about
recovering an independent logical subsystem: applying $e\star E$ as a correction leaves the residual component $(\mathrm{id}-e)\star E$, which must be addressed via a distinct recovery mechanism.
\end{remark}

\begin{example}[Ordinary parameters do not determine the relative distances]
\label{ex:nonoptimal-pair}
Fix a self-dual normal basis of $\F_8/\F_2$, identifying $\F_8^{3}$ with $\F_2^{3\times3}$ and the
trace pairing with the standard one (Lemma~\ref{lem:gram-reduction}), write $E_{ij}$ for the matrix
units, and let $e$ act by left multiplication by the symmetric, hence trace-self-adjoint,
idempotent $P=\mathrm{diag}(1,0,0)$. With $A:=E_{23}$ and $B:=E_{11}+E_{31}$, define two CSS pairs
by $S_X^{(i)}:=\{0\}$ and $S_Z^{(i)}:=\cD_i^{\ptr}$, where $\cD_1:=\langle A,I_3\rangle_{\F_2}$
and $\cD_2:=\langle A,B\rangle_{\F_2}$, so that $\ker\sigma_H=\cD_i$. Both have
$\dim_{\F_2}S_X=0$, $\dim_{\F_2}S_Z=7$, $K=2$ and $d^{R}_X=d^{R}_Z=1$, since $\rank A=1$ and
$E_{11}\notin S_Z^{(i)}$, and neither is Singleton-optimal, \eqref{eq:rank-singleton} allowing
$K\le9$. Their relative distances differ: in the first, $PA=0$ while $I_3$ and $I_3+A$ have rank
three and are sent to $E_{11}\notin\cD_1$, giving
$\delta_{\mathrm{quot}}(e)=\delta_{\mathrm{syn}}(e)=3$; in the second, $\rank B=1$ and
$PB=E_{11}\notin\cD_2$, giving $\delta_{\mathrm{quot}}(e)=\delta_{\mathrm{syn}}(e)=1$. No
function of the ordinary parameters can therefore give the relative distances, and at $t=1$ the
first pair satisfies $\mathrm I_5$ and $\mathrm I_6$ while the second does not.
\end{example}

\section{Linear Interfaces: Splitting, Logical Cost and One-Sidedness}
\label{sec:linear}

Corollary~\ref{cor:optimal-rigidity} rules out linear interfaces at a nontrivial idempotent on a
Singleton-optimal pair with $m\le n$ and both sector distances at least two. This section characterizes
the structural constraints imposed on a code when such interfaces exist: an ambient projected-syndrome
interface in both sectors is equivalent to a tensor decomposition across the projector, an ambient
quotient-valued interface incurs a dimensional penalty equal to the logical dimension of the projected branch, and a design in which the two sectors occupy complementary projector images is necessarily one-sided.

\subsection{Compatible check families split}
\label{subsec:splitting}

\begin{theorem}[Splitting of invariant families]
\label{thm:invariant-splitting}
Let $e$ be an $\F_q$-linear idempotent on $\Fqm$, acting componentwise. An $\F_q$-subspace
$S\subseteq\Fqm^n$ satisfies $e\star S\subseteq S$ if and only if
$S=\bigl(S\cap\Imag(e)^n\bigr)\oplus\bigl(S\cap\ker(e)^n\bigr)$. Consequently, if $e=e_1$ is
self-adjoint of rank $r$, $V=\Imag(e_1)$, $W=\Imag(e_2)=V^{\ptr}$, and the trace-orthogonal check
spans $S_X$, $S_Z$ are both $e_1$-invariant, equivalently both admit a linear projected-syndrome
post-processing map by Theorem~\ref{thm:projected-recoverability}, then in a \emph{split}
$\F_q$-basis of $\Fqm=V\oplus W$, that is, one that is the union of a basis of $V$ and a basis of
$W$, the stabilizer group is the direct product of a group supported on the $nr$ qudits
$A=\{(j,\ell):\ell\le r\}$ and a group supported on the remaining $n(m-r)$, the code is a
tensor product $Q_A\otimes Q_{\bar A}$ with $K_q=K_A+K_{\bar A}$, and
$d^{R}_X=\min(d^{R,A}_X,d^{R,\bar A}_X)$, $d^{R}_Z=\min(d^{R,A}_Z,d^{R,\bar A}_Z)$, a factor without
nontrivial logical operators contributing $+\infty$ in accordance with the convention of
Section~\ref{subsec:rank-metric}. The converse of the displayed consequence also holds: if, in a
split basis, the realized stabilizer group is the direct product of two subgroups supported on
complementary qudit sets $A$ and $\bar A$, then $S_X$ and $S_Z$ are $e_1$-invariant. Hence, for a
trace-self-adjoint idempotent, the existence of a linear ambient projected-syndrome interface
\emph{in both sectors} is equivalent to the code being a tensor product across the $e_1$-split; a
linear interface in one sector alone does not force the splitting of the other check space.
\end{theorem}

\begin{proof}
If $S$ splits as stated and $x=a+b$, then $e\star x=a\in S$; conversely for $x\in S$,
$e\star x\in S\cap\Imag(e)^n$ and $x-e\star x\in S\cap\ker(e)^n$ by idempotence, the sum being
direct. Applying this decomposition to $S_X$ and $S_Z$, in a split basis $B'$ the Gram matrix $G_{B'}$ is block
diagonal, so coordinates of $V$-parts occupy the first $r$ intra-carrier positions and those of
$W$-parts the last $m-r$, before and after multiplication by $G_{B'}$; under either convention of
Proposition~\ref{prop:sector-split-embedding}, the realized operators of the two parts are therefore
supported on the disjoint sets $A$ and $\bar A$. Consequently, the stabilizer group is the direct product of two
commuting subgroups on disjoint qudits, the code decomposes as the tensor product of the factors, and the
logical count is additive.

For the distances, write $\Phi_{B'}(\ell)=\left(\begin{smallmatrix}\ell_A\\
\ell_{\bar A}\end{smallmatrix}\right)$, the first $r$ rows carrying the $A$-supported part. A row
submatrix has rank at most the rank of the matrix, so
\begin{equation}
\label{eq:row-submatrix-bound}
\wt_R(\ell)\ge\max\bigl\{\rank(\ell_A),\rank(\ell_{\bar A})\bigr\},
\end{equation}
where we emphasize that rank is not additive across disjoint row supports. The stabilizer
group being the direct product of the factor groups, $\ell=\ell_A+\ell_{\bar A}$ descends to
logical cosets: a class is nontrivial exactly when some factor component is, independently of the
representative, a stabilizer shift altering each component by a factor stabilizer. Hence by
\eqref{eq:row-submatrix-bound} every nontrivial representative has rank at least
$\min(d^{R,A},d^{R,\bar A})$ in the relevant sector, and conversely a minimum-rank nontrivial
representative of the minimizing factor, extended by zero, attains it.

For the converse, suppose the group is the direct product of subgroups $G_A$ and $G_{\bar A}$
supported on $A$ and $\bar A$. Every element factors as $g=g_Ag_{\bar A}$, and since the supports
are disjoint the Pauli type at each qudit is contributed by exactly one factor; hence an $X$-type
element has $X$-type factors. In the split basis the label of an $A$-supported $X$-type element
lies in $\Imag(e_1)^n$ and that of an $\bar A$-supported one in $\ker(e_1)^n$, so
$S_X=(S_X\cap\Imag(e_1)^n)\oplus(S_X\cap\ker(e_1)^n)$, which is $e_1$-invariance by the first
assertion; the same argument applies to $S_Z$. The final equivalence combines this with
Theorem~\ref{thm:projected-recoverability} and Corollary~\ref{cor:linearity-dichotomy}.
\end{proof}

\subsection{Aligned families and trace-self-adjoint idempotents}
\label{subsec:aligned}

The criterion of Theorem~\ref{thm:projected-recoverability} is met in the simplest way by families
whose checks are fixed by the projector.

\begin{definition}[Aligned check family]
\label{def:aligned-check}
Let $e$ be an $\F_q$-linear idempotent. A check vector $h$ is \emph{$e$-aligned} if $e\star h=h$,
and a family is \emph{aligned} if all of its members are. For self-adjoint $e$ one has
$\Fix(e^{\dagger})=\Imag(e)$, so the aligned families are exactly those with all checks in
$\Imag(e)^n$.
\end{definition}

\begin{corollary}[The exact class of post-processing-free families]
\label{cor:aligned-canonical}
Let $e$ be a self-adjoint $\F_q$-linear map. Then for all $h$ and $E$,
\begin{equation}
\label{eq:projected-trace-syndrome}
s(e\star h;E)=s(h;e\star E),
\end{equation}
and if $h$ is $e$-aligned then $s(h;E)=s(h;e\star E)$. Conversely the measured syndrome of a family
$H$ coincides with the projected syndrome check by check exactly for the families contained in
$\Fix(e)^n$, and if $e\star S_H\not\subseteq S_H$ then no function of the measured syndrome returns
the projected syndrome on the whole label space. The admissible families are closed under change of
basis, reordering and adjunction of redundant checks within $\Fix(e)^n$. If $e$ is moreover
idempotent, an aligned family is obtained from an arbitrary one, for instance from the rows of a
Gabidulin parity-check matrix, by $e$-projection, since $e\star(e\star h)=e\star h$; some projected
checks may vanish or become dependent, so an $\F_q$-independent basis of their span is measured,
and that choice affects only the measurement schedule, two equally weighted measured bases of one
check span generating the same stabilizer group with syndromes related by a fixed invertible
matrix.
\end{corollary}

\begin{proof}
Equation~\eqref{eq:projected-trace-syndrome} is \eqref{eq:adjoint-star} with $e^{\dagger}=e$, and
the special case follows from $e\star h=h$. The converse statements are
Theorem~\ref{thm:projected-recoverability} at $U=\Fqm$, respectively
\eqref{eq:relative-alignment} and (a)$\Leftrightarrow$(d), with $e^{\dagger}=e$; the closure
properties hold because $\Fix(e)^n$ is an $\F_q$-subspace and \eqref{eq:relative-alignment}
constrains each member separately.
\end{proof}

Two configurations recur below. A CSS pair is \emph{trace-aligned and complementary} at a
nontrivial trace-self-adjoint idempotent $e_1$, with $e_2:=\mathrm{id}-e_1$, if
$S_X\subseteq\Imag(e_1)^n$ and $S_Z\subseteq\Imag(e_2)^n$, and it is \emph{saturated} if the
first containment is an equality.

Corollary~\ref{cor:aligned-canonical} directs the checks into projector images. For the result to be
a stabilizer code the two images must be trace-orthogonal, which is what self-adjointness provides.
Let $e_1$ be an $\F_q$-linear self-adjoint idempotent and $e_2:=\mathrm{id}-e_1$; then $e_2$ is also
a self-adjoint idempotent and $e_1\circ e_2=e_2\circ e_1=0$.

\begin{lemma}[Trace orthogonality of complementary images]
\label{lem:complementary-image-orthogonality}
$\Imag(e_1)\ptr\Imag(e_2)$, that is, $\Tr(uv)=0$ for all $u\in\Imag(e_1)$ and $v\in\Imag(e_2)$.
\end{lemma}

\begin{proof}
Write $u=e_1(x)$, $v=e_2(y)$. By \eqref{eq:trace-adjoint} and $e_1^{\dagger}=e_1$,
$\Tr(e_1(x)e_2(y))=\Tr\bigl(x\,e_1(e_2(y))\bigr)=0$ because $e_1\circ e_2=0$.
\end{proof}

\subsection{Existence without basis hypotheses}
\label{subsec:idempotent-existence}

For an $\F_q$-subspace $V\subseteq\Fqm$ write $V^{\perp}$ for its orthogonal complement under
$B(x,y):=\Tr(xy)$, and call $V$ \emph{nondegenerate} if $V\cap V^{\perp}=\{0\}$; then
$\Fqm=V\oplus V^{\perp}$ orthogonally.

\begin{proposition}[Idempotent--subspace correspondence]
\label{prop:idempotent-subspace}
The assignment $V\mapsto e_V$, where $e_V$ is the projection onto $V$ along $V^{\perp}$, is a
bijection between the rank-$r$ nondegenerate $\F_q$-subspaces of $\Fqm$ and the $\F_q$-linear
self-adjoint idempotents of rank $r$, with inverse $e\mapsto\Imag(e)$; moreover $\Imag(e_V)=V$ and
$\ker(e_V)=V^{\perp}$.
\end{proposition}

\begin{proof}
If $V$ is nondegenerate then $\Fqm=V\oplus V^{\perp}$, so $e_V$ is a well-defined idempotent with
the stated image and kernel, and it is self-adjoint: for $x=v_1+w_1$, $y=v_2+w_2$ with $v_i\in V$,
$w_i\in V^{\perp}$, one has $B(e_Vx,y)=B(v_1,v_2)=B(x,e_Vy)$. Conversely let $e$ be a self-adjoint
idempotent of rank $r$ and $V:=\Imag(e)$. Nondegeneracy of $B$ gives
$\Imag(T)^{\perp}=\ker(T^{\dagger})$ for every $\F_q$-linear $T$; with $T=e=e^{\dagger}$ this yields
$V^{\perp}=\ker(e)$, and idempotence gives $\Fqm=V\oplus V^{\perp}$, so $V$ is nondegenerate and
$e=e_V$.
\end{proof}

\begin{lemma}[Unconditional existence of self-adjoint idempotents]
\label{lem:idempotent-existence}
For every prime power $q$, every $m\ge1$ and every $0\le r\le m$ there is an $\F_q$-linear
self-adjoint idempotent $e:\Fqm\to\Fqm$ with $\dim_{\F_q}\Imag(e)=r$. Equivalently, $(\Fqm,B)$
admits nondegenerate $\F_q$-subspaces of every dimension.
\end{lemma}

\begin{proof}
By Proposition~\ref{prop:idempotent-subspace} it suffices to produce nondegenerate subspaces of
every dimension. We first show that $B$ is non-alternating, that is, that $Q(x):=\Tr(x^2)$ is not identically zero. If
$\mathrm{char}\,\F_q\neq2$ and $Q\equiv0$, polarization gives
$B(x,y)=\tfrac12\bigl(Q(x+y)-Q(x)-Q(y)\bigr)=0$ for all $x,y$, contradicting nondegeneracy. If
$\mathrm{char}\,\F_q=2$, the Frobenius identity gives $Q(x)=\Tr(x)^2$, which is nonzero for any $x$
with $\Tr(x)\neq0$, and such $x$ exists because the trace is surjective.

A nondegenerate non-alternating symmetric form on a space $V$ admits nondegenerate subspaces of
every dimension $0\le r\le\dim V$, by induction on $\dim V$, the case $\dim V\le1$ being
immediate. Choose $v$ with $B(v,v)\neq0$, so that $V=\langle v\rangle\oplus v^{\perp}$ with
$B|_{v^{\perp}}$ nondegenerate. If $B|_{v^{\perp}}$ is non-alternating, the inductive hypothesis
realizes every dimension up to $\dim V-1$ inside $v^{\perp}$, and $V$ itself realizes $\dim V$. If
it is alternating it is symplectic, so $\dim V-1$ is even and $v^{\perp}$ has nondegenerate
subspaces of every even dimension up to $\dim V-1$; for a target $r$ take such a subspace when $r$
is even, the orthogonal sum of $\langle v\rangle$ with one of dimension $r-1$ when $r$ is odd and
$r<\dim V$, and $V$ when $r=\dim V$, an orthogonal sum of nondegenerate subspaces being
nondegenerate.
\end{proof}

If $\Fqm/\F_q$ admits a self-dual basis then $G_B=I$ and a self-adjoint idempotent is a symmetric
idempotent matrix. Such bases exist exactly when $q$ is even, or $q$ and $m$ are both odd
\cite{seroussi1980}. Self-dual \emph{normal} bases are subject to more restrictive conditions, existing for even $q$ exactly when $m$ is odd or $m\equiv2\pmod4$ \cite{lempel1988,bayerlenstra1990}, and it is that stronger notion which \cite{delfosse2024stacked} and Section~\ref{subsec:qgab} use. A symmetric idempotent need
not be orthogonally equivalent to a coordinate projector; what holds for every self-adjoint
idempotent is the split-basis form
$G_{B'}=G_V\oplus G_W$.

\begin{remark}[Explicit construction and the qubit instance]
\label{rem:explicit-idempotents}
A rank-one self-adjoint idempotent is obtained for any $v$ with $\Tr(v^2)\neq0$ by
$e_1(x)=\Tr(vx)\,v/\Tr(v^2)$, so that $\Imag(e_1)=\langle v\rangle$. For $q=2$, $m=5$ the degree is
odd, so $\Tr(1)=1$ and $v=1$ is admissible, giving the trace projector
$e_1(x)=\Tr_{\F_{2^5}/\F_2}(x)\cdot1$ with $\Imag(e_1)=\langle1\rangle$ and $\Imag(e_2)=\ker\Tr$ of
dimensions $1$ and $4$.
\end{remark}

\subsection{The trace-compatible CSS split}
\label{subsec:css-split}

\begin{theorem}[CSS admissibility of complementary idempotent images]
\label{thm:css-admissibility}
Let $e_1$ be an $\F_q$-linear self-adjoint idempotent and $e_2:=\mathrm{id}-e_1$. Define
$C_X:=\Imag(e_1)^n$ and $C_Z:=\Imag(e_2)^n$. Then $C_Z=C_X^{\ptr}$ and $C_X=C_Z^{\ptr}$; in
particular any implemented families with $X$-checks in $C_X$ and $Z$-checks in $C_Z$ commute under
the trace metric \cite{ashikhmin2001,ketkar2006}.
\end{theorem}

\begin{proof}
Lemma~\ref{lem:complementary-image-orthogonality} gives $C_Z\subseteq C_X^{\ptr}$. Every $a\in\Fqm$
decomposes as $a=e_1(a)+e_2(a)$, and $\Imag(e_1)\cap\Imag(e_2)=\{0\}$, since $u=e_1(a)=e_2(b)$
implies $u=e_1(u)=e_1(e_2(b))=0$; hence
\begin{equation}
\label{eq:field-direct-sum}
\Fqm=\Imag(e_1)\oplus\Imag(e_2)\quad\text{as $\F_q$-vector spaces},
\end{equation}
so that with $r:=\dim_{\F_q}\Imag(e_1)$ one has $\dim_{\F_q}\Imag(e_2)=m-r$ and
\begin{equation}
\label{eq:CX-CZ-dims}
\dim_{\F_q}C_X=nr,\qquad\dim_{\F_q}C_Z=n(m-r).
\end{equation}
By \eqref{eq:trace-dual-dimension}, $\dim_{\F_q}(C_X^{\ptr})=nm-nr=\dim_{\F_q}C_Z$, so
$C_Z=C_X^{\ptr}$; taking trace duals again gives $C_X=C_Z^{\ptr}$.
\end{proof}

Every statement below is invariant under the simultaneous exchange $X\leftrightarrow Z$ and
$e_1\leftrightarrow e_2$, by symmetry of the trace pairing, of Theorem~\ref{thm:css-admissibility}
in $(e_1,e_2)$, and of the two conventions of Proposition~\ref{prop:sector-split-embedding}; we
protect the $X$ sector throughout and call this sector interchange.

\subsection{The logical cost of an ambient quotient-valued interface}
\label{subsec:quotient-cost}

Definition~\ref{def:quotient-interface} asks only for the projected error modulo a subspace of
harmless labels, which for stabilizer correction is the check space of the corrected sector. The
next theorem determines exactly when such an interface exists on the whole label space, and the following corollary quantifies its requirement in terms of logical dimension.

\begin{theorem}[Characterization of ambient quotient-valued interfaces]
\label{thm:quotient-characterization}
Let $e$ be an $\F_q$-linear trace-self-adjoint idempotent on $\Fqm$ of rank $\rho$, acting
componentwise, and let $(S_X,S_Z)$ be trace-orthogonal $\F_q$-linear check spans, so that
$S_X\subseteq S_Z^{\ptr}$. Put
\[
K_{\Imag(e)}:=n\rho-\dim_{\F_q}\bigl(S_X\cap\Imag(e)^n\bigr)-\dim_{\F_q}\bigl(S_Z\cap\Imag(e)^n\bigr).
\]
The following are equivalent.
\begin{enumerate}
\item[(a)] The measured $Z$-check family admits an $S_X$-valued projected interface for $e$ on the
whole label space, that is, by Theorem~\ref{thm:quotient-interface}(a),
\begin{equation}
\label{eq:quotient-collapse-crit}
e\star S_Z^{\ptr}\subseteq S_X.
\end{equation}
\item[(b)] $S_X$ and $S_Z$ are both $e$-invariant and $K_{\Imag(e)}=0$.
\end{enumerate}
Under either condition the code splits along $e$ as in Theorem~\ref{thm:invariant-splitting} and
$K_{\Imag(e)}$ is the logical dimension of the branch supported on $\Imag(e)^n$; thus (a) holds
exactly when the projected branch carries no logical qudit. If moreover $K_q>0$, all logical
operators are supported on the $\ker(e)$-branch, to which Theorem~\ref{thm:rank-singleton} then applies on
its own layout. For a trace-aligned complementary
pair with $e=e_2$ one has $\rho=m-r$, so in the saturated regime condition~(a) is incompatible
with $K_q>0$; the mirrored statement holds by sector interchange.
\end{theorem}

\begin{proof}
(a)$\Rightarrow$(b). CSS commutation gives $S_X\subseteq S_Z^{\ptr}$, so by
\eqref{eq:quotient-collapse-crit}, $e\star S_X\subseteq e\star S_Z^{\ptr}\subseteq S_X$: the space
$S_X$ is $e$-invariant. Also $e\star S_Z^{\ptr}\subseteq S_X\subseteq S_Z^{\ptr}$, so $S_Z^{\ptr}$
is $e$-invariant. If $A$ is $e$-invariant and $y\in A^{\ptr}$, the adjoint identity
\eqref{eq:adjoint-star} with $e=e^{\dagger}$ gives
$\langle a,e\star y\rangle_{\Tr}=\langle e\star a,y\rangle_{\Tr}=0$ for every $a\in A$, so
$A^{\ptr}$ is $e$-invariant; applying this to $A=S_Z^{\ptr}$ and using
$(S_Z^{\ptr})^{\ptr}=S_Z$ from \eqref{eq:trace-dual-dimension} shows that $S_Z$ is $e$-invariant.

For $K_{\Imag(e)}=0$: by Theorem~\ref{thm:invariant-splitting} an $e$-invariant subspace $A$
satisfies $A=(A\cap\Imag(e)^n)\oplus(A\cap\ker(e)^n)$, whence $e\star A=A\cap\Imag(e)^n$. Applied
to $A=S_Z^{\ptr}$, hypothesis \eqref{eq:quotient-collapse-crit} reads
$S_Z^{\ptr}\cap\Imag(e)^n\subseteq S_X\cap\Imag(e)^n$; the reverse inclusion is CSS commutation, so
\begin{equation}
\label{eq:branch-equality}
S_Z^{\ptr}\cap\Imag(e)^n=S_X\cap\Imag(e)^n .
\end{equation}
The trace form is nondegenerate on $\Imag(e)^n$, of $\F_q$-dimension $n\rho$, by
Proposition~\ref{prop:idempotent-subspace}; and since $S_Z$ is $e$-invariant with
$\Imag(e)^n\ptr\ker(e)^n$ (Lemma~\ref{lem:complementary-image-orthogonality} applied to $e$ and
$\mathrm{id}-e$), an element of $\Imag(e)^n$ is trace-orthogonal to $S_Z$ exactly when it is
trace-orthogonal to $S_Z\cap\Imag(e)^n$. Hence $S_Z^{\ptr}\cap\Imag(e)^n$ is the orthogonal
complement of $S_Z\cap\Imag(e)^n$ inside $\Imag(e)^n$ and
\begin{equation}
\label{eq:branch-dim}
\dim_{\F_q}\bigl(S_Z^{\ptr}\cap\Imag(e)^n\bigr)=n\rho-\dim_{\F_q}\bigl(S_Z\cap\Imag(e)^n\bigr).
\end{equation}
Combining \eqref{eq:branch-equality} with \eqref{eq:branch-dim} gives
$\dim_{\F_q}(S_X\cap\Imag(e)^n)+\dim_{\F_q}(S_Z\cap\Imag(e)^n)=n\rho$, that is
$K_{\Imag(e)}=0$.

(b)$\Rightarrow$(a). Because $S_Z$ is $e$-invariant, so is $S_Z^{\ptr}$ by the adjoint argument
above, and Theorem~\ref{thm:invariant-splitting} gives
$e\star S_Z^{\ptr}=S_Z^{\ptr}\cap\Imag(e)^n$. The identity \eqref{eq:branch-dim} holds as before,
using only $e$-invariance of $S_Z$, so
$\dim_{\F_q}(S_Z^{\ptr}\cap\Imag(e)^n)=n\rho-\dim_{\F_q}(S_Z\cap\Imag(e)^n)$, which by
$K_{\Imag(e)}=0$ equals $\dim_{\F_q}(S_X\cap\Imag(e)^n)$. CSS commutation gives the inclusion
$S_X\cap\Imag(e)^n\subseteq S_Z^{\ptr}\cap\Imag(e)^n$, and equal dimensions turn it into
\eqref{eq:branch-equality}. Hence $e\star S_Z^{\ptr}=S_X\cap\Imag(e)^n\subseteq S_X$, which is
\eqref{eq:quotient-collapse-crit}.

\emph{The remaining assertions.} Under (b), the splitting and the branch decomposition follow from
Theorem~\ref{thm:invariant-splitting}, and $K_{\Imag(e)}$ is by construction the logical count of
the branch supported on the $n\rho$ qudits of $\Imag(e)^n$. If $K_q>0$ then $K_q=K_{\Imag(e)}+K_{\ker(e)}=K_{\ker(e)}>0$, so all logical operators are supported on the
single branch $\ker(e)$, of intra-carrier dimension $m-\rho$. In the saturated trace-aligned
complementary pair with $e=e_2$ one has $S_X=\Imag(e_1)^n$ and $S_Z\subseteq\Imag(e_2)^n$; then
$e_2\star S_Z^{\ptr}\subseteq S_X\cap\Imag(e_2)^n=\{0\}$ by \eqref{eq:field-direct-sum}, so
$S_Z^{\ptr}\subseteq\ker(e_2\star\cdot)=\Imag(e_1)^n$ and therefore
$S_Z\supseteq\Imag(e_1)^{n\,\ptr}=\Imag(e_2)^n$, forcing $S_Z=\Imag(e_2)^n$ and $K_q=0$.
\end{proof}

\begin{corollary}[Cost of an ambient quotient-valued interface]
\label{cor:interface-cost}
Let $S_X,S_Z$ be trace-orthogonal and $e$-invariant with $K_{\Imag(e)}>0$. If $S'_X\supseteq S_X$ is
trace-orthogonal to $S_Z$ and the pair $(S'_X,S_Z)$ satisfies \eqref{eq:quotient-collapse-crit},
then $\dim_{\F_q}S'_X-\dim_{\F_q}S_X\ge K_{\Imag(e)}$. The minimum is attained by
$S'_X:=(S_Z\cap\Imag(e)^n)^{\perp_{\Imag(e)^n}}\oplus(S_X\cap\ker(e)^n)$, the complement being
taken inside $\Imag(e)^n$, and the resulting code has $K'_q=K_q-K_{\Imag(e)}=K_{\ker(e)}$. The
$K_{\Imag(e)}$ additional checks are independent new stabilizer generators, not redundant
measurements of the existing group: the existence of the interface is achieved at the expense of modifying the code, incurring a reduction in logical dimension equal to that of the projected branch.
\end{corollary}

\begin{proof}
By Theorem~\ref{thm:quotient-characterization} applied to $(S'_X,S_Z)$, the space $S'_X$ is
$e$-invariant with $K'_{\Imag(e)}=0$, so by \eqref{eq:branch-dim}
\[
\begin{aligned}
\dim_{\F_q}(S'_X\cap\Imag(e)^n)&=n\rho-\dim_{\F_q}(S_Z\cap\Imag(e)^n)\\
&=\dim_{\F_q}(S_X\cap\Imag(e)^n)+K_{\Imag(e)} .
\end{aligned}
\]
Since $S'_X$ is $e$-invariant, $\dim S'_X=\dim(S'_X\cap\Imag(e)^n)+\dim(S'_X\cap\ker(e)^n)$ and
$S'_X\supseteq S_X$ gives $\dim(S'_X\cap\ker(e)^n)\ge\dim(S_X\cap\ker(e)^n)$, whence
$\dim S'_X\ge\dim S_X+K_{\Imag(e)}$. The displayed $S'_X$ contains $S_X$ by
\eqref{eq:branch-dim} and CSS commutation in the $\Imag(e)$-branch, is trace-orthogonal to $S_Z$
branch by branch, is $e$-invariant by construction, and has $K'_{\Imag(e)}=0$; so it satisfies (b)
and hence (a), and its dimension is $\dim S_X+K_{\Imag(e)}$. The logical count follows from
\eqref{eq:logical-count}.
\end{proof}

\subsection{One-sidedness of trace-aligned complementary designs}
\label{subsec:one-sidedness}

A branch of an invariant splitting is itself a matrix-label CSS pair, so
Theorem~\ref{thm:rank-singleton} bounds its rank distances. The next theorem shows that the bound
collapses to the value one under the additional hypothesis that the two sectors occupy complementary images, which is what the
post-processing-free design of Section~\ref{subsec:aligned} imposes.

\begin{theorem}[One-sidedness of trace-aligned complementary CSS codes]
\label{thm:no-go}
Let $e_1$ be an $\F_q$-linear self-adjoint idempotent of rank $r$, $V=\Imag(e_1)$,
$W=\Imag(e_2)=V^{\ptr}$, and consider a CSS code with $S_X\subseteq V^n$ and $S_Z\subseteq W^n$ and
$K_q=nm-\dim_{\F_q}S_X-\dim_{\F_q}S_Z$. Then (i) $S_X\subsetneq V^n$ implies $d^{R}_X=1$; (ii)
$S_Z\subsetneq W^n$ implies $d^{R}_Z=1$; (iii) $K_q>0$ implies $\min(d^{R}_X,d^{R}_Z)=1$; and (iv)
if both containments are proper then $d^{R}_X=d^{R}_Z=1$. In particular no trace-aligned
complementary CSS code with positive logical dimension protects both Pauli sectors in the rank
metric, for any $q$, $m$, $n$, $r$ and any choice of the two check spaces.
\end{theorem}

\begin{proof}
(i) By Lemma~\ref{lem:rankone-span} the rank-one vectors of $V^n$ span $V^n$; for if they were all
contained in $S_X$, it would follow that $V^n\subseteq S_X$. Choose a rank-one $x\in V^n\setminus S_X$. By
Lemmas~\ref{lem:complementary-image-orthogonality} and \ref{lem:physical-normalizer},
$\langle x,z\rangle_{\Tr}=0$ for every $z\in W^n\supseteq S_Z$, so $x\in S_Z^{\ptr}$ and $X(x)$
commutes with every $Z$-stabilizer; since $x\notin S_X$ it represents a nontrivial $X$-logical class
of rank weight one, and rank weights are positive integers. (ii) is the same argument with
$(e_1,V,X)$ and $(e_2,W,Z)$ exchanged. (iii) By \eqref{eq:CX-CZ-dims},
$\dim_{\F_q}V^n+\dim_{\F_q}W^n=nm$; if both containments were equalities then
$\dim S_X+\dim S_Z=nm$ and $K_q=0$. (iv) is immediate from (i) and (ii).
\end{proof}

\subsection{Structure of the saturated codes}
\label{subsec:frozen}

\begin{theorem}[Frozen-direction decomposition]
\label{thm:frozen-product}
Assume the saturated regime $S_X=\Imag(e_1)^n=V^n$ with $\dim_{\F_q}V=r$ and
$S_Z\subseteq W^n=\Imag(e_2)^n$, and let $Q$ be the realized CSS code on $N=nm$ physical
$\F_q$-qudits. Let $B'$ be a split $\F_q$-basis of $\Fqm=V\oplus W$ and
$M\in\mathrm{GL}_m(\F_q)$ the change of basis from $B$ to $B'$, applied identically at every
carrier. Then, after the carrier-wise identical Clifford circuit implementing $M$,
\[
Q=\ket{+}^{\otimes nr}\otimes Q_{\mathrm{one}}(S_Z),
\]
where $Q_{\mathrm{one}}(S_Z)$ is the CSS code on the remaining $n(m-r)$ qudits whose stabilizer
group is generated by $Z$-type operators only, namely the realized operators of
$S_Z\subseteq W^n\cong\F_q^{(m-r)\times n}$. It has $K_q=n(m-r)-\dim_{\F_q}S_Z$ logical qudits,
$X$-rank-distance $d_R(S_Z^{\perp_2})$ with $S_Z^{\perp_2}:=S_Z^{\ptr}\cap W^n$, and
$Z$-rank-distance $1$ whenever $K_q>0$. Moreover $M$ preserves every rank-metric quantity of Corollary~\ref{cor:rank-invariance} and maps
the set of rank-one labels bijectively onto itself.
\end{theorem}

\begin{proof}
In the split basis $G_{B'}=G_V\oplus G_W$ is block diagonal with both restrictions nondegenerate
(Proposition~\ref{prop:idempotent-subspace}), so for $x\in V^n$ both $[x_j]_{B'}$ and
$G_{B'}[x_j]_{B'}$ are supported on the first $r$ intra-carrier positions, and every such vector
arises this way. Hence $\varphi_X(V^n)$ is exactly the vectors supported on $A$, the realized
$X$-stabilizer group is generated by the $X_{(j,\ell)}$ with $(j,\ell)\in A$, those $nr$ qudits are
in the state $\ket{+}$ and factor out, and the remaining generators are the realized
$Z$-operators of $S_Z\subseteq W^n$.

For the parameters of the residual code, \eqref{eq:logical-count} gives
$K_q=nm-nr-\dim_{\F_q}S_Z=n(m-r)-\dim_{\F_q}S_Z$. Since $S_X=V^n$,
Lemma~\ref{lem:physical-normalizer}
and Theorem~\ref{thm:css-admissibility} identify the $X$-logical labels with
$S_Z^{\ptr}$ modulo $V^n$. Every $y\in S_Z^{\ptr}$ splits as $y=e_1\star y+e_2\star y$ with
$e_1\star y\in V^n\subseteq S_Z^{\ptr}$ by Lemma~\ref{lem:complementary-image-orthogonality},
whence $e_2\star y\in S_Z^{\perp_2}$ and
\begin{equation}
\label{eq:perp-decomposition}
S_Z^{\ptr}=V^n\oplus S_Z^{\perp_2},
\end{equation}
the sum being direct by \eqref{eq:field-direct-sum}. Thus every $X$-logical coset has a unique
representative in $S_Z^{\perp_2}$, and for $\ell\in S_Z^{\perp_2}$ and $s\in V^n$ one has
$e_2\star(\ell+s)=\ell$, so Lemma~\ref{lem:rank-non-expansion} gives
$\wt_R(\ell)\le\wt_R(\ell+s)$ with equality at $s=0$: the coset minimum is attained at the
$S_Z^{\perp_2}$-representative and $d^{R}_X=d_R(S_Z^{\perp_2})$. Dually the $Z$-logical labels are
$S_X^{\ptr}=W^n$ modulo $S_Z$; if $K_q>0$ then $S_Z\subsetneq W^n$, so by
Lemma~\ref{lem:rankone-span} some rank-one vector of $W^n$ lies outside $S_Z$ and $d^{R}_Z=1$.

Finally, $[x]_{B'}=M[x]_B$ gives $G_{B'}=M^{-\top}G_BM^{-1}$, so realized $Z$-coordinates
transform by $M$ and $X$-coordinates by $M^{-\top}$, each a single invertible matrix applied
identically at every carrier, and Corollary~\ref{cor:rank-invariance} preserves all ranks sector by
sector; a rank-one label $\bar u\,c^{\top}$ maps to $(M\bar u)c^{\top}$, a bijection of the burst
ensemble. The per-carrier label action $(x,z)\mapsto(M^{-\top}x,Mz)$ preserves the symplectic form,
hence is Clifford for every prime power \cite{ketkar2006,gottesman1999higher,hostens2005}, and for
$q=2$ the factorization of $M$ into elementary row operations realizes it by
controlled-\textsc{not} gates in depth $O(m^2)$ \cite{gottesman1997}.
\end{proof}

In particular the uniquely correctable component $e_1\star E_Z$ of a $Z$-error is supported on the
$nr$ frozen directions, which carry no logical information, so no value of $r$ confers logical-$Z$
protection, while these physical qudits incur an unmitigated spatial overhead.

\section{Two Witness Families}
\label{sec:witnesses}

The equality structure of Theorem~\ref{thm:singleton-equality-mrd} has two extreme specializations,
$u=v$ and $v=0$, and both are realized by explicit families. The two-sided one, the quantum
Gabidulin codes of \cite{delfosse2024stacked}, shows that the unrestricted bounded-rank interface
classes are strictly larger than the linear ones, and carries a certifying decoder at the optimal
radius. The one-sided one is generated by projecting a Gabidulin code with a trace-self-adjoint
idempotent, and is the family on which the one-sidedness of
Theorem~\ref{thm:no-go} is exact.

\subsection{Two-sided: quantum Gabidulin codes}
\label{subsec:qgab}

Throughout this section $q=2$, $F:=\F_{2^n}$, $m=n$, and
$\boldsymbol\alpha=(\beta,\beta^2,\dots,\beta^{2^{n-1}})$ is a self-dual normal basis of
$F/\F_2$, existing if and only if $n$ is odd or $n\equiv2\pmod4$ \cite{lempel1988}. For
$i\in\mathbb Z/n\mathbb Z$ let $a_i:=\boldsymbol\alpha^{[2^i]}$, the cyclic shift of
$\boldsymbol\alpha$ by $i$, and for $I\subseteq\mathbb Z/n\mathbb Z$ put
$C_I:=\mathrm{span}_F\{a_i:i\in I\}$. Evaluating at $\alpha_j=\beta^{2^{j}}$ gives
$a_i=\mathrm{ev}(z^{2^i})$ and $C_I=\mathrm{ev}(\mathrm{span}_F\{z^{2^i}:i\in I\})$, exponents
read modulo $n$; in particular $C_{\{0,\dots,k-1\}}=\mathrm{Gab}(\boldsymbol\alpha,k)$.

\begin{lemma}[Orthonormality and duals]
\label{lem:normal-basis-duals}
For all $i,j$ one has $a_i\cdot a_j=\delta_{ij}$, where $x\cdot y:=\sum_kx_ky_k$ denotes the
$F$-valued dot product; hence $(a_i)_{i}$ is an $F$-orthonormal basis of $F^n$. For an $F$-linear
code $C\subseteq F^n$ the trace dual coincides with the dual with respect to the dot product, and
consequently $C_I^{\ptr}=C_{I^{c}}$ for every $I\subseteq\mathbb Z/n\mathbb Z$.
\end{lemma}

\begin{proof}
By definition $a_i\cdot a_j=\sum_{k}\beta^{2^{i+k}}\beta^{2^{j+k}}
=\sum_k\bigl(\beta^{2^i}\beta^{2^j}\bigr)^{2^k}=\Tr\bigl(\beta^{2^i}\beta^{2^j}\bigr)=\delta_{ij}$,
the middle equality using that Frobenius is a field automorphism. If $C$ is $F$-linear and
$y\in C^{\ptr}$, then for every $c\in C$ and every $\lambda\in F$ one has $\lambda c\in C$, so
$\Tr\bigl(\lambda\,(c\cdot y)\bigr)=\langle\lambda c,y\rangle_{\Tr}=0$; nondegeneracy of the trace
form on $F$ forces $c\cdot y=0$. The converse inclusion is immediate from
$\langle c,y\rangle_{\Tr}=\Tr(c\cdot y)$. Since the $a_i$ are $F$-orthonormal, the dot-product dual
of $C_I$ is the $F$-span of the $a_j$ with $j\notin I$, which is $C_{I^{c}}$.
\end{proof}

\begin{lemma}[Rank distance of cyclically consecutive spans]
\label{lem:consecutive-spans}
Let $I=\{s,s+1,\dots,s+k-1\}\subseteq\mathbb Z/n\mathbb Z$ with $1\le k\le n$. Then
$d_R(C_I)=n-k+1$, and the value is attained by $\mathrm{ev}\bigl(P_U^{2^s}\bigr)$ for any
$\F_2$-subspace $U\subseteq F$ of dimension $k-1$, where $P_U(z):=\prod_{u\in U}(z-u)$.
\end{lemma}

\begin{proof}
An element of $C_I$ is $\mathrm{ev}(f)$ with $f(z)=\sum_{i\in I}c_iz^{2^i}$ and $c_i\in F$. The
evaluation points form an $\F_2$-basis of $F$, so
$\mathrm{span}_{\F_2}\{f(\alpha_j)\}_j=\Imag(f)$ and
$\wt_R(\mathrm{ev}(f))=\dim_{\F_2}\Imag(f)=n-\dim_{\F_2}\ker f$. Since $z^{2^n}=z$ on $F$, the map
$g(z):=f(z)^{2^{n-s}}=\sum_{j=0}^{k-1}c_{s+j}^{2^{n-s}}z^{2^{j}}$ is an $\F_2$-linearized
polynomial of $q$-degree at most $k-1$ with $\ker g=\ker f$, because raising to the power $2^{n-s}$
is a bijection of $F$. A nonzero such $g$ has at most $2^{k-1}$ roots in $F$, so
$\dim_{\F_2}\ker f\le k-1$ and $\wt_R(\mathrm{ev}(f))\ge n-k+1$.

Conversely, for an $\F_2$-subspace $U\subseteq F$ of dimension $k-1$ the subspace polynomial $P_U$
is $\F_2$-linearized of $q$-degree $k-1$ with $\ker P_U=U$ \cite{gabidulin1985}, so
$f:=P_U^{2^s}$ is $\F_2$-linearized with exponent set contained in $I$ and $\ker f=U$; hence
$\mathrm{ev}(f)\in C_I$ has rank weight $n-(k-1)$.
\end{proof}

\begin{proposition}[Two-sidedness with a nonlinear bounded-rank interface at a nontrivial
idempotent]
\label{prop:separation}
Let $n\ge5$ satisfy $n$ odd or $n\equiv2\pmod4$, let $2\le r<n/2$, and set
$S_X:=C_{\{0,\dots,r-1\}}=\mathrm{Gab}(\boldsymbol\alpha,r)$ and
$S_Z:=C_{\{r,\dots,2r-1\}}=\mathrm{Gab}(\boldsymbol\alpha^{2^r},r)$,
which is the quantum Gabidulin code $\mathrm{QGab}(\boldsymbol\alpha,r,r)$ of
\cite{delfosse2024stacked}. Let $e$ be \emph{any} $\F_2$-linear idempotent on $F$ with
$1\le\rank(e)\le n-1$, self-adjoint or not; trace-self-adjoint idempotents of every rank exist by
Lemma~\ref{lem:idempotent-existence}, and for $n$ odd the rank-one choice $e(x)=\Tr(x)\cdot1$ of
Remark~\ref{rem:explicit-idempotents} is admissible. Write $p$ for whichever of $e$ and
$\mathrm{id}-e$ has the smaller rank, so that $s:=\rank(p)=\min(\rank(e),n-\rank(e))
\le\lfloor n/2\rfloor$. Then:
\begin{enumerate}
\item[(i)] $(S_X,S_Z)$ is trace-orthogonal, $\dim_{\F_2}S_X=\dim_{\F_2}S_Z=nr$,
$K_q=n(n-2r)>0$, and $d^{R}_X=d^{R}_Z=r+1$ exactly;
\item[(ii)] neither check span is invariant under $e$ or under $e^{\dagger}$, so for
trace-self-adjoint $e$ the code does not split along $e$
(Theorem~\ref{thm:invariant-splitting}); neither measured family admits an ambient
projected-syndrome interface at $e$, linear or not
(Theorem~\ref{thm:projected-recoverability} at $U=\Fqm$), and no ambient quotient-valued interface
exists in either sector. The pair attains \eqref{eq:rank-singleton} with equality and has
$d^{R}_X=d^{R}_Z$, so Corollary~\ref{cor:optimal-rigidity}(iii) applies; a direct verification uses
$1\le s\le\lfloor n/2\rfloor<n-r+1=d_R(S_X)=d_R(S_Z)$ and $\wt_R(p\star a_i)=s$, whence
$p\star a_0\notin S_X$; note that $d_R(S_X)=n-r+1$ is the \emph{check-space} distance, whereas the
logical distances are $d^{R}_X=d^{R}_Z=r+1$;
\item[(iii)] $d_R(S_Z^{\ptr})=d_R(S_X^{\ptr})=r+1$, so for every $t\le\lfloor r/2\rfloor$ both
sectors satisfy the minimum-rank certificate and admit $\mathrm{I}_5$ and $\mathrm{I}_6$
(Corollary~\ref{cor:optimal-rigidity}), for every admissible $e$ at once, the
certificate being a property of the code alone; and by
Corollaries~\ref{cor:linearity-dichotomy} and \ref{cor:quotient-linearity} none is given by an
$\F_2$-linear map on the ambient label space.
\end{enumerate}
In the terminology of Definition~\ref{def:interfaces}, conditions $\mathrm{I}_5$ and
$\mathrm{I}_6$ hold in both sectors at every nontrivial idempotent for a two-sided code that is not a
tensor product across the $e$-split, while $\mathrm{I}_1,\dots,\mathrm{I}_4$ all fail. The
unrestricted bounded-rank classes are therefore strictly larger than the linear ones already at
$n=5$, and by Corollary~\ref{cor:linearity-dichotomy} the gap is created by nonlinearity of the
recovery map and not by the restriction to a rank ball. The conclusion is uniform in $\rank(e)$. By
Corollary~\ref{cor:optimal-rigidity} this behavior is not a feature of the quantum
Gabidulin family but of Singleton optimality itself.
\end{proposition}

\begin{proof}
(i) The index sets $\{0,\dots,r-1\}$ and $\{r,\dots,2r-1\}$ are disjoint because $2r<n$, so
$S_Z\subseteq S_X^{\ptr}$ by Lemma~\ref{lem:normal-basis-duals}. Each span has $F$-dimension $r$,
hence $\F_2$-dimension $nr$, and \eqref{eq:logical-count} gives $K_q=n(n-2r)>0$.

By Lemma~\ref{lem:normal-basis-duals}, $S_Z^{\ptr}=C_{I}$ with
$I=\{2r,\dots,n-1\}\cup\{0,\dots,r-1\}$ and $S_X^{\ptr}=C_{\{r,\dots,n-1\}}$, both cyclically
consecutive of length $n-r$, so Lemma~\ref{lem:consecutive-spans} gives
\begin{equation}
\label{eq:qgab-dual-distance}
d_R(S_Z^{\ptr})=d_R(S_X^{\ptr})=n-(n-r)+1=r+1,
\end{equation}
and likewise $d_R(S_X)=d_R(S_Z)=n-r+1$. Since $n>2r$ we have $n-r+1>r+1$, so an element of
$S_Z^{\ptr}$ of rank weight $r+1$, which exists by Lemma~\ref{lem:consecutive-spans}, cannot lie in
$S_X$ and labels a nontrivial $X$-logical class by Lemma~\ref{lem:physical-normalizer}; with $d^{R}_X\ge d_R(S_Z^{\ptr})$, valid because every $X$-logical label lies in $S_Z^{\ptr}$, this gives
$d^{R}_X=r+1$, and by symmetry, $d^{R}_Z=r+1$.

(ii) By (i) the pair attains \eqref{eq:rank-singleton} with equality, $n^2-2nr=n(n-2r)=K_q$, with
$d^{R}_X=d^{R}_Z=r+1\ge3$, so Corollary~\ref{cor:optimal-rigidity} gives every assertion, its
hypothesis $m\le n$ holding with $m=n$. The direct verification is equally short: the coordinates
of $a_i$ form an $\F_2$-basis of $F$, so those of $p\star a_i$ span $\Imag(p)$ and
$\wt_R(p\star a_i)=s\le\lfloor n/2\rfloor<n-r+1=d_R(S_X)$, whence $p\star a_0\neq0$ cannot lie in $S_X$ although $a_0\in S_X$, and an identical argument using $a_r\in S_Z$ shows that $p\star a_r\notin S_Z$; invariance under $e$ and under $\mathrm{id}-e$ are equivalent, and $e^{\dagger}$ is again an idempotent of rank $\rank(e)$.

(iii) The distances are \eqref{eq:qgab-dual-distance}, in agreement with
Corollary~\ref{cor:optimal-rigidity}, and $2t\le r<r+1$ for
$t\le\lfloor r/2\rfloor$, so that corollary supplies both interfaces in both sectors. Were either
given by an $\F_2$-linear map on the ambient space,
Corollary~\ref{cor:linearity-dichotomy} respectively Corollary~\ref{cor:quotient-linearity} would
force the ambient criterion, contradicting (ii).
\end{proof}

The family of Proposition~\ref{prop:separation} also admits an efficient certifying decoder at the
optimal radius, which we record because \cite{delfosse2024stacked} lists an efficient decoder as an
open requirement for these codes. The statement is a code-capacity
result for the parent family, and uses nothing beyond Lemmas~\ref{lem:normal-basis-duals} and
\ref{lem:consecutive-spans}.

\begin{proposition}[Certifying decoder for the quantum Gabidulin codes]
\label{prop:qgab-decoder}
Let $(S_X,S_Z)=\mathrm{QGab}(\boldsymbol\alpha,r,r)$ be as in Proposition~\ref{prop:separation}.
\begin{enumerate}
\item[(i)] The $X$-sector kernel is
$D_X=S_Z^{\ptr}=C_{\{2r,\dots,n-1\}\cup\{0,\dots,r-1\}}
=\mathrm{Gab}(\boldsymbol\alpha^{2^{2r}},n-r)$, an $F$-linear Gabidulin code of length $n$,
$F$-dimension $n-r$ and minimum rank distance $r+1$, whose parity-check matrix is the $r\times n$
Moore matrix with rows $a_r,\dots,a_{2r-1}$. Symmetrically, $D_Z=S_X^{\ptr}$ has a parity-check matrix
with rows $a_0,\dots,a_{r-1}$.
\item[(ii)] Writing $S_j:=a_{s+j}\cdot E$ for $0\le j<r$, with $s=r$ in the $X$ sector and $s=0$ in
the $Z$ sector, the measured record of $E$ against the $\F_2$-basis
$\{\beta^{2^{\ell}}a_{s+j}\}_{j,\ell}$ of the corresponding check space is
$\bigl(\Tr(\beta^{2^{\ell}}S_j)\bigr)_{j,\ell}$, and self-duality of the basis recovers each $S_j$
as $S_j=\sum_{\ell}\Tr(\beta^{2^{\ell}}S_j)\,\beta^{2^{\ell}}$, with no matrix inversion.
\item[(iii)] For every $t\le\lfloor r/2\rfloor$ the following procedure is a certifying radius-$t$
decoder in the sense of Definition~\ref{def:certifying-decoder}, in either sector, and returns $E$
for every label of rank at most $t$:
\begin{enumerate}
\item solve over $F$ the $(r-t)\times t$ linear system
\begin{equation}
\label{eq:qgab-key}
\sum_{l=1}^{t}\lambda_l\,S_{j-l}^{\,2^{l}}=S_j,\qquad j=t,\dots,r-1;
\end{equation}
return $\bot$ if \eqref{eq:qgab-key} is inconsistent, and otherwise select one solution by setting
the free unknowns to zero and put $\Lambda(x):=x+\sum_{l=1}^{t}\lambda_lx^{2^{l}}$;
\item compute an $\F_2$-basis $u_1,\dots,u_{\kappa}$ of $\ker\Lambda$; in exact arithmetic
$\kappa\le t$ automatically, since the coefficient of $x$ in $\Lambda$ is $1$, so that $\Lambda$ is
a nonzero $\F_2$-linearized polynomial of $2$-degree at most $t$ and has at most $2^{t}$ roots;
the test $\kappa\le t$ is therefore a consistency check on the implementation rather than a
rejection mechanism;
\item solve over $\F_2$ the $rn$ syndrome equations for the coefficients $c_{ik}$ of
$\widehat E_k=\sum_ic_{ik}u_i$, returning $\bot$ if they are inconsistent;
\item return $\widehat E$ if $\wt_R(\widehat E)\le t$ and $a_{s+j}\cdot\widehat E=S_j$ for all $j$,
and $\bot$ otherwise.
\end{enumerate}
The procedure is defined on the whole syndrome domain $F^{r}$ and returns $\bot$ on every input
that is not the syndrome of a label of rank at most $t$: it rejects whenever either linear system,
that of step~(a) or that of step~(c), is inconsistent, or the final certification in step~(d)
fails. Step~(a) costs $O(r^2t)$ operations over $F$, and step~(c) requires solving a linear system of size
$\kappa n$ over $\F_2$. The procedure is a reduction of the measured binary trace syndrome to a
Gabidulin syndrome together with a certification, and not a new rank-decoding algorithm: any
Gabidulin syndrome decoder may be substituted for steps (a)--(c), at $O(n^2)$ operations over $F$
\cite{gabidulin1985,silva2009}.
\item[(iv)] The radius $\lfloor r/2\rfloor$ is the largest at which unique decoding is possible,
since $d^{R}_X=d^{R}_Z=r+1$ by Proposition~\ref{prop:separation}(i).
\end{enumerate}
\end{proposition}

\begin{proof}
(i) By Lemma~\ref{lem:normal-basis-duals}, $S_Z^{\ptr}=C_{I^{c}}$ with
$I^{c}=\{2r,\dots,n-1\}\cup\{0,\dots,r-1\}$, cyclically consecutive of length $n-r$ starting at
$2r$. For $f(z)=\sum_{i\in\{s,\dots,s+k-1\}}c_iz^{2^{i}}$ one has $f(z)=g(z^{2^{s}})$ with
$g(y)=\sum_{j<k}c_{s+j}y^{2^{j}}$, so $C_{\{s,\dots,s+k-1\}}=\mathrm{Gab}(\boldsymbol\alpha^{2^{s}},k)$;
Lemma~\ref{lem:consecutive-spans} gives $d_R(C_{I^{c}})=r+1$. Its dot-product dual is
$C_I=\mathrm{span}_F\{a_r,\dots,a_{2r-1}\}$ by Lemma~\ref{lem:normal-basis-duals}, and
$(a_{r+i})_k=(\beta^{2^{r+k}})^{2^{i}}$ exhibits a Moore matrix.

(ii) $\Tr(\beta^{2^{\ell}}(a_{s+j}\cdot E))=\langle\beta^{2^{\ell}}a_{s+j},E\rangle_{\Tr}$ is a
measured parity, and for a self-dual basis $(b_{\ell})$ every $S\in F$ satisfies
$S=\sum_{\ell}\Tr(Sb_{\ell})b_{\ell}$.

(iii) Let $\wt_R(E)=\tau\le t$, write $E_k=\sum_{i=1}^{\tau}c_{ik}\eta_i$ with $c_{ik}\in\F_2$ and
$\eta_1,\dots,\eta_{\tau}$ an $\F_2$-basis of $\mathrm{span}_{\F_2}\{E_k\}$, and put
$\gamma_i:=\sum_kc_{ik}\beta^{2^{k}}$. Since $c_{ik}\in\F_2$,
\[
\begin{aligned}
S_j&=\sum_k\beta^{2^{s+j+k}}E_k
=\sum_{i=1}^{\tau}\eta_i\Bigl(\sum_kc_{ik}\beta^{2^{k}}\Bigr)^{2^{s+j}}\\
&=\sum_{i=1}^{\tau}\eta_i\gamma_i^{2^{s+j}} .
\end{aligned}
\]
Raising $S_{j-l}$ to the power $2^{l}$ gives
$S_{j-l}^{2^{l}}=\sum_i\eta_i^{2^{l}}\gamma_i^{2^{s+j}}$, so with $\lambda_0:=1$,
\[
\sum_{l=0}^{t}\lambda_lS_{j-l}^{2^{l}}
=\sum_{i=1}^{\tau}\gamma_i^{2^{s+j}}\,\Lambda(\eta_i),\qquad j=t,\dots,r-1 .
\]
The $\gamma_i$ are $\F_2$-independent, because the coefficient vectors $c_i$ are and
$(\beta^{2^{k}})_k$ is a basis of $F$; hence the $(r-t)\times\tau$ Moore matrix
$(\gamma_i^{2^{s+j}})_{j,i}$ has full column rank, having $r-t\ge t\ge\tau$ rows
\cite{gabidulin1985}. Therefore every solution $\lambda$ of \eqref{eq:qgab-key} satisfies
$\Lambda(\eta_i)=0$ for all $i$. A solution exists: for any $\F_2$-subspace $U$ of dimension $t$
containing $\mathrm{span}_{\F_2}\{\eta_i\}$, the subspace polynomial $P_U$ is $\F_2$-linearized of
$2$-degree $t$ with nonzero coefficient of $x$, and its normalization has $\lambda_0=1$.
Consequently $\ker\Lambda\supseteq\mathrm{span}_{\F_2}\{\eta_i\}$ and $\kappa\le t$, so step~(b)
does not return $\bot$; step~(c) is consistent, the true coefficients being a solution; and two
solutions of step~(c) differ by an element of $D_X$ whose coordinates lie in $\ker\Lambda$, hence
of rank at most $\kappa\le t<r+1=d_R(D_X)$, so they coincide. Step~(d) verifies membership, the
rank bound and syndrome consistency, which are the three conditions of
Definition~\ref{def:certifying-decoder}; uniqueness is guaranteed a priori by
$d_R(D_X)=r+1>2t$ through Theorem~\ref{thm:exact-unique}, so those checks certify the output. The
substitution of a Gabidulin syndrome decoder is legitimate because, by (i), the $S_j$ are the
syndrome coordinates of $E$ for the parity-check matrix of
$\mathrm{Gab}(\boldsymbol\alpha^{2^{2r}},n-r)$.

(iv) Claim~(iv) follows directly from $d^{R}_X=d^{R}_Z=r+1$ and Theorem~\ref{thm:exact-unique}.
\end{proof}

For $r=2$ and $t=1$ the rejection branches are exact and not vacuous: the key equation
\eqref{eq:qgab-key} is inconsistent when $S_0=0\neq S_1$, and when $S_0\neq0=S_1$ it forces
$\lambda_1=0$, so that step~(c) is inconsistent unless the error vanishes; hence exactly
$1+(2^{n}-1)^{2}$ of the $2^{2n}$ sector records are decoded, namely those of the labels of rank at
most one, and the remaining $2(2^{n}-1)$ are rejected, $962$ against $62$ at $n=5$.

\begin{example}[Explicit separation instances]
\label{ex:separation}
Integers denote $\F_2$-coordinates in the basis $(1,x,\dots,x^{n-1})$, least significant bit first.
For $n=5$, $r=2$ take $F=\F_2[x]/(x^5+x^2+1)$, $\beta=1+x$,
$\boldsymbol\alpha=(3,5,17,12,26)$; then $\Tr(1)=1$, the rank-one trace projector is admissible,
$e\star a_0=e\star a_2=\mathbf1$, and the code is $\llbracket25,5,3^{R}/3^{R}\rrbracket_2$ with
$t=1$. For $n=6$, $r=2$, $F=\F_2[x]/(x^6+x+1)$ with $\beta=x^3+x^5$ gives
$\llbracket36,12,3^{R}/3^{R}\rrbracket_2$, using the existence case $n\equiv2\pmod4$
\cite{lempel1988}. Both attain
\eqref{eq:rank-singleton} with equality, and in the $n=5$ instance the degree and the
length are odd yet $d^{R}_X=d^{R}_Z=3$, the all-ones label being the obstruction to
$e$-invariance rather than a logical operator or a stabilizer.
\end{example}

\subsection{One-sided: idempotent-generated Gabidulin codes}
\label{subsec:igg}

The family of this subsection attains Theorem~\ref{thm:rank-singleton} with equality in its
logical-carrying branch and realizes the one-sided endpoint $v=0$ of
Theorem~\ref{thm:singleton-equality-mrd}.

\emph{Standing hypotheses for this subsection.} Throughout this subsection we take $n=m$, so the evaluation points form an $\F_q$-basis of $\Fqm$; $e_1$ is
a trace-self-adjoint idempotent with $V:=\Imag(e_1)$ of dimension $r$, $e_2:=\mathrm{id}-e_1$ and
$W:=\Imag(e_2)$ of dimension $a:=m-r$. With $0\le k_X,k_Z\le m$ and $\delta_X:=m-k_X$,
$\delta_Z:=m-k_Z$ in $[0,m]$ we set
$S_X:=\mathrm{span}_{\F_q}\{e_1\star\mathrm{ev}(f):f\in\mathcal L_{<\delta_X}\}$ and
$S_Z:=\mathrm{span}_{\F_q}\{e_2\star\mathrm{ev}(g):g\in\mathcal L_{<\delta_Z}\}$, aligned by
Corollary~\ref{cor:aligned-canonical} and trace-orthogonal by
Theorem~\ref{thm:css-admissibility}; the code is an \emph{idempotent-generated Gabidulin} (IGG)
code. The assumption $n=m$ is essential: \eqref{eq:sx-dimension} does not hold in general when
$n<m$, as demonstrated by the counterexample $q=2$, $m=3$, $n=2$, $r=1$, with evaluation points $(1,x)$ and $\delta_X=\delta_Z=1$, where $\dim_{\F_2}S_X=2$ while \eqref{eq:sx-dimension} would yield $3>\dim_{\F_2}\Imag(e_1)^n=2$.

\begin{theorem}[Projected dimensions, logical dimension and logical operators]
\label{thm:igg-parameters}
Under the assumptions above the following hold, with no genericity condition on $e_1$.
\begin{enumerate}
\item[(i)] \emph{Projected dimensions.}
\begin{equation}
\label{eq:sx-dimension}
\begin{split}
\dim_{\F_q}S_X&=\min(\delta_Xm,\,mr),\\
\dim_{\F_q}S_Z&=\min(\delta_Zm,\,m(m-r)),
\end{split}
\end{equation}
and $g\mapsto e_2\star\mathrm{ev}(g)$ is injective on $\mathcal L_{<\delta_Z}$ when
$\delta_Z\le m-r$, dually for $e_1$ when $\delta_X\le r$.
\item[(ii)] \emph{Logical dimension.} The IGG code has $N=nm$ physical $\F_q$-qudits and
$K_q=nm-\min(\delta_Xm,mr)-\min(\delta_Zm,m(m-r))$ logical ones; in the saturated regime
$\delta_Xm\ge nr$, $\delta_Zm\le n(m-r)$ this is $K_q=nm-nr-\delta_Zm$, and with $n=m$,
\begin{equation}
\label{eq:logical-dimension-saturated}
K_q=m(k_Z-r),\qquad K_{q^m}=k_Z-r.
\end{equation}
\item[(iii)] \emph{Logical operators.} In the saturated regime $S_X=\Imag(e_1)^n$,
$\dim_{\F_q}S_Z=\delta_Zm$ with $K_q>0$, set $S_Z^{\perp_2}:=S_Z^{\ptr}\cap\Imag(e_2)^n$. There are
exactly $K_q$ independent $X$-logical generators, with unique representatives in $S_Z^{\perp_2}$,
of dimension $n(m-r)-\dim_{\F_q}S_Z=K_q$ and minimum rank weight
$d^{R}_X=d_R(S_Z^{\perp_2})$, and exactly $K_q$ independent $Z$-logical generators with
representatives in $\Imag(e_2)^n$ modulo $S_Z$ and minimum rank weight $d^{R}_Z=1$. The saturation
hypothesis is needed for the equality $d^{R}_X=d_R(S_Z^{\perp_2})$: if
$S_X\subsetneq\Imag(e_1)^n$, then by Lemma~\ref{lem:rankone-span} a rank-one label of
$\Imag(e_1)^n$ outside $S_X$ is a nontrivial $X$-logical, so $d^{R}_X=1$ while
$d_R(S_Z^{\perp_2})$ is unconstrained.
\end{enumerate}
\end{theorem}

\begin{proof}
(i) We prove the formula for $S_X$; that for $S_Z$ follows on exchanging $(e_1,r)$ with
$(e_2,m-r)$. Let $\Phi_X(f):=e_1\star\mathrm{ev}(f)$ on $\mathcal L_{<\delta_X}$, with image $S_X$.
The $\alpha_j$ form an $\F_q$-basis, so $\mathrm{ev}$ is injective and
$\ker\Phi_X=\{f:\Imag(f)\subseteq V^{\ptr}\}$, since $e_1\star\mathrm{ev}(f)=0$ says that $f$ maps
a basis, hence all of $\Fqm$, into $\ker(e_1)=V^{\ptr}$
(Proposition~\ref{prop:idempotent-subspace}). If $\delta_X\le r$: a nonzero
$f\in\mathcal L_{<\delta_X}$ has kernel of dimension at most $\delta_X-1$, so
$\rank(f)\ge m-\delta_X+1>m-r$ \cite{gabidulin1985}, contradicting
$\Imag(f)\subseteq V^{\ptr}$; hence $\ker\Phi_X=0$ and $\dim_{\F_q}S_X=\delta_Xm$. If
$\delta_X>r$: the same argument gives $\ker\Phi_X\cap\mathcal L_{<r}=\{0\}$, so
$\dim\ker\Phi_X\le(\delta_X-r)m$; and
$\ker\Phi_X=\mathcal L_{<\delta_X}\cap\mathrm{Hom}_{\F_q}(\Fqm,V^{\ptr})$ intersects subspaces of
dimensions $\delta_Xm$ and $m(m-r)$ inside $\mathrm{End}_{\F_q}(\Fqm)$ of dimension $m^2$, so
$\dim\ker\Phi_X\ge(\delta_X-r)m$. Hence $\dim_{\F_q}S_X=mr$.

(ii) The realized group is abelian by Theorem~\ref{thm:css-admissibility} and scalar-free of order
$q^{\dim_{\F_q}S_X+\dim_{\F_q}S_Z}$ by Proposition~\ref{prop:sector-split-embedding}(iv), so
\eqref{eq:logical-count} holds for every prime power $q$; substituting \eqref{eq:sx-dimension}
gives \eqref{eq:logical-dimension-saturated}, and the saturated values $\dim_{\F_q}S_X=nr$,
$\dim_{\F_q}S_Z=\delta_Zm$ complete the verification.

(iii) In the saturated regime the hypotheses of Theorem~\ref{thm:frozen-product} hold, and its
proof establishes \eqref{eq:perp-decomposition}, the identification of the $X$-logical labels with
$S_Z^{\perp_2}$, the equality $d^{R}_X=d_R(S_Z^{\perp_2})$ and the value $d^{R}_Z=1$. The count
$\dim_{\F_q}S_Z^{\perp_2}=n(m-r)-\dim_{\F_q}S_Z=K_q$ follows from
\eqref{eq:perp-decomposition}, \eqref{eq:trace-dual-dimension} and (ii).
\end{proof}

\begin{theorem}[MRD property of the protected branch]
\label{thm:analytic-mrd}
Assume the standing hypotheses of Section~\ref{subsec:igg} with $\dim_{\F_q}S_Z=(m-k_Z)m$ and $r<k_Z<m$. Then $S_Z$, as an $a\times m$ matrix code, is MRD with minimum rank distance $a-(m-k_Z)+1$, and its restricted trace dual $S_Z^{\perp_2}$ is MRD with $d_R(S_Z^{\perp_2})=n-k_Z+1$. If moreover the pair is \emph{saturated}, that is $S_X=\Imag(e_1)^n$, then $d^{R}_X=d_R(S_Z^{\perp_2})=n-k_Z+1$ and $K_q=m(k_Z-r)$; for $k_Z=m$, $S_Z=\{0\}$ and $d^{R}_X=1=n-k_Z+1$, and for $k_Z=r$ the restricted trace dual is zero and $K_q=0$.

The saturation hypothesis is essential for the last two assertions and is not merely an artifact of the proof technique: without it, the restricted-dual distance need not coincide with the logical distance. For $q=2$, $m=n=4$, $r=2$ and $k_X=k_Z=3$, \eqref{eq:sx-dimension} gives $\dim_{\F_2}S_X=\dim_{\F_2}S_Z=4$, both properly inside the branch spaces of dimension $8$; then $K_q=8$ and $d_R(S_Z^{\perp_2})=2=m-k_Z+1$, whereas Theorem~\ref{thm:no-go} gives $d^{R}_X=d^{R}_Z=1$.
\end{theorem}

\begin{proof}
Expanding each codeword $e_2\star\mathrm{ev}(g)$ in an $\F_q$-basis of $\Imag(e_2)$ realizes
$S_Z\subseteq\F_q^{a\times m}$, and the rank weight of a codeword is the rank of its matrix because
the evaluation points form a basis, so $\mathrm{span}_{\F_q}\{g(\alpha_j)\}_j=\Imag(g)$. We assume
$r<k_Z<m$ throughout the proof until the discussion of endpoints.

\emph{Lower bound.} A nonzero $g\in\mathcal L_{<m-k_Z}$ has $\rank(g)\ge k_Z+1$
\cite{gabidulin1985}, and rank--nullity for $e_2$ on $\Imag(g)$, with
$\dim\ker(e_2)=r$, gives $\rank(e_2\circ g)\ge\rank(g)-r\ge a-(m-k_Z)+1$; the map
$g\mapsto e_2\star\mathrm{ev}(g)$ is injective there, a nonzero $g$ in its kernel having
$\rank(g)\le r\le k_Z$.

\emph{Singleton equality.} By Theorem~\ref{thm:igg-parameters}(i), $\dim_{\F_q}S_Z=(m-k_Z)m$. The
bound \eqref{eq:matrix-singleton} for $a\times m$ codes with $a\le m$ reads $\dim_{\F_q}C\le m(a-d_R+1)$, giving $d_R(S_Z)\le a-(m-k_Z)+1$. Combining this with the lower bound proves that $S_Z$ is an MRD code.

\emph{The restricted trace dual.} Identify $\Imag(e_2)^n\cong\F_q^{a\times m}$ through the same
basis. Writing $A_{\cdot j}$ for the $j$-th column, the restricted trace form is
$\langle A,B\rangle_{\Tr}=\sum_jA_{\cdot j}^{\top}G\,B_{\cdot j}=\mathrm{tr}(A^{\top}GB)$, where
$G:=\bigl(\Tr(c_ic_k)\bigr)_{i,k}$ is symmetric and invertible because the trace form is
nondegenerate on $\Imag(e_2)$ (Proposition~\ref{prop:idempotent-subspace}). By
Lemma~\ref{lem:gram-reduction},
\begin{equation}
\label{eq:delsarte-transport}
S_Z^{\perp_2}=(GS_Z)^{\pstd}.
\end{equation}
Left multiplication by the invertible $G$ is a rank-preserving row operation, so $GS_Z$ is MRD with
the same minimum rank distance. The Delsarte dual of an MRD code is MRD
\cite{gabidulin1985,delsarte1978}, and $\dim_{\F_q}S_Z^{\perp_2}=am-(m-k_Z)m=(k_Z-r)m$, so
Singleton equality for the dual gives $d_R(S_Z^{\perp_2})=a-(k_Z-r)+1=(m-k_Z)+1$. With
Theorem~\ref{thm:igg-parameters}(iii) and $n=m$ this is $d^{R}_X=n-k_Z+1$.

\emph{Endpoints.} For $k_Z=m$, $S_Z=\{0\}$ and $S_Z^{\perp_2}=\Imag(e_2)^n$ contains rank-one
vectors (Lemma~\ref{lem:rankone-span}); for $k_Z=r$ the dimension count gives
$S_Z^{\perp_2}=\{0\}$ and $K_q=0$.
\end{proof}

\begin{corollary}[Equality in the Singleton bound]
\label{cor:singleton-equality}
In the saturated square IGG family with $0\le r<k_Z<m$ and $n=m$, the branch code
$Q_{\mathrm{one}}(S_Z)$ of Theorem~\ref{thm:frozen-product}, defined on $\F_q^{a\times n}$ with
$a=m-r$, has $K=K_q=m(k_Z-r)$, $d^{R}_X=m-k_Z+1$ and $d^{R}_Z=1$, and attains
\eqref{eq:rank-singleton} with equality for every admissible $(m,r,k_Z)$, since
$an-\max(a,n)(d^{R}_X+d^{R}_Z-2)=m(m-r)-m(m-k_Z)=m(k_Z-r)$; the ambient application to the same
code, with $a=m$, is strictly weaker whenever $r>0$, giving $K_q\le mk_Z$, and equality does not hold on the full label space. The codes $\mathrm{QGab}(\boldsymbol\alpha,r,r)$ of
Proposition~\ref{prop:separation} attain \eqref{eq:rank-singleton} as well, since
$n^2-n(2r+2-2)=n(n-2r)=K_q$. Both realize Theorem~\ref{thm:singleton-equality-mrd}: the branch
code is its $v=0$ specialization, with $S_X=\{0\}$ in the branch and $S_Z$ MRD of dimension
$m(m-k_Z)$ and minimum rank distance $(m-r)-(m-k_Z)+1$ (Theorem~\ref{thm:analytic-mrd}), and the
quantum Gabidulin codes are its $u=v$ specialization, with both spaces MRD of dimension $nr$ and
minimum rank distance $n-r+1$ (Lemma~\ref{lem:consecutive-spans}). The one-sided and the two-sided
witness are therefore two specializations of one structure.
\end{corollary}

\begin{proof}
The parameter values of the branch code are established by Theorems~\ref{thm:igg-parameters} and
\ref{thm:analytic-mrd} together with Theorem~\ref{thm:frozen-product}, which identifies the branch
check spaces as $S_X\cap W^n=\{0\}$ and $S_Z$ and the branch logical count as $K_q$; those of
$\mathrm{QGab}(\boldsymbol\alpha,r,r)$ follow from Proposition~\ref{prop:separation}(i) applied in
$\F_2^{n\times n}$, whose trace Gram matrix in the self-dual normal basis is the identity by Lemma~\ref{lem:normal-basis-duals}. The displayed identities follow by direct algebraic simplification, and the final assertion follows from Theorem~\ref{thm:singleton-equality-mrd} applied to both pairs.
\end{proof}

The statements above are specific to the square aligned regime: for $n<m$ the projected dimensions
are not determined here and \eqref{eq:sx-dimension} may fail, whereas
Theorem~\ref{thm:rank-singleton} holds for all $a$ and $n$.

\begin{example}[The primary instance]
\label{ex:primary-instance}
Let $q=2$, $(m,n,k_X,k_Z,r)=(5,5,0,3,1)$ over $\F_{2^5}=\F_2[x]/(x^5+x^2+1)$ with the polynomial
basis $B=(1,x,x^2,x^3,x^4)$, evaluation points $\alpha_j=b_j$, and the trace projector
$e_1(z)=\Tr(z)\cdot1$ of Remark~\ref{rem:explicit-idempotents}. Then $a=4$ and $\delta_X=m=5>r$,
so $S_X=\Imag(e_1)^n$ by Theorem~\ref{thm:igg-parameters}(i) and the pair is saturated;
$\dim_{\F_2}S_Z=(m-k_Z)m=10$, and Theorem~\ref{thm:analytic-mrd} gives
$d_R(S_Z)=a-(m-k_Z)+1=3$ and $d^{R}_X=d_R(S_Z^{\perp_2})=m-k_Z+1=3$, whence
$\llbracket25,10,3^{R}/1^{R}\rrbracket_2$ with $K_q=m(k_Z-r)=10$ and radius $t=1$ admissible by
Theorem~\ref{thm:exact-unique}. Direct computation in this field gives $\Tr(b_{\ell})\neq0$ exactly
for $\ell\in\{1,4\}$ and $G_B^2\neq G_B$, the instance announced in
Proposition~\ref{prop:sector-split-embedding}.
\end{example}

The evaluation construction of this subsection does not extend to layouts with $n>m$, and the
obstruction is again a rank-one label.

\begin{remark}[Dependent evaluation points]
\label{rem:dependent-points}
The construction requires evaluation points that are $\F_q$-linearly independent, and therefore
does not extend to $n>m$. Indeed, let the points be $\F_q$-linearly dependent, as is forced
whenever $n>m$, choose $0\neq c\in\F_q^n$ with $\sum_jc_j\alpha_j=0$ and $0\neq\beta\in W$, and put
$E:=(c_1\beta,\dots,c_n\beta)$, of rank weight one and lying in $W^n$. Every generator of $S_Z$ has
$j$-th coordinate $h_j=e_2\bigl(g(\alpha_j)\bigr)$ with $g$ and $e_2$ both $\F_q$-linear and
$c_j\in\F_q$, so $\sum_jc_jh_j=e_2\bigl(g(0)\bigr)=0$ and
$\langle h,E\rangle_{\Tr}=\Tr\bigl(\beta\sum_jc_jh_j\bigr)=0$. Hence
$0\neq E\in\Kker=S_Z^{\ptr}\cap W^n$ with $\wt_R(E)=1$, so $d_R(\Kker)=1$ and the certifying
decoder of Definition~\ref{def:certifying-decoder} returns $\bot$ on the fault-free record for
every $t\ge1$; the mirrored statement holds in the $Z$ sector. Theorems~\ref{thm:exact-unique} and
\ref{thm:rank-singleton} are unaffected, being statements about an arbitrary matrix syndrome map
respectively an arbitrary $a\times n$ CSS pair. This does not assert that rank-metric CSS codes of
length $n>m$ are impossible; rather, it precludes this specific evaluation-based construction, for which the
dependency is inherited by every projected check.
\end{remark}
\section{Discussion and Conclusion}
\label{sec:discussion}

\subsection{Relation to prior work}
\label{subsec:prior-work}

Theorem~\ref{thm:rank-singleton} is the rank-metric form of the asymmetric quantum Singleton bound
for Hamming-metric CSS codes \cite{sarvepalli2009asymmetric}, specializing to
\eqref{eq:matrix-singleton} \cite{gabidulin1985,delsarte1978} when one sector is trivial. What the
rank metric requires is a shortening on row supports rather than on coordinate positions, so that
the shortened quantity is the rank of a matrix, together with the equality analysis
\eqref{eq:singleton-equality} it makes available and the transposition step, which has no
Hamming-metric counterpart. Row shortening alone yields \eqref{eq:row-erasure}, the rank-metric
specialization of the entropic quantum Singleton argument \cite{cerf1997,grassl2022entropic}, and
carrier shortening yields \eqref{eq:col-erasure}.

In the Hamming metric the correspondence between optimal asymmetric CSS codes and nested MDS pairs
is classical \cite{sarvepalli2009asymmetric}, with the parameters classified conditionally on the
MDS conjecture \cite{ezerman2013aqmds}; the rank-metric statements here are unconditional, and
purity is a conclusion rather than a hypothesis. The non-invariance of nontrivial MRD codes under projectors is similarly well understood \cite{lunardon2018kernels,csajbok2020idealizers}; Lemmas~\ref{lem:mrd-fullrank}--\ref{lem:mrd-generation} provide a self-contained treatment that deduces the required properties directly from \eqref{eq:matrix-singleton} and Delsarte duality. The two relative distances of
Definition~\ref{def:relative-distances} are first relative generalized matrix weights in the sense
of \cite{martinezpenas2018rgmw}; our primary contribution lies in identifying the two nested pairs
that a projector and the quantum syndrome semantics produce, and providing the exact evaluation of both on
every layout (Theorem~\ref{thm:radius-dichotomy}).

Among quantum rank-metric constructions, \cite{delfosse2024stacked} gives the CSS family used here
as the two-sided witness, identifying a fault-tolerant extraction circuit and an efficient decoder as essential
prerequisites; its minimum rank distance is bounded below by $r+1$ there, and the same value is
reported in \cite{nizuka2026}, whereas Proposition~\ref{prop:separation}(i) establishes the
sector-by-sector equality that Corollary~\ref{cor:optimal-rigidity} and
Proposition~\ref{prop:qgab-decoder} require. We address the code-capacity part of the second
prerequisite for the parent family at the optimal unique-decoding radius; the code measured after
the Clifford layer of the stacked protocol is a row-wise relabeling of that family, so rank
weights, and with them the decoding radius, are unchanged, but faulty extraction is not treated
here. The decoder is a reduction to Gabidulin syndrome decoding with certification, not a new rank-decoding algorithm; some algebraic structure is unavoidable, as minimum-rank syndrome decoding is NP-complete
\cite{bussfrandsen1999}, generic $\Fqm$-linear rank-metric decoding would place NP in ZPP if solvable
in polynomial time \cite{gaboritzemor2016}, and list-decoding sizes are exponential beyond half the
minimum distance in general \cite{wachterzeh2013,raviv2016}. The codes of \cite{nizuka2026} lie outside the CSS class and
supply the achievability that turns Corollary~\ref{cor:css-penalty-fixed} into a separation at
fixed resources. Projecting a Gabidulin code onto an $\F_q$-subspace is standard
\cite{ndiaye2023subspace}; the contribution in
Section~\ref{subsec:igg} is that one projector, dictated by trace self-adjointness, simultaneously
makes the projected checks commute with the
complementary sector, admits a post-processing-free projected syndrome, and realizes the restricted
trace dual as the logical quotient.

\subsection{Intended regime and limitations}
\label{subsec:limitations}

The intended regime is one in which the dominant physical faults are correlated and aligned along intra-carrier directions; we do not assert uniform superiority over standard baseline codes across arbitrary noise environments. In strictly local layouts with weak correlations, local codes remain the natural choice \cite{fowler2012}. A hybrid arrangement follows, in which an inner code suppresses local
stochastic noise while an outer layer targets the residual correlated component. Three limitations
remain. First, the one-sided witness of Section~\ref{subsec:igg} restricts its primary applicability to strongly biased settings, since Theorem~\ref{thm:no-go} forces $d^{R}_Z=1$ there. Second, syndrome extraction in the presence of measurement faults is not analyzed in this work, nor do we evaluate the physical overhead of
measuring high-weight check operators or the circuit-level implications of the one-sidedness established in
Theorem~\ref{thm:no-go}. Third, while
Theorem~\ref{thm:optimal-existence} realizes every admissible Singleton-optimal parameter triple,
it does not supply an efficient decoder for an arbitrary optimal MRD pair: the decoder of
Section~\ref{subsec:qgab} is specific to the quantum Gabidulin family, and the fixed
evaluation-and-projection construction of Section~\ref{subsec:igg} does not extend to the regime
in which the evaluation points become dependent (Remark~\ref{rem:dependent-points}).

\subsection{Finite verifications and reproducibility}
\label{subsec:verification}

The statements above are accompanied by finite computations, all exhaustive at the stated
parameters unless marked as sampled: the enumeration of the $92\,881$ trace-orthogonal pairs with
$K_q>0$ in $\F_2^{3\times2}$, of which $97$ attain \eqref{eq:rank-singleton} and all satisfy
Theorem~\ref{thm:singleton-equality-mrd} and Corollary~\ref{cor:purity}; the two instances of
\cite{nizuka2026} and the $3\times2$ non-CSS code of Proposition~\ref{prop:css-specific}, with the
$168$ elements of $\mathrm{GL}_3(\F_2)$ tested there; the relative distances of
$\mathrm{QGab}(\boldsymbol\alpha,2,2)$ over $\F_{32}$ at three idempotents and at randomly drawn
non-idempotent maps; the $882\,912$ idempotents of rank at most three tested against the $6\times3$ pair of
Proposition~\ref{prop:stacking-sharpness}, of which $72$ are invariant
(Remark~\ref{rem:dichotomy}); the decoder of
Proposition~\ref{prop:qgab-decoder} on all $2^{10}$ records per sector at $n=5$, and on $300$ labels
per sector at $n=9$, $t=2$; the construction of Theorem~\ref{thm:optimal-existence} on
the layouts $3\times4$, $3\times5$, $4\times5$, $3\times6$ and their transposes, at $(u,v)=(1,1)$
and $(2,1)$, each instance being checked for equality in \eqref{eq:rank-singleton}, for the
structure of Theorem~\ref{thm:singleton-equality-mrd}, for purity, and for the two relative
distances of Theorem~\ref{thm:radius-dichotomy} at randomly drawn maps; the erasure optimization
of Proposition~\ref{prop:erasure-bounds}, enumerated against the closed form of
\eqref{eq:erasure-general} at all $472$ feasible triples with $a\le16$, $n\le8$ and $D\le9$; and
the stacking of Lemma~\ref{lem:row-stacking} on the $8\times2$ layout, where $K=8$ and no label of
stacked rank one lies in $N(S)\setminus S$. The scripts that produce these
numbers, the field and basis specifications they use, and the expected outputs are provided as
supplementary material; they are self-contained and depend on no external library. All numerical values reported throughout the text were generated and verified by these scripts.

\subsection{Conclusion and open problems}
\label{subsec:open}

Optimality in the asymmetric rank-metric Singleton bound is a structural property: it forces both
check spaces to be MRD, it is attainable at every admissible parameter triple, it forces purity,
and it fixes exactly how far a projector-based partial recovery can go. This resolves the central question posed in Section~\ref{sec:introduction}: relaxing the recovery target from the
full error to the projected error modulo stabilizer labels fails to enlarge the worst-case radius on any
layout, whereas relaxing it further to the projected syndrome either removes the radius restriction
altogether or leaves it unchanged, according to a single invariance condition, which at a
nontrivial idempotent requires $a>n$ or $d^{R}_X=1$. Corollary~\ref{cor:css-penalty-fixed} settles
the attainability of the column bound \eqref{eq:col-erasure} for every odd distance on the layouts
$a=2\ell n$; four questions remain open: whether it is attained at even distances, where the parity
term $\iota=1$ makes the two erasure bounds differ, and on layouts outside that family; a
classification of the invariant case of
Theorem~\ref{thm:radius-dichotomy} for $a>n$, of which
Proposition~\ref{prop:stacking-sharpness} is one instance; what replaces Theorem~\ref{thm:radius-dichotomy} for codes that do not attain
\eqref{eq:rank-singleton}, given that Example~\ref{ex:nonoptimal-pair} rules out any formula in
the ordinary parameters; and the dimension-sensitive analogue of
Theorem~\ref{thm:rank-singleton} in the sum-rank metric \cite{byrne2021sumrank}. Fault-tolerant syndrome extraction for the stacked architecture of \cite{delfosse2024stacked} remains an essential
open problem.

\appendices

\section{Realization over a Nonprime Base Field}
\label{app:nonprime}

This appendix supplies the two ingredients that Proposition~\ref{prop:sector-split-embedding} and
Lemma~\ref{lem:physical-normalizer} defer for $q=p^{\nu}$ with $\nu>1$: the measurement protocol
recovering the $\F_q$-valued syndrome, and the group accounting behind \eqref{eq:logical-count}.
Both rest on the fact that
\begin{equation}
\label{eq:absolute-trace-nondeg}
(\alpha,\beta)\mapsto\Tr_{\F_q/\F_p}(\alpha\beta)\ \text{is nondegenerate on }\F_q.
\end{equation}

\begin{proposition}[Recovery of the $\F_q$-valued syndrome over a nonprime base field]
\label{prop:fq-syndrome}
Let $q=p^{\nu}$ with $p$ prime, let $\{\lambda_1,\dots,\lambda_{\nu}\}$ be an $\F_p$-basis of
$\F_q$, and adopt the embedding of Proposition~\ref{prop:sector-split-embedding}. For a check label
$h\in\Fqm^n$ and an error label $E$ in the complementary sector, measuring $\lambda_ih$ returns
$\Tr_{\F_q/\F_p}(\lambda_i\,s(h;E))$, and the $\nu$ outcomes determine $s(h;E)\in\F_q$ through a
fixed $\F_p$-linear map computed offline. The $\F_q$-valued trace syndrome is therefore measurable
for every prime power $q$, at a $\nu$-fold increase in measured generators and with no change to
the stabilizer group, since $\lambda_ih\in S_H$.
\end{proposition}

\begin{proof}
$s(\lambda_ih;E)=\lambda_i\,s(h;E)$ by $\F_q$-linearity, and the physical phase is
$\Tr_{\F_q/\F_p}\bigl(s(\lambda_ih;E)\bigr)$ by Proposition~\ref{prop:sector-split-embedding}(ii). By
\eqref{eq:absolute-trace-nondeg}, $\alpha\mapsto(\Tr_{\F_q/\F_p}(\lambda_i\alpha))_i$ is an
$\F_p$-isomorphism $\F_q\to\F_p^{\nu}$; and $\lambda_ih\in S_H$, so only the measured generating
set is enlarged.
\end{proof}

\begin{lemma}[$\F_p$-accounting of the realized CSS group]
\label{lem:fp-accounting}
Let $q=p^{\nu}$ with $p$ prime, let $S_X,S_Z\subseteq\Fqm^n$ be trace-orthogonal $\F_q$-linear check
spans with $\F_q$-bases $(x_1,\dots,x_{d_X})$ and $(z_1,\dots,z_{d_Z})$, let
$\lambda_1,\dots,\lambda_{\nu}$ be an $\F_p$-basis of $\F_q$, and adopt the embedding of
Proposition~\ref{prop:sector-split-embedding}. Then:
(i) $\{\lambda_jx_i\}_{i,j}$ and $\{\lambda_jz_i\}_{i,j}$ are $\F_p$-bases of $S_X$ and $S_Z$, so
the enlarged measured set of Proposition~\ref{prop:fq-syndrome} consists of $\nu(d_X+d_Z)$
generators of the same group; (ii) the realized group $G$ generated by
$\{X^{\varphi_X(x)}\}\cup\{Z^{\varphi_Z(z)}\}$ is scalar-free and
$(x,z)\mapsto X^{\varphi_X(x)}Z^{\varphi_Z(z)}$ is an isomorphism $(S_X\times S_Z,+)\to G$, so $G$
is elementary abelian of order $q^{\,d_X+d_Z}$; (iii) the code space has dimension $q^{K_q}$ with
$K_q=N-d_X-d_Z$ and $N=nm$, which is \eqref{eq:logical-count}.
\end{lemma}

\begin{proof}
(i) A vanishing combination $\sum_{i,j}c_{ij}\lambda_jx_i=0$ reads
$\sum_i(\sum_jc_{ij}\lambda_j)x_i=0$, and independence of the $x_i$ over $\F_q$ and of the
$\lambda_j$ over $\F_p$ forces $c_{ij}=0$; the set has $\nu d_X=\dim_{\F_p}S_X$ elements, so it forms
a basis; an identical argument applies to $S_Z$. (ii) Generators commute by trace orthogonality and
Proposition~\ref{prop:sector-split-embedding}(ii), so reordering a word into the normal form
$X^{\varphi_X(x)}Z^{\varphi_Z(z)}$ produces only trivial phases; the same computation gives
multiplicativity in $(x,z)$, and injectivity follows from that of the physical realization on labels. Every element has order dividing $p$: for odd $p$, $(X^{a}Z^{b})^{p}=\omega_p^{\binom{p}{2}\Tr_{\F_q/\F_p}(a\cdot b)}I$ with $\binom{p}{2}\equiv0$, and for $p=2$, $(X^{a}Z^{b})^{2}=(-1)^{\Tr_{\F_q/\F_2}(a\cdot b)}I$ with
$a\cdot b=0$ by \eqref{eq:exponent-identity}. (iii) A scalar-free abelian subgroup of order
$p^{\varsigma}$ of the generalized Pauli group on $p^{\nu N}$ dimensions fixes a subspace of
dimension $p^{\nu N-\varsigma}$ \cite{ashikhmin2001,ketkar2006,gottesman1997}; with
$\varsigma=\nu(d_X+d_Z)$ that is $q^{N-d_X-d_Z}$.
\end{proof}

\section{A Singleton-Optimal Pair Carrying an Invariant Projector}
\label{app:sharpness}

This appendix exhibits a Singleton-optimal pair on a tall layout carrying an invariant projector
(Proposition~\ref{prop:stacking-sharpness}), the witness of Remark~\ref{rem:dichotomy}.

The construction stacks a square Singleton-optimal pair on disjoint row blocks. The resulting
tall pair is again Singleton-optimal and is invariant under each block projector, which is what
Lemma~\ref{lem:mrd-noninvariance} excludes on layouts with $a\le n$.

The next proposition shows that the hypothesis $a\le n$ of
Lemma~\ref{lem:mrd-noninvariance}, and hence of
Corollary~\ref{cor:optimal-rigidity}, cannot be dropped either.

\begin{proposition}[Singleton-optimal pairs with an invariant projector when $a>n$]
\label{prop:stacking-sharpness}
Let $(C_X,C_Z)\subseteq\F_q^{n\times n}$ be a square pair attaining equality in
\eqref{eq:rank-singleton} with sector distances $d^{R}_X,d^{R}_Z$ and logical dimension
$K_0=n^2-n(d^{R}_X+d^{R}_Z-2)>0$, and let $\ell\ge2$. Define
$S_X:=\{[X_1;\dots;X_{\ell}]:X_i\in C_X\}\subseteq\F_q^{\ell n\times n}$ by vertical stacking and
$S_Z$ likewise. Then $(S_X,S_Z)$ is a CSS pair on the $\ell n\times n$ layout with $K=\ell K_0$ and
the same sector distances, it attains equality in \eqref{eq:rank-singleton}, and both spans are
invariant under each block projector $P_i=\mathrm{diag}(0,\dots,I_n,\dots,0)$, a symmetric
idempotent of rank $n$. Consequently, at $P_i$ both sectors admit an exact $\F_q$-linear ambient
projected-syndrome interface, so $\mathrm{I}_1$, $\mathrm{I}_2$ and $\mathrm{I}_5(t)$ hold for
every $t\ge1$ even though $2t\ge d^{R}_X$ for $t\ge\lceil d^{R}_X/2\rceil$, while $\mathrm{I}_3$
fails because the branch $\Imag(P_i)^n$ carries $K_0>0$ logical qudits
(Theorem~\ref{thm:quotient-characterization}).
\end{proposition}

\begin{proof}
The standard pairing splits over row blocks, so $S_Z\subseteq S_X^{\pstd}$ and $S_Z^{\pstd}$ is the
blockwise direct sum of copies of $C_Z^{\pstd}$; hence $\dim S_X=\ell\dim C_X$,
$\dim S_Z=\ell\dim C_Z$ and $K=\ell n^2-\ell(\dim C_X+\dim C_Z)=\ell K_0$.

For the distances, the argument of Theorem~\ref{thm:invariant-splitting} applies verbatim: by
\eqref{eq:row-submatrix-bound} the rank of a block matrix is at least the rank of each block, a
logical class is nontrivial exactly when some block component is, and a minimum-rank nontrivial
representative of one block, extended by zero, attains the minimum. Hence the sector distances of
$(S_X,S_Z)$ are $d^{R}_X$ and $d^{R}_Z$, and
\[
\begin{aligned}
&\ell n\cdot n-\max(\ell n,n)\bigl(d^{R}_X+d^{R}_Z-2\bigr)\\
&\qquad=\ell n^2-\ell n\bigl(d^{R}_X+d^{R}_Z-2\bigr)=\ell K_0=K ,
\end{aligned}
\]
so equality holds. The invariance under $P_i$ is an immediate consequence of the block structure, and the interface statements follow in the standard pairing for every $q$. Let $H$ be a basis of $S_Z$ whose elements are
each supported on a single row block, so that $\sigma_H(E)=\bigl(s_1(E),\dots,s_{\ell}(E)\bigr)$ with
$s_j$ depending only on the $j$-th block of $E$; since $P_iE$ has $i$-th block equal to that of $E$ and all
other blocks zero, $\sigma_H(P_iE)=(0,\dots,s_i(E),\dots,0)$ is an $\F_q$-linear function of
$\sigma_H(E)$ on the whole label space, which establishes condition $\mathrm{I}_1$, and consequently $\mathrm{I}_2$ and $\mathrm{I}_5(t)$ for every $t$. For $\mathrm{I}_3$, choose $L_0\in C_Z^{\pstd}\setminus C_X$,
which exists because $K_0>0$, and let $E$ carry $L_0$ in block $i$ and zero elsewhere; then
$\sigma_H(E)=0=\sigma_H(0)$ while $P_iE=E\notin S_X$, so no function of the syndrome returns
$P_iE+S_X$. Where a self-dual basis exists, for $q$ even and, for $q$ odd, whenever $\ell n$ is odd
\cite{seroussi1980}, realizing the labels in $\F_{q^{\ell n}}^{\,n}$ through it makes the trace
Gram matrix the identity, so $P_i$ is trace-self-adjoint and the same conclusions follow directly from
criterion~(d) of Theorem~\ref{thm:projected-recoverability}, criterion~(b) of
Theorem~\ref{thm:rank-restricted}, and Theorem~\ref{thm:quotient-characterization} with
$K_{\Imag(P_i)}=K_0>0$.
\end{proof}

Realized in $\F_{q^{\ell n}}^{\,n}$ as at the end of the proof, the stacked pair has, by
Theorem~\ref{thm:radius-dichotomy},
$\delta_{\mathrm{quot}}(P_i)=d^{R}_X$ and $\delta_{\mathrm{syn}}(P_i)=+\infty$ at every block
projector, whatever the base pair; so it separates $\mathrm{I}_5$ from $\mathrm{I}_6$: for two
copies of $\mathrm{QGab}(\boldsymbol\alpha,1,1)$ over $\F_8$ on the $6\times3$ layout, the instance
quoted in Remark~\ref{rem:dichotomy}, these values are $2$ and $+\infty$, while two copies of
$\mathrm{QGab}(\boldsymbol\alpha,2,2)$ over $\F_{32}$ give the $10\times5$ layout with $K=10$,
$d^{R}_X=d^{R}_Z=3$ and $\delta_{\mathrm{quot}}(P_i)=3$.

\bibliographystyle{IEEEtran}
\bibliography{references}

\end{document}